\documentclass{lmcs}
\pdfoutput=1

\usepackage{lastpage}
\lmcsdoi{21}{3}{4}
\lmcsheading{}{\pageref{LastPage}}{}{}%
{Aug.~01,~2023}{Jul.~15,~2025}{}


\newtheoremstyle{convention}
  {6pt}
  {6pt}
  {\normalfont}
  {}
  {\bfseries}
  {{\bfseries .}}
  {5pt plus 1pt minus 1pt}
  {\thmname{#1}}
\theoremstyle{convention}
\newtheorem{convention}{Notational convention}

\usepackage{hyperref}
\hypersetup{
  citecolor=blue,
  colorlinks=true,
  pdftitle={A Curry-Howard Correspondence for Linear, Reversible Computation}
}

\theoremstyle{plain}

\usepackage{amsmath,amssymb}
\usepackage{amsfonts}
\usepackage[utf8]{inputenc}
\usepackage[T1]{fontenc}

\usepackage{anyfontsize}

\usepackage{tikz}
\usetikzlibrary{calc}
\usetikzlibrary{arrows.meta}
\usetikzlibrary{shapes, patterns, decorations.pathreplacing, decorations.text}
\usetikzlibrary{trees}
\usetikzlibrary{positioning}
\usetikzlibrary{tikzmark, cd}

\tikzset{
  node rotated/.style = {rotate=180},
  border rotated/.style = {shape border rotate=180},
  downtriangle/.style = {fill=white, draw=black, regular polygon, regular polygon sides=3, border rotated},
  triangle/.style = {fill=white, draw=black, regular polygon, regular polygon sides=3}
}

\usepackage{longtable}
\usepackage{proof}
\usepackage{mathpartir}
\usepackage{ebproof}
\usepackage{dsfont}
\usepackage{amsfonts}
\usepackage{stmaryrd}
\usepackage{mathtools}
\usepackage{cmll}
\usepackage{scalerel}

\usepackage{pdflscape}

\usepackage{graphicx}
\usepackage{color}
\usepackage{wrapfig}
\usepackage{enumitem}

\usepackage{longtable}
\usepackage{multirow}
\usepackage{nicematrix}
\usepackage{xspace}

\newcommand{\set}[1]{\ensuremath{\{#1\}}}

\newcommand{\RPP}[1]{\ensuremath{\operatorname{RPP}^{#1}}}

\newcommand{\mumall}{\ensuremath{\mu\textsf{MALL}}\xspace}
\newcommand{\alt}{~\mid~}

\newcommand{\iso}{\leftrightarrow}

\newcommand{\inl}[1]{\ensuremath{\mathtt{inj}_l{\;#1}}}

\newcommand{\fold}[1]{\ensuremath{\mathtt{fold}{\;#1}}}

\newcommand{\inr}[1]{\ensuremath{\mathtt{inj}_r{\;#1}}}

\newcommand{\pv}[2]{\langle #1,#2 \rangle}

\newcommand{\letv}[3]{{\tt let}\,{#1}={#2}~{\tt in}~{#3}}

\newcommand{\isoterm}{\omega}

\newcommand{\PE}{\ensuremath{\mathtt{OD}}}

\newcommand{\FV}{{\rm FV}}

\newcommand{\entailiso}{\vdash_{i}~}

\newcommand{\tensor}{\otimes}

\newcommand{\cut}{\ensuremath{\mathtt{cut}}}
\newcommand{\id}{\ensuremath{\mathtt{id}}}

\newcommand{\lett}{\ensuremath{\mathtt{let}}}
\newcommand{\expLet}{\ensuremath{\to_{\mathtt{elet}}}}
\newcommand{\letsig}[2]{{\mathtt{let}}\,{#1}~{\mathtt{in}}~{#2}}
\newcommand{\Val}[1]{{\ensuremath{\mathit{Val}(#1)}}}

\newcommand{\isobasique}{\ensuremath{\{v_1~\iso~e_1 \alt \dots \alt v_n~\iso~e_n\}}}

\newcommand{\Pos}{\ensuremath{\mathtt{Pos}}}
\newcommand{\myPos}[1]{\ensuremath{\mathtt{Pos}}\left(\begin{array}{c}#1\end{array}\right)}
\newcommand{\Neg}{\ensuremath{\mathtt{Neg}}}

\newcommand{\Foc}[1]{\ensuremath{\mathtt{Neg(}#1\mathtt{)}}}

\newcommand{\entaile}{\vdash_e}

\newcommand{\one}{\ensuremath{\mathds{1}}}
\DeclareMathAlphabet{\mymathbb}{U}{BOONDOX-ds}{m}{n}
\newcommand{\zero}{\ensuremath{\mymathbb{0}}}
\newcommand{\expB}{\ensuremath{\to_{e\beta}}}
\newcommand{\cute}{\ensuremath{\rightsquigarrow}}

\newcommand{\fix}{\ensuremath{\mathbf{fix}~}}

\newcommand{\mynl}{\\[1ex]}

\newcommand{\up}{\ensuremath{\uparrow}}

\newcommand{\natT}{\ensuremath{\mathbb{N}}}
\newcommand{\intT}{\ensuremath{\mathbb{Z}}}

\newcommand{\circu}{\ensuremath{\mathtt{Circ}}}
\newcommand{\unfold}{\ensuremath{\mathcal{U}}}
\newcommand{\premiss}[1]{\ensuremath{\operatorname{premiss}(#1)}}
\newcommand{\opn}[1]{\ensuremath{\operatorname{#1}}}

\newcommand{\ov}[1]{\ensuremath{\overline{#1}}}

\newcommand{\TypCtxta}{\Theta}
\newcommand{\TypCtxtb}{\Sigma}

\usepackage{yfonts}
\newcommand\goth[1]{\textgoth{#1}}
\usepackage{regexpatch}

\usepackage[normalem]{ulem}
\usepackage{soul}

\allowdisplaybreaks

\newcommand{\mkrule}[1]{{
    \text{\scriptsize
    \ensuremath{
      \mathsf
      {#1}
      }}}
}

\newcommand\rax{\mkrule{id}}
\newcommand\rcut{\mkrule{cut}}
\newcommand\rbot{\mkrule{\bot}}
\newcommand\rtop{\mkrule{\top}}

\newcommand\rmu{\mkrule{\mu}}
\newcommand\rnu{\mkrule{\nu}}
\newcommand\rwith{\mkrule{\with}}
\newcommand\rparr{\mkrule{\parr}}
\newcommand\rotimes{\mkrule{\otimes}}
\newcommand\rtensor{\mkrule{\otimes}}

\newcommand\rone{\mkrule{\one}}
\newcommand\roplusi{\mkrule{\oplus^i}}
\newcommand\rex{\mkrule{ex}}

\def\tinfex{\begin{prooftree}
\hypo{\vdash \Gamma, G\tikzmark{tex2}, F, \Delta}
\infer{1}[\rex]{\vdash \Gamma, F\tikzmark{tex1}, G, \Delta}
\end{prooftree}}

\def\tinfparr{\begin{prooftree}
\hypo{\vdash F,\tikzmark{tpa2} G, \Gamma}
\infer{1}[\rparr]{\vdash F\parr\tikzmark{tpa1} G, \Gamma}
\end{prooftree}}

\def\tinftensor{\begin{prooftree}
\hypo{\vdash F\tikzmark{tta2},\Gamma}
\hypo{\vdash G\tikzmark{tta3},\Delta}
\infer{2}[\rotimes]{\vdash F\tensor G, \tikzmark{tta1}\Gamma,\Delta}
\end{prooftree}}

\def\tinfplus{\begin{prooftree}
\hypo{\vdash F_i\tikzmark{tpla2}, \Gamma}
\infer{1}[$\roplusi,~ i\in\{1, 2\}$]{\vdash F_1\oplus\tikzmark{tpla1} F_2, \Gamma}
\end{prooftree}}

\def\tinfwith{\begin{prooftree}
\hypo{\vdash F\tikzmark{twa2},\Gamma}
\hypo{\vdash G\tikzmark{twa3},\Gamma}
\infer{2}[\rwith]{\vdash F\with G, \tikzmark{twa1}\Gamma}
\end{prooftree}}

\def\tinfbot{\begin{prooftree}
\hypo{\vdash \Gamma\tikzmark{tbotb2}}
\infer{1}[\rbot]{\vdash \bot, \Gamma\tikzmark{tbotb1}}
\end{prooftree}}

\def\tinftop{\begin{prooftree}
\hypo{}
\infer{1}[\rtop]{\vdash \top, \Gamma}
\end{prooftree}}

\def\tinfone{\begin{prooftree}
\hypo{}
\infer{1}[\rone]{\vdash \one}
\end{prooftree}}
\def\tinfmu{\begin{prooftree}
\hypo{ \vdash F[X \leftarrow\mu X.F]\tikzmark{tma2}, \Gamma}
\infer{1}[\rmu]{ \vdash \mu X.\tikzmark{tma1}F, \Gamma}
\end{prooftree}}

\def\tinfnu{\begin{prooftree}
\hypo{\vdash G[X\leftarrow \nu X.G]\tikzmark{tna2}, \Gamma}
\infer{1}[\rnu]{\vdash \nu X.\tikzmark{tna1}G, \Gamma}
\end{prooftree}}

\def\tinfax{
\begin{prooftree}
\infer{0}[\rax]{\vdash F, F^\perp}
\end{prooftree}}
\def\tinfcut{
  \begin{prooftree}
  \hypo{\vdash \Gamma,\tikzmark{tca2} F}
  \hypo{\vdash F^\perp,\tikzmark{tca3} \Delta}
\infer{2}[{\rcut}]{\vdash \Gamma,\tikzmark{tca1} \Delta}
\end{prooftree}}

\theoremstyle{defC}
\newtheorem{exaC}[thm]{Example}

\begin{document}

\title[A Curry-Howard Correspondence for Linear, Reversible Computation]{A
Curry-Howard Correspondence for \texorpdfstring{\protect\\}{} Linear, Reversible
Computation}
\titlecomment{This is an extended version
of the proceeding paper published in LIPICS~\cite{klin2023lipics,
chardonnet2022curry}. It contains all the missing proofs, together with
additional discussions.}
\thanks{We would like to warmly thank the reviewers for their accurate comments and their detailed reading of the paper: their comments contributed to greatly improve our paper.}

\thanks{This work has been partially funded by the French National Research
Agency (ANR) with the projects RECIPROG ANR-21-CE48-0019, PPS
ANR-19-CE48-0014, TaQC ANR-22-CE47-0012 and within the framework of
``Plan France 2030'', under the research projects EPIQ
ANR-22-PETQ-0007, OQULUS ANR-23-PETQ-0013, HQI-Acquisition
ANR-22-PNCQ-0001 and HQI-R\&D ANR-22-PNCQ-0002.}

\thanks{This work has also been partially funded by MIUR FARE project CAFFEINE, ``Compositional and Effectful Program Distances'', R20LW7EJ7L}

\author[K.~Chardonnet]{Kostia Chardonnet\lmcsorcid{0009-0000-0671-6390}}[a, d]
\author[A.~Saurin]{Alexis Saurin\lmcsorcid{0009-0002-1304-5518}}[b]
\author[B.~Valiron]{Benoît Valiron\lmcsorcid{0000-0002-1008-5605}}[c]

\address{Department of Computer Science and Engineering, University of Bologna, Italy}
\email{kostia.chardonnet@pm.me}

\address{IRIF, Universit\'e Paris Cité, CNRS \and \'Equipe INRIA
Picube}
\email{alexis.saurin@irif.fr}

\address{Université Paris-Saclay, CNRS,
CentraleSupélec, ENS Paris-Saclay, Inria, LMF, 91190, Gif-sur-Yvette, France}
\email{benoit.valiron@universite-paris-saclay.fr}

\address{Université de Lorraine, CNRS, Inria, LORIA, F-54000 Nancy, France}

\begin{abstract}
  In this paper, we present a linear and reversible programming language with
  inductive types and recursion. The semantics of the language is based on
  pattern-matching; we show how ensuring syntactical exhaustivity and
  non-overlapping of clauses is enough to ensure reversibility. The language
  allows to represent any Primitive Recursive Function. We then provide a
  Curry-Howard correspondence with the logic {\mumall}: linear logic extended
  with least and greatest fixed points allowing (co)inductive statements.
  The critical part of
  our work is to show how primitive recursion yields circular proofs that
  satisfy {\mumall} validity criterion and how the language simulates the
  cut-elimination procedure of {\mumall}.
\end{abstract}

\maketitle

\section{Introduction}\label{section:intro}

Computation and logic are two faces of the same coin.  For instance,
consider a proof~$s$ of $A\rightarrow B$ and a proof $t$ of $A$. With
the logical rule \textit{Modus Ponens} one can construct a proof of
$B$:~\autoref{fig:modus} features a graphical presentation of the
corresponding proof. Horizontal lines stand for deduction steps
---they separate conclusions (below) and hypotheses (above). These
deduction steps can be stacked vertically up to axioms in order to
describe complete proofs. In~\autoref{fig:modus} the proofs of $A$ and
$A\to B$ are symbolized with vertical ellipses. The ellipsis annotated
with $s$ indicates that $s$ is a complete proof of $A\to B$ while $t$
stands for a complete proof of $A$. From the perspective of a
programmer, Figure~\ref{fig:modus} can however also be interpreted as
the application of a function of type $A\to B$ to an argument for type
$A$.

\begin{figure}[h]
  \centering
  $\infer{
      B
    }{
      \infer*{A\rightarrow B}{s}
      &
      \infer*{A}{t}
    }
  $
  \caption{Modus Ponens}
  \label{fig:modus}
\end{figure}

This connection is known as the \emph{Curry-Howard
  correspondence}~\cite{curry1934functionality,howard1980formulae}.
In this general framework, types correspond to formulas and programs to
proofs, while program evaluation is mirrored with proof
simplification (the so-called cut-elimination).
The Curry-Howard correspondence formalizes the fact that the proof $s$
of $A\to B$ can be regarded as a \emph{function} ---parametrized by an
argument of type $A$--- that produces a proof of~$B$ whenever it is
fed with a proof of~$A$. Therefore, the computational interpretation
of Modus Ponens corresponds to the \emph{application} of a function
(i.e. $s$) of type $A\to B$ to an argument (i.e. $t$) of type $A$.
When computing the corresponding program, one substitutes the
parameter of the function with $t$ and get a result of type $B$. On
the logical side, this corresponds to substituting every axiom
introducing $A$ in the proof $s$ with the full proof $t$ of $A$. This
yields a direct proof of $B$ without any invocation of the ``lemma''
$A$.

Paving the way toward the verification of critical software, the
Curry-Howard correspondence provides a versatile framework. It has
been used to mirror first and second-order logics with dependent-type
systems~\cite{coqart,leroy2009compcert}, separation logics with
mem\-o\-ry-aware type
systems~\cite{reynolds02separation,jung2018rustbelt},
resource-sensitive logics with differential
privacy~\cite{gaboardu2013differential-privacy}, logics with monads
with reasoning on
side-effects~\cite{swamy2016dependent,maillard2020relational},
\textit{etc}. One aspect that has not yet been covered by the
Curry-Howard correspondence is the realm of \emph{reversible
  computation}.

\medskip
The idea of reversible computation comes from Landauer and
Bennett~\cite{Landauer61,bennett1973logical}, with the analysis of the
link between the erasure of information and the dissipation of energy
as heat~\cite{Landauer61, berut2012experimental}.
A reversible process $P$ can always be \emph{reversed}, in the sense
that there should be a process $P'$ for which $P$ followed by $P'$ and
$P'$ followed by $P$ both return the system to the initial state.
Reversible computation has been described in many ways, for instance
with the use of transition systems ensuring both \emph{forward} and
\emph{backward} determinism, as in the case of reversible Turing
Machines~\cite{rtm1}, but also with the use of reversible
programming languages, both imperative and
functional~\cite{GLUCK2023113429,
yokoyama2011reversible,thomsen2015interpretation,
james2014theseus,sabry2018symmetric,JacobsenKT18} and their
semantics~\cite{Chardonnet_2021, Kaarsgaard_2021}.

In a (typed) reversible programming language, a function describes a
bijection between the domain and the co-domain of the function:
composing the function with its inverse yields an identity map. This
connects to the notion of \emph{type isomorphisms}, where we identify
types and formulas. One then looks at when two types are ``the same''
according to the structure of some logic. For instance, in
intuitionistic logic, there exists an isomorphism between the formulas
$A\times B$ and $B\times A$, meaning that there is a pair of proofs
$A\times B \vdash B\times A$ and $B\times A \vdash A \times B$ which,
when \emph{cut} together, reduces to the axiom rule on either
$A \times B$ or $B \times A$, depending on the way the proofs were
cut together.

Type isomorphisms have been studied in several contexts, such as
finding suitable equations in various logical systems in order to
characterize all isomorphisms between two
types~\cite{bruce1992provable, bruce1985provable,
soloviev1997decision}, and in higher-order type systems using game
semantics~\cite{de2008second}. On a practical side, the notion of type
isomorphism has been at the root of code reuse by searching through
libraries~\cite{rittri1993retrieving}. Type isomorphisms have also
been studied in the context of linear logic, for the multiplicative
fragment~\cite{balat1999linear} and more recently for the
multiplicative additive fragment~\cite{laurent-mall-isos-fscd23}.

\medskip
\emph{Although natural, the relationship between bijective functions and
type isomorphisms has never been studied in details: this paper aims
at initiating such an analysis.}

\medskip
On the language side, we base our study on the approach presented
in~\cite{sabry2018symmetric}. In this model, reversible computation is
restricted to two main type constructs: the tensor, written
$A\tensor B$ and the co-product, written $A\oplus B$. The former
corresponds to the type of all pairs of el\-e\-ments of type $A$ and
elements of type $B$, while the latter represents the disjoint union
of all elements of type $A$ and elements of type $B$. For instance, a
bit can be typed with $\one\oplus\one$, where $\one$ is a type with
only one element. For expressivity, the type system is extended with
the inductive type of lists, defining $[A]$ with
$\one\oplus(A\otimes[A])$.  The language in~\cite{sabry2018symmetric}
offers the possibility to encode isomorphisms ---reversible maps--- with
pattern matching, and features a term construction for building
fixed points. The paper~\cite{sabry2018symmetric} discusses how
terminating functions indeed describe (total) bijections.

Although the type system hints at multiplicative additive linear logic
(MALL), the connection has not formally been done. In this paper, we
propose a correspondence with the logic \mumall~\cite{baelde07least,
  baelde12least, baelde2016infinitary, these-doumane}: an extension of
MALL with least and greatest fixed points allowing inductive and
coinductive statements. This logic contains both a tensor and a
co-product, and its strict linearity makes it a good fit for capturing
the type system of~\cite{sabry2018symmetric}.

In the literature, multiple proof-systems
have been considered for {\mumall}, some finitary proof system with
explicit induction inferences à la Park~\cite{baelde07least,
baelde12least} as well as non-wellfounded proof systems which allow
to build infinite derivations~\cite{baelde2016infinitary,bouncing}. The
present paper focuses on the latter. In general, an infinite
derivation is called a \emph{pre-proof}. The ability to derive from
infinite branches easily leads to inconsistency: to solve this
problem, {\mumall} comes equipped with a
\emph{validity criterion}, describing when an infinite derivation
can be considered as a logical proof.

\paragraph*{Contributions}
The main contribution of this work is a Curry-Howard correspondence
between (a fragment of) \mumall and a purely reversible typed language
in the style of~\cite{sabry2018symmetric}, with added generalized
inductive types and terminating recursion, in which recursive
functions must be structurally recursive. In particular, we show the
following:
\begin{itemize}
\item Totality, reversibility. Every function that we can encode
is reversible and total.
\item Expressivity. The language captures the class of primitive
  recursive functions~\cite{rogers1987theory}.
\item Curry-Howard correspondence. We show how well-typed
  encodable functions can be regarded as valid proofs
  of type isomorphisms between \mumall formulae.
\end{itemize}

\paragraph*{Organization of the paper}
The paper is organized as follows: in
Section~\ref{sec:technical-background} we introduce the technical
background needed for this work, split into two parts:
Subsection~\ref{background:rpp} focuses on a reversible model of
computation, called $\RPP{}$~\cite{rpp}, while
Subsection~\ref{background:mumall} introduces the logic {\mumall}. In
Section~\ref{section:language} we introduce the language, its syntax,
typing rules and operational semantics and show that any function that
can be encoded in our language represents an isomorphism. We finally
show in Section~\ref{sec:rpp} the expressiveness of the language by
encoding the language $\RPP{}$ into our language. Then in
Section~\ref{sec:ch}, we develop on the Curry-Howard correspondence
part: we show how to translate a well-typed term from our language
into a circular derivation of the logic {\mumall},
show that the given derivation respects the validity condition and show how
our evaluation strategy simulates the cut-elimination procedure of the
logic.

\section{Technical Background}
\label{sec:technical-background}

We introduce the technical background necessary for this work. The
first part, in Subsection~\ref{background:rpp}, focuses on the language of
Recursive Primitive Permutations (RPP), which we use as a model of
reversible computing to show the expressiveness of our language. The
second part, in Subsection~\ref{background:mumall}, focuses on the logic
{\mumall}. We do not give all the details on {\mumall} but just the
necessary definitions and intuitions. More details for RPP can be
found in~\cite{rpp}
while~\cite{baelde07least,baelde12least,baelde2016infinitary,
these-doumane, bouncing} contains more details on the
logic {\mumall}.

\subsection{Background on RPP}
\label{background:rpp}
Although reversible computation aims at capturing computability in a
reversible setting, not all bijections are computable. Defining
classes of such computable bijections have been the subject of several
research programs, yielding reversible Turing machine~\cite{bennett1973logical} and
functional programming languages~\cite{GLUCK2023113429, yokoyama2011reversible,james2014theseus}.
The main problem with a
general approach towards computable functions is when we care about
ruling out diverging functions. In the classical setting,
\emph{primitive recursive} functions (PRF) are an answer to this
problem: primitive recursive functions are both \emph{total} and
computable. Although limited, this class nonetheless captures a
sensible notion of computability.

The class of Recursive Primitive Permutations (RPP)~\cite{rpp} is a
class of total bijective functions, expressive enough to capture
primitive recursion.  Similar to what is done for RPF, $\RPP{}$
consists in a finite number of (total and bijective) generating
functions and combinators. If for RPF, domains and codomains are
products of $\mathbb{N}$, the elements of RPP are instead bijections
over product of $\intT$.

\subsubsection{Generators and combinators.}
More precisely, the generators are as follows. We first have the
successor ($S$), predecessor ($P$), identity ($\opn{Id}$) and
sign-change ($\opn{Sign}$), all acting on $\intT$, We then have the
binary swap $X$ acting on $\intT\times\intT$: $X(x,y)=(y,x)$. Finally,
given the RPP functions $f, g, h$ acting on $\intT^k$ and a RPP
function $j$ acting on $\intT^l$, the following functions are also in
RPP:
\begin{itemize}
\item The sequential composition $f;g$, defined by
  $(x_1,\ldots,x_k)\mapsto g(f(x_1,\ldots,x_k))$, bijection on
  $\intT^k$.
\item The parallel composition $f\mid\mid j$, defined as the
  bijection $f\times j$ acting on $\intT^{k+l}$.
\item The iterator $\mathbf{It}[f]\in\RPP{k+1}$, which iterate $f$
  on the $k$ first arguments, by the absolute value of its $(k+1)$th
  argument.
\item The selection $\mathbf{If}[f,g,h]\in\RPP{k+1}$, which applies
  either $f$, (resp. $g$ or $h$) depending on whether its $(k+1)$th
  argument is strictly greater than $0$ (resp. equal to $0$ or strictly
  less than $0$).
\end{itemize}
The semantics of generators and combinators is also given in
\autoref{fig:rpp}, using a graphical form where the left-hand side
variable of a diagram represent the input of the function and the
right-hand side its output.

\begin{figure}[t]
\[
    \begin{bNiceMatrix}[first-col, last-col=2, nullify-dots]
        x & S & x + 1
    \end{bNiceMatrix}
\qquad
    \begin{bNiceMatrix}[first-col, last-col=2, nullify-dots]
        x & P & x - 1
    \end{bNiceMatrix}
\qquad
    \begin{bNiceMatrix}[first-col, last-col=2, nullify-dots]
        x & \operatorname{Sign} & -x
    \end{bNiceMatrix}
\qquad
    \begin{bNiceMatrix}[first-col, last-col=2, nullify-dots]
        x & \operatorname{Id} & x
    \end{bNiceMatrix}
\qquad
\begin{bNiceMatrix}[first-col, last-col=2, nullify-dots]
    x & \phantom{\mathcal{X}} \begin{tikzpicture}[remember picture, overlay] \node at (-0.14, -.15) {$\mathcal{X}$}; \end{tikzpicture} & y \\ 
    y & \phantom{\mathcal{X}} & x
\end{bNiceMatrix}
\]
\[
    \begin{bNiceMatrix}[first-col, last-col=2, nullify-dots]
        x_1 & ~~ & y_1 \\
        \vdots & f;g & \vdots \\
        x_k & ~~ & y_k \\
    \end{bNiceMatrix} =
    \begin{bNiceMatrix}[first-col, nullify-dots]
        x_1 & ~~  \\
        \vdots & f \\
        x_k & ~~ \\
    \end{bNiceMatrix}\begin{bNiceMatrix}[last-col=2, nullify-dots]
        ~~ & y_1 \\
        g & \vdots \\
        ~~ & y_k \\
    \end{bNiceMatrix}
\qquad
    \begin{bNiceMatrix}[first-col, last-col=2, nullify-dots]
        x_1 & ~~ & y_1 \\
        \vdots & ~~ & \vdots \\
        x_k & f \mid\mid j  & y_k \\
        x'_1 & ~~ & y'_1 \\
        \vdots & ~~ & \vdots \\
        x'_l & ~~ & y'_l \\
    \end{bNiceMatrix} =
        \begin{array}{c}
        \begin{bNiceMatrix}[first-col, last-col=2]
        x_1 & ~~ & y_1 \\
        \vdots & f & \vdots \\
        x_k & ~~ & y_k \\
        \end{bNiceMatrix} \\
        \begin{bNiceMatrix}[first-col, last-col=2]
        x'_1 & ~~ & y'_1 \\
        \vdots & j & \vdots \\
        x_l & ~~ & y'_l \\
        \end{bNiceMatrix}
        \end{array}
\]
\[
        \begin{bNiceMatrix}[first-col, last-col=2, nullify-dots]
            x_1 & ~~ & y_1 \\
            \vdots & \mathbf{If}[f,g,h] & \vdots \\
            x_k & ~~ & y_k \\
            x &  & x \\
        \end{bNiceMatrix}
        \hspace{-0em}\setlength{\arraycolsep}{0pt}
        \begin{array}{ c }
          \left.\kern-\nulldelimiterspace
          \vphantom{\begin{array}{ c }
              b_N  \\
            \vdots \\
              z_N
          \end{array}}
          \right\}= \left\{ \begin{array}{cc}
            f~(x_1, \dots, x_k) & \text{ if } x > 0 \\
            g~(x_1, \dots, x_k) & \text{ if } x = 0 \\
            h~(x_1, \dots, x_k) & \text{ if } x < 0 \\
        \end{array} \right. \\
          \vphantom{a_N}
        \end{array}
\qquad
\begin{bNiceMatrix}[first-col, last-col=2, nullify-dots]
    x_1 & ~~ & y_1 \\
    \vdots & \mathbf{It}[f] & \vdots \\
    x_k & ~~ & y_k \\
    x &  & x \\
\end{bNiceMatrix}
\hspace{-0em}\setlength{\arraycolsep}{0pt}
\begin{array}{ c }
    \left.\kern-\nulldelimiterspace
    \vphantom{\begin{array}{ c }
        b_N  \\
    \vdots \\
        z_N
    \end{array}}
    \right\}= \underbrace{(f;\dots ; f)}_{\alt x \alt} (x_1, \dots, x_k) \\
    \vphantom{a_N}
\end{array}
\]
\caption{Generators of RPP}
\label{fig:rpp}
\end{figure}

\begin{rem}
  Note that the class $\RPP{}$ is naturally graded: we can define
  $\RPP{k}$ as the set of bijections where domain and codomain
  coincide. The class $\RPP{}$ is thus the union of all of the
  sub-classes $\RPP{k}$, when $k$ ranges over the natural numbers.
\end{rem}

\begin{rem}
  In the original paper \cite{rpp}, the class $\RPP{}$ is defined with
  two other kinds of generators: generalized permutations and
  weakening. They were only added for
  convenience: they can be derived from the ones
  we give. We therefore do not consider them here.
\end{rem}

\subsubsection{Inversion.}
For all $k$, the elements of $\RPP{k}$ are indeed bijections over
$\intT^k$. Moreover, the inverse of a given element can be inductively
constructed as follows:
\begin{center}
\begin{tabular}{ccc}
     $\operatorname{Id}^{-1} = \operatorname{Id}$ &
     $\operatorname{S}^{-1} = \operatorname{P}$ &
     $\operatorname{P}^{-1} = \operatorname{S}$ \\
     $\operatorname{Sign}^{-1} = \operatorname{Sign}$ &
     $\mathcal{X}^{-1} = \mathcal{X}$ & $(g; f)^{-1} =
     f^{-1} ; g^{-1}$ \\
    $(f \mid\mid g)^{-1} = f^{-1} \mid\mid g^{-1}$ &
    $(\mathbf{It}[f])^{-1} = \mathbf{It}[f^{-1}]$ & $(\mathbf{If}[f,
    g, h])^{-1} = \mathbf{If}[f^{-1}, g^{-1}, h^{-1}]$.
\end{tabular}
\end{center}

\subsubsection{Relationship with PRF}
Finally, and most importantly, $\RPP{}$ is PRF-Sound and Complete: it
indeed captures the notion of primitive recursive permutations.

\begin{thm}[Soundness \& Completeness~\cite{rpp}]
  \label{thm:rpp:soudness-completeness}
  $\RPP{}$ is PRF-Sound and PRF-Complete: it can represent any
  Primitive Recursive Function (PRF). Conversely, every function in
  $\RPP{}$ is indeed an element of the class RPF.\qed
\end{thm}

\begin{rem}
  While $\RPP{}$ is expressive enough to encode any Primitive
  Recursive Function, it does so at the cost of auxiliary inputs and
  outputs. The canonical example requiring such extension is the
  Cantor Pairing, building a bijection between $\natT$ and
  $\natT \times \natT$~\cite[Theorem 2 and Theorem 4]{rpp}. Note
  however, that this is a standard trick in reversible models of
  computation.
\end{rem}

\subsection{Background on \texorpdfstring{\mumall}{muMALL}}
\label{background:mumall}

Functional programming languages often feature the ability to encode
\emph{inductive types} and \emph{coinductive types} as data structure.
On a Curry-Howard point of view, how inductive and coinductive types
and reasoning are related to Linear Logic is not directly clear.

From a proof theory point of view, inductive and coinductive reasoning
have been studied for a long time. Most notably in the modal
$\mu$-calculus~\cite{de1979recursive, park1969fixpoint,
kozen1983results}, an extension of modal logic with fixed point
operators. Modal logic was extended with a new formula of the form
$\mu X. A$ where $A$ is a formula, along with the dual operator $\nu
X. A$. Already at this point in time, the $\mu$ and $\nu$
operators described the least and greatest fixed point
operators. By the Curry-Howard correspondence, those new logics
helped in modeling possibly infinite computation. Among those logics
and derivations, \emph{circular proofs}~\cite{santocanale,
fortier_et_al}, infinite proofs with finitely many subtrees, are of
particular interest to represent recursive programs. As we are
concerned with the particular case of Linear Logic, we look at the
logic {\mumall}. In his PhD, Baelde~\cite{baelde-phd} extended
linear logic with fixed points and showed the cut-elimination and the
equivalence of provability with regards to the higher-order
linear logic. Finally, Doumane et
al~\cite{baelde2016infinitary,these-doumane} investigated more
precisely the infinitary aspects of the derivations with a
proof-theoretical approach, defining validity criterions and their
properties.

\subsubsection{\mumall non-wellfounded derivations}
Baelde and his collaborators~\cite{baelde12least,baelde2016infinitary}
introduced the logic \mumall, an extension of the additive and
multiplicative fragment of linear logic~\cite{girard1987linear} with
least and greatest fixed-points. The grammar for linear logic
formulas, denoted by $F,G$, is extended with the construct $\mu X. F$
and its dual construct, $\nu X. F$ (where $X$ is a type variable and
$\mu,\nu$ are variable binders), to be interpreted as the least and
greatest fixed points of the operator $X\mapsto F$ respectively. These
permit to form inductive and coinductive statements: one can for
instance define the type of natural numbers as $\natT = \mu X.
\one\oplus X$ or of lists of type $F$ as $[F] = \mu X. \one \oplus
(F\otimes X)$.

\begin{defi}[Pre-formulas]
  Given an infinite set of variables $\mathcal{V} = \{X, Y, \dots\}$,
  the set of \mumall \emph{pre-formulas} is inductively defined by the following grammar:
  $$F,G ::= X \mid \one \mid \zero \mid \top \mid \bot \mid F
  \otimes G \mid F\parr G\mid F\oplus G \mid F\with
  G\mid \mu X. F\mid \nu X. F.$$
In $\mu X. F$ (resp. $\nu X. F$), $\mu$ (resp. $\nu$) binds the free occurrences of variable
$X$ in $F$.

A \emph{formula} is a \emph{closed}
  pre-formula ({\it i.e.} with no free variable).

  The \emph{(linear) negation} is the involution on pre-formulas defined as:
  $X^\bot = X, \zero^\bot = \top, \one^\bot = \bot, (F\parr
  G)^\bot = F^\bot \otimes G^\bot, (F\oplus G)^\bot =
  F^\bot \with G^\bot, (\nu X. F)^\bot = \mu X. F^\bot$.

    Given $F,G$ two formulas, $F[X\leftarrow G]$ denotes the
capture-free substitution of all the free-occurrences of variable $X$ in $F$ by $G$.
\end{defi}

In the following, all derivations will manipulate formulas only.
Setting $X^\bot = X$ is therefore harmless in the definition of negation.
In the following, we shall consider one-sided sequents:

\begin{defi}[\mumall sequents]
A \emph{\mumall sequent} is a finite ordered list of \mumall
formulas, usually denoted $\vdash \Gamma$. The concatenation of two
lists of formulas $\Gamma$ and $\Delta$ is simply written as
$\Gamma,\Delta$.
\end{defi}

In the following, we will be interested in the non-wellfounded and
circular proof systems used to derive \mumall sequents,
that we shall simply refer to as \mumall sequent calculus in the
following\footnote{Note that, in the literature, the sequent calculus
we consider here is often called $\mumall^\infty$ to distinguish it
from Baelde and Miller's calculus but since we shall not consider the
latter logical system in the present work, this
introduces no confusion.}: their inference rules are the usual
$\textsf{MALL}$ inference together with the addition of two new
inference rules (in the one-sided sequent calculus) for the $\mu$ and
$\nu$ connectives:

\[\infer[\mu]{\vdash \mu X. F, \Delta}{\vdash F[X \leftarrow \mu
X. F], \Delta} \qquad \infer[\nu]{\vdash \nu X. F,
\Delta}{\vdash F[X \leftarrow \nu X. F], \Delta}\]

\begin{defi}[\mumall inference rules and pre-proofs]
  \label{def:rules-mumall}
  A {\mumall} \emph{pre-proof} of conclusion $\Gamma$ is a sequent derivation tree coinductively  generated by the rules of~\autoref{fig:rules-mumall} such that the root is labelled with $\vdash \Gamma$.
  A \emph{circular pre-proof} is a pre-proof having only finitely many distinct subtrees.

  A formula is \textit{principal} when it is the formula to which the
  rule is being applied. Given an inference rule $r$, we denote by
  $\premiss r$ the set of its premiss sequents
 (i.e. above the line).

 Each inference rule is given together with an \emph{immediate
 sub-occurrence} relation, relating a formula of the conclusion
 sequent with formulas of the premisses, depicted by the green and
 blue lines in the proof system. The reflexive transitive closure of
 this relation is written $F\sqsubseteq G$ when $F$ is a formula
 occurrence (that is a position in a given sequent occurrence) that is
 above $G$ in a proof and related by any number of immediate
 sub-occurrence relations.
\end{defi}

The suboccurrence relation will most often be left implicit since it
can usually be reconstructed from the context or the way the sequents
are written. A case when we shall make explicit this suboccurrence
relation is in the threading relation used to distinguish pre-proofs
from valid proofs in what follows.

Note also that in what follows, we will most often use the exchange
rule implicitly, denoting with an inference rule name, say $\rtensor$,
all the derived rules obtained by applying an exchange rule on top of
each premise of the rule and below the conclusion of the rule. For
example, we can freely write (here, $\rex^\star$ refers to a sequence
of exchange rules):

\begin{prooftree}
\hypo{\vdash \Gamma_1, F,\Gamma_2}
\hypo{\vdash \Delta_1, G,\Delta_2}
\infer{2}[\rotimes]{\vdash \Gamma_1, \Delta_2, F\tensor G, \Gamma_2,\Delta_1}
\end{prooftree}
\qquad for \qquad
\begin{prooftree}
\hypo{\vdash \Gamma_1, F,\Gamma_2}
\infer1[$\rex^\star$]{\vdash F, \Gamma_1,\Gamma_2}
\hypo{\vdash \Delta_1, G,\Delta_2}
\infer1[$\rex^\star$]{\vdash G, \Delta_1, \Delta_2}
\infer{2}[\rotimes]{\vdash F\tensor G, \Gamma_1,\Gamma_2,\Delta_1,\Delta_2}
\infer{1}[$\rex^\star$]{\vdash \Gamma_1, \Delta_2, F\tensor G, \Gamma_2,\Delta_1}
\end{prooftree}

\bigskip
What follows is an important notational convention that is used in most of the technical development of the paper when relating isos and \mumall proofs:

\begin{convention}
  The two-sided notation $\Gamma \vdash \Delta$ will denote sequent $\vdash
  \Gamma^\bot, \Delta$.
\end{convention}

For example
  we will write $\begin{array}{c}\infer[\nu]{\mu X. X \vdash F}{\mu
  X. X \vdash F}\end{array}$ for
  $\begin{array}{c}\infer[\nu]{\vdash \nu X. X, F}{\vdash \nu X. X,
  F}\end{array}$.
  (Note in particular that we do not modify the rule labels.)

\begin{rem}
  This two-sided notation helps in clarifying the computational
  reading of a derivation by distinguishing between input type and
  output type (even though in classical linear logic they are
  interchangeable) and it will also make lighter the notation for
  translating isos into \mumall proofs. In particular, we will
  essentially use notation $\Gamma \vdash \Delta$ in a specific case
  when the formulas in $\Gamma$ and $\Delta$ are positive formulas
  (that is, are built from $\tensor, \oplus, \mu, \one, \zero$ and
  atoms). This will allow for a convenient correspondence between
  \emph{left formulas} (in the two sided-presentation) and
  \emph{negative formulas} (in the one-sided presentation) that we
  will heavily rely on in the following sections.

  Note that an alternative would have been to present {\mumall} in the
  two-sided notation all along: apart from being essentially superficial
  notational distinction (at least for our present concerns), this would
  have departed from most of the literature on \mumall.
\end{rem}

\begin{table}[t]
  \setlength{\tabcolsep}{0.6em}

\[\begin{array}{@{}cccc@{}}
  \tinfax & \tinfcut & \tinfex & \\[20pt]
 \tinfparr  & \tinftensor & \tinfbot & \tinfone \\[20pt]
 \tinfwith  & \tinfplus & \tinftop & \text{(no rule for 0)}  \\[20pt]
&  \tinfnu & \tinfmu &
\end{array}
\]
\begin{tikzpicture}[overlay,remember picture,-,line cap=round,line width=0.1cm]
   \draw[rounded corners, smooth=2,green, opacity=.4] ([xshift=-2mm] pic cs:tex1) to ([xshift=3mm, yshift=2mm] pic cs:tex2);
   \draw[rounded corners, smooth=2,green, opacity=.4] ([xshift=3mm] pic cs:tex1) to ([xshift=-2mm, yshift=2mm] pic cs:tex2);
   \draw[rounded corners, smooth=2,cyan, opacity=.4] ( [xshift=7mm]pic cs:tex1) to ([xshift=7mm, yshift=2mm] pic cs:tex2);
   \draw[rounded corners, smooth=2,cyan, opacity=.4] ( [xshift=-5mm]pic cs:tex1) to ([xshift=-5mm, yshift=2mm] pic cs:tex2);
   \draw[rounded corners, smooth=2,cyan, opacity=.4] ([xshift=-1mm] pic cs:tbotb1) to ([xshift=-1mm, yshift=2mm] pic cs:tbotb2);
   \draw[rounded corners, smooth=2,green, opacity=.4] ([xshift=-1mm] pic cs:tpa1) to ([xshift=-3mm,yshift=2mm] pic cs:tpa2);
   \draw[rounded corners, smooth=2,green, opacity=.4] ([xshift=-1mm] pic cs:tpa1) to ([xshift=1mm,yshift=2mm] pic cs:tpa2);
   \draw[rounded corners, smooth=2,cyan, opacity=.4] ( [xshift=5mm]pic cs:tpa1) to ([xshift=5mm,yshift=2mm] pic cs:tpa2);
   \draw[rounded corners, smooth=2,green, opacity=.4] ([xshift=-6mm] pic cs:tta1) to ([xshift=-1mm,yshift=2mm] pic cs:tta2);
   \draw[rounded corners, smooth=2,green, opacity=.4] ([xshift=-6mm] pic cs:tta1) to ([xshift=-1mm,yshift=2mm] pic cs:tta3);
   \draw[rounded corners, smooth=2,cyan, opacity=.4] ([xshift=1mm] pic cs:tta1) to ([xshift=2mm,yshift=2mm] pic cs:tta2);
   \draw[rounded corners, smooth=2,cyan, opacity=.4] ([xshift=5mm] pic cs:tta1) to ([xshift=3mm,yshift=2mm] pic cs:tta3);

   \draw[rounded corners, smooth=2,green, opacity=.4] ([xshift=-2mm] pic cs:tpla1) to ([xshift=-3mm,yshift=2mm] pic cs:tpla2);
   \draw[rounded corners, smooth=2,cyan, opacity=.4] ( [xshift=6.5mm]pic cs:tpla1) to ([xshift=2.5mm,yshift=2mm] pic cs:tpla2);

   \draw[rounded corners, smooth=2,green, opacity=.4] ([xshift=-6mm] pic cs:twa1) to ([xshift=-1mm,yshift=2mm] pic cs:twa2);
   \draw[rounded corners, smooth=2,green, opacity=.4] ([xshift=-6mm] pic cs:twa1) to ([xshift=-1mm,yshift=2mm] pic cs:twa3);
   \draw[rounded corners, smooth=2,cyan, opacity=.4] ([xshift=1mm] pic cs:twa1) to ([xshift=2mm,yshift=2mm] pic cs:twa2);
   \draw[rounded corners, smooth=2,cyan, opacity=.4] ([xshift=1mm] pic cs:twa1) to ([xshift=3mm,yshift=2mm] pic cs:twa3);
   \draw[rounded corners, smooth=2,green, opacity=.4] (pic cs:tma1) to ([xshift=-18mm,yshift=2mm] pic cs:tma2);
   \draw[rounded corners, smooth=2,cyan, opacity=.4] ([xshift=5mm] pic cs:tma1) to ([xshift=2mm,yshift=2mm] pic cs:tma2);
   \draw[rounded corners, smooth=2,green, opacity=.4] (pic cs:tna1) to ([xshift=-18mm,yshift=2mm] pic cs:tna2);
   \draw[rounded corners, smooth=2,cyan, opacity=.4] ([xshift=5mm] pic cs:tna1) to ([xshift=2mm,yshift=2mm] pic cs:tna2);
   \draw[rounded corners, smooth=2,cyan, opacity=.4] ([xshift=-2mm] pic cs:tca1) to ([xshift=-2mm] pic cs:tca2);
   \draw[rounded corners, smooth=2,cyan, opacity=.4] ([xshift=1mm] pic cs:tca1) to ([xshift=1mm] pic cs:tca3);
   \end{tikzpicture}
  \caption{Rules of \mumall. The sub-occurrence relation is graphically depicted with colored lines between formulas of the conclusion and premisses sequents, in green for the principal formula, in blue for the other formulas.}
  \label{fig:rules-mumall}
\end{table}

For instance, taking the type of natural numbers that we defined
earlier, it is possible to encode any natural number as a derivation
of \mumall.

\begin{exa}
  \label{ex:mumall:nat}
  Representing the type of natural numbers by the pre-formula
  $\natT = \mu X. \one\oplus X$, one can define the proof $\pi_n$,
  encoding any natural number $n$. We give this encoding, by induction
  on $n$ as:

  \[\pi_0 = \begin{array}{c}\infer[\mu]{\vdash \mu X. \one\oplus
  X}{\infer[\oplus^1]{\vdash \one\oplus \natT}{\infer[\one]{\vdash
  \one}{}}}\end{array} \qquad\qquad \pi_n =
  \begin{array}{c}\infer[\mu]{\vdash \mu X. \one\oplus
  X}{\infer[\oplus^2]{\vdash \one\oplus \natT}{\infer{\vdash
  \natT}{\pi_{n-1}}}}\end{array}\]
\end{exa}

The particularity is that the derivations are generated
\textbf{coinductively} over the set of rules, allowing for infinite
derivation trees, called pre-proofs.

Their name comes from the fact that, if we consider that any
derivation is a \emph{proof}, then we can prove any statement $\psi$
using the cut-rule, as shown in~\autoref{mumall:ex:degenerated}. This is
why {\mumall} comes with a validity criterion, separating
\emph{pre-proofs} from actual \emph{proofs}.

The usual cut-elimination reduction rules of
MALL~\cite{girard1987linear} is extended with a new reduction rule and
two commutation rules (here, $\sigma\in\{\mu,\nu\}$):

\begin{center}
\scalebox{0.8}{
  \begin{prooftree}
  \hypo{\vdash F^\bot[X\leftarrow \mu X.F^\bot], \Gamma}
  \infer1[$\mu$]{\vdash \mu X. F^\bot, \Gamma}
  \hypo{\vdash F[X\leftarrow \nu X. F], \Delta}
  \infer1[$\nu$]{\vdash \nu X. F, \Delta}
  \infer2[cut]{\vdash \Gamma, \Delta}
\end{prooftree}
$\quad\leadsto\quad$\begin{prooftree}
  \hypo{\vdash F^\bot[X\leftarrow \nu X. F^\bot], \Gamma}
  \hypo{\vdash F[X\leftarrow \nu X. F], \Delta}
  \infer2[cut]{\vdash \Gamma,\Delta}
\end{prooftree}
}

\scalebox{0.8}{
  \begin{prooftree}
  \hypo{\vdash F[X\leftarrow \sigma X. F], \Gamma, C}
  \infer1[$\sigma$]{\vdash \sigma X. F, \Gamma, C}
  \hypo{\vdash C^\bot, \Delta}
  \infer2[cut]{\vdash \sigma X. F, \Gamma, \Delta}
\end{prooftree}
$\quad\leadsto\quad$\begin{prooftree}
  \hypo{\vdash F[X\leftarrow \sigma X. F], \Gamma, C}
  \hypo{\vdash C^\bot, \Delta}
  \infer2[cut]{\vdash F[X\leftarrow \sigma X. F], \Gamma,\Delta}
  \infer1[$\sigma$]{\vdash \sigma X. F, \Gamma, \Delta}
\end{prooftree}
}
\end{center}

\begin{figure}[t]
  \centering
  \[\begin{prooftree} \hypo{} \ellipsis{}{}
      \infer1[$\nu$]{\vdash \nu X. X}

      \hypo{} \ellipsis{}{}
      \infer1[$\mu$]{\vdash \mu X. X, \Gamma}
      \infer2[\cut]{\vdash \Gamma}
  \end{prooftree}\]
  \caption{Invalid pre-proof}
  \label{mumall:ex:degenerated}
\end{figure}

\subsubsection{Bouncing Validity}
The validity criterion requires that each infinite branch can be
justified by a form of coinductive reasoning. The criterion also
ensures that the cut-elimination procedure holds.
In the following, we will rely on a validity criterion recently introduced by Baelde et al.~\cite{bouncing} in a slightly simplified reformulation.

This criterion is based on Girard's Geometry of Interaction where some
data (namely a \emph{thread}) moves through the derivation, following
a subformula and collecting information (its \emph{weight}). Then, one
can analyze the collected information and determine whether or not it
is \emph{valid}. More formally, a
\emph{thread}~\cite{baelde2016infinitary,bouncing} is an infinite
sequence of tuples of formulas, sequents and directions (either up or
down) written $(F; \vdash \Phi; d)$. Intuitively, these threads follow
some formula (here $F$) starting from a sequent of the derivation and
starting by going up. The thread has the possibility to \emph{bounce}
on axioms and cuts and change its direction, either going back-down on
an axiom or back-up on a cut. A thread will be called \emph{valid}
when it is non-stationary (does not follow a formula that is never a
principal formula of a rule), and when in the set of formulas
appearing infinitely often, the minimum formula (according to the
subformula ordering) is a $\nu$ formula. We will then say that a
derivation is valid if any infinite branch is inhabited by a valid
thread.\footnote{Note that in {\mumall} one needs an additional notion of
\emph{slices} for the additive part of the logic. However, since we
will only consider a fragment of {\mumall} this notion will not be
required in our work.} The system of thread previously mentioned is
developed in~\cite{baelde2016infinitary} and extended
in~\cite{bouncing} and called the \emph{bouncing-thread-validity}.
This is the criterion that we present here and that we will use
in~\autoref{sec:ch}.

\medskip

We recall the formalism of bouncing validity from~\cite{bouncing},
adapting it to sequents as lists of formulas instead of locative
occurrences~\cite[Def 2.5]{bouncing}.

As mentioned earlier, not all derivations are indeed proofs. To answer
this problem, {\mumall} comes with a validity criterion for
derivations. This makes use of the notion of
\textit{bouncing-threads}: paths that travel along the infinite
derivation and collect some information along the way. In order to
formally define bouncing-threads we first need to introduce some
notations: given an alphabet $\Sigma$, we denote by
$\Sigma^\omega$ the set of infinite words over $\Sigma$ and
$\Sigma^\infty = \Sigma^* \cup \Sigma^\omega$. The letter
$\varrho$ will denote ordinals in $\omega + 1$. Finally, we use a
special concatenation: given $u = (u_i)_{i\leq n<
\omega}$ and $v = (v_i)_{i\in\varrho}$ such that $u_n = v_0$, we
define $u \odot v$ as the concatenation of $u$ and $v$ without the
first element of $v$, i.e. : $u \cdot
(v_i)_{i\in\varrho\backslash\{0\}}$. For instance $aab \odot bac =
aabac$.

We begin with the definition of \textit{pre-thread} : the
basic construction of a path along an infinite tree that follows a
formula occurrence. Then, we only look at some special ones called
\textit{threads}, and finally define the notion of \textit{valid
thread} that validates an infinite branch as being
valid. A pre-proof is a proof whenever all of its infinite branches
are valid.

\begin{defi}[Pre-thread~\cite{bouncing}]
  \label{mumall:pre-thread-def-csl}
  A pre-thread is a sequence $(F_i, s_i, d_i)_{i\in\varrho\in\omega+1}$ of tuples
  of a formula (occurrence), a sequent (occurrence) and a direction $d\in\{\uparrow,
  \downarrow\}$ such that for all $i\in\varrho$ and $i+1\in\varrho$
  one of the following holds:
  \begin{itemize}
    \item $d_i = d_{i+1} = \uparrow$, $s_{i+1} \in \premiss{s_i}$ and $F_{i+1}\sqsubseteq F_i$
    \item $d_i = d_{i+1} = \downarrow, s_i\in\premiss{s_{i+1}}$ and $F_i \sqsubseteq F_{i+1}$
    \item $d_i = \downarrow, d_{i+1} = \uparrow$, $s_i$ and $s_{i+1}$
    are the two premisses of the same cut rule and $F_i =
    F^\bot_{i+1}$
    \item $d_i = \uparrow, d_{i+1} = \downarrow$ and $s_i = s_{i+1} =
    \vdash F_i, F_{i+1}$ (or $\vdash F_{i+1}, F_i$) is the conclusion of an axiom rule.
  \end{itemize}
\end{defi}

Bouncing validity is a condition expressed on pre-threads of a
specific form. To state it, we need to first recall the notion of
\emph{weight} of a pre-thread from~\cite{bouncing}:

\begin{defi}[Weight of a pre-thread]
  Consider $t = (F_j, s_j, d_j)_{j\in\varrho\in\omega+1}$ a
  pre-thread. The weight of $t$, denoted $w(t)$, is a word over
  $(w_j)_{j\in\varrho\in\omega+1}\in
  \{l,r,i,\ov{l},\ov{r},\ov{i},\mathcal{W}, \mathcal{A},
  \mathcal{C}\}^\infty$ such that for every $j\in\varrho$ one of the
  following holds:
  \begin{itemize}
    \item thread going up ($d_j = \uparrow$):
    \begin{itemize}
      \item $w_j = l$ if $F_j = F_{j+1} \star G$ is principal in $s_j$, with $\star \in \{\parr,\tensor,\oplus,\with\}$;
      \item $w_j = r$ if $F_j = G \star F_{j+1}$ is principal in $s_j$, with $\star \in \{\parr,\tensor,\oplus,\with\}$;
      \item $w_j = i$ if $F_j = \sigma X. F$ is principal in $s_j$ and $F_{j+1} = F[X \leftarrow \sigma X. F]$;
      \item $w_j = \mathcal{W}$ if $F_j$ is not principal in $s_{j}$;
      \item $w_j = \mathcal{A}$ if $d_{j+1} =  \downarrow$.
    \end{itemize}
    \item thread going down ($d_j = \downarrow$):
    \begin{itemize}
      \item $w_j = \ov{l}$ if $F_{j+1} = F_{j} \star G$ is principal in $s_{j+1}$, with $\star \in \{\parr,\tensor,\oplus,\with\}$;
      \item $w_j = \ov{r}$ if $F_{j+1} =G \star F_{j}$ is principal in $s_{j+1}$, with $\star \in \{\parr,\tensor,\oplus,\with\}$;
      \item $w_j = \ov{i}$ if $F_{j+1} = \sigma X. F$ is principal in $s_{j+1}$ and $F_{j} = F[X \leftarrow \sigma X. F]$;
      \item $w_j = \mathcal{W}$ if $F_{j+1}$ is not principal in $s_{j+1}$;
      \item $w_j = \mathcal{C}$ if $d_{j+1} = \uparrow$.
    \end{itemize}
  \end{itemize}
\end{defi}

Weights of pre-thread allow us to characterize certain sets of
pre-threads of interest.

\begin{defi}[b-paths and h-paths]
We define two sets of words $\goth{B}$ and $\goth{H}$ inductively as
follows:
\[\goth{B} ::= \mathcal{C} \mid \goth{B} \mathcal{W}^* \mathcal{A}
      \mathcal{W}^* \goth{B} \mid \ov{x} \mathcal{W}^* \goth{B}
      \mathcal{W}^* x \qquad \qquad\goth{H} ::= \epsilon \mid \mathcal{A} \mathcal{W}^* \goth{B}\]

A finite pre-thread $t$ is called a \emph{b-path} if $w(t) \in
\goth{B}$, and it is called a \emph{h-path} if $w(t)\in \goth{H}$.
\end{defi}

We are now ready to state the notion of bouncing
thread~\cite{bouncing}:
\begin{defi}[(Bouncing) thread]
  A pre-thread $t$ is a \emph{(bouncing) thread} when it can be decomposed as
 $\bigodot_{i\in 1 + \varrho}(H_i \odot V_i)$ where for all $i\in 1 + \varrho$:
  \begin{itemize}
    \item $w(V_i) \in \{l,r,i,\mathcal{W}\}^\infty$, and it is non-empty if $i\not=\varrho$
    \item $w(H_i) \in \goth{H}$, and it is non-empty if $i\not= 0$
  \end{itemize}
\end{defi}

The decomposition can be read as a thread initially going up,
before going back down after encountering an axiom, accumulating
debt in the form of the alphabet $\set{\ov{l}, \ov{r}, \ov{i}}$, which
will need to be repaid by their opposites, $\set{l, r, i}$, when going
back up after encountering a cut.
 Those dual alphabets correspond to steps of
the cut-elimination: making sure that the correct formulas will at
some point interact by a cut.

Such a decomposition is unique, and we call $(V_i)_{i\in 1 + \varrho}$
the \textit{visible part} of $t$ and $(H_i)_{i\in 1 + \varrho}$ the
\textit{hidden part}. A thread is \textit{stationary} when its visible
part is a finite sequence (of finite words) or when there exists $k\in
1 + \varrho$ such that $w(V_i) \in \{\mathcal{W}\}^\infty$ for all
$k\leq i\in 1 + \varrho$.

\begin{exaC}[{\cite[Example 4.2]{bouncing}}]
  Consider the following derivation:

  \[
  \begin{prooftree}

    \hypo{\tikzmark{bax}}
    \infer1[$\id$]{\vdash \tikzmark{b3}\nu X. X, \tikzmark{b4}\mu X. X}
    \hypo{\tikzmark{vax}}
    \infer1[$\id$]{\vdash \tikzmark{v3}\nu X. X, \tikzmark{v4}\mu X. X}
    \infer2[$\parr, \otimes$]{\vdash \nu X. X \tikzmark{b2}\tikzmark{v2}\parr \nu X. X,
     \mu X. X \tikzmark{b5}\tikzmark{v5}\otimes \mu X. X}

    \hypo{\tikzmark{b8}\tikzmark{v8}\vdots}
    \infer1{\vdash \tikzmark{b7}\tikzmark{v7}\nu X. X, \nu X. X}
    \infer1[$\nu$]{\vdash \tikzmark{b6}\nu X. X, \tikzmark{v6}\nu X. X}
    \infer1[$\parr$]{\vdash \tikzmark{bb}\nu X. X \parr \tikzmark{vv}\nu X.X}
    \infer2[$\cut$]{\vdash \nu X. X \tikzmark{b1}\tikzmark{v1}\parr \nu X. X}
  \end{prooftree}
  \]
  ~\hfill\begin{tikzpicture}[overlay,remember picture,-,line cap=round,line width=0.1cm]
    \draw[rounded corners, smooth=2,red, opacity=.25]
    ($(pic cs:v1)+(.3cm,.2cm)$)   to
    ($(pic cs:v2)+(.4,.2cm)$)     to
    ($(pic cs:v3)+(.3,.2cm)$)     to
    ($(pic cs:vax)+(.3,.2cm)$)    to
    ($(pic cs:v4)+(.3,.0cm)$)     to
    ($(pic cs:v5)+(.3,.0cm)$)     to
    ($(pic cs:vv)+(.1,.0cm)$)     to
    ($(pic cs:v6)+(-.3,.0cm)$)    to
    ($(pic cs:v7)+(.3,-.0cm)$)    to
    ($(pic cs:v7)+(.3,.8cm)$);

    \draw[rounded corners, smooth=2,cyan, opacity=.25]
    ($(pic cs:b1)+(.2cm,.1cm)$)   to
    ($(pic cs:b2)+(.2,.2cm)$)     to
    ($(pic cs:b3)+(.2,.2cm)$)     to
    ($(pic cs:bax)+(.2,.2cm)$)    to
    ($(pic cs:b4)+(.3,.2cm)$)     to
    ($(pic cs:b5)+(.3,.1cm)$)     to
    ($(pic cs:bb)+(1,.1cm)$)      to
    ($(pic cs:b6)+(.2,.2cm)$)     to
    ($(pic cs:b7)+(.2,.0cm)$)     to
    ($(pic cs:b7)+(.2,.8cm)$);
\end{tikzpicture}\hfill~

We have two pre-threads, painted in red and blue. The weight of the
blue pre-thread is:
$\mathcal{W}l\mathcal{W}\mathcal{A}\bar{l}\mathcal{W}\mathcal{C}li\dots$,
which can be decomposed into its visible part (from the root to the
axiom) and its hidden part (from the axiom to the $\parr$ on the right
side of the derivation). Notice that such a decomposition is
impossible for the red pre-thread: it is not a thread.
\end{exaC}

We can now define the validity criterion on threads:

\begin{defi}[Valid threads~\cite{bouncing}]
  If we take the sequence of formulas followed by a non-stationary
  thread on its visible part and skipping the steps corresponding to
  $\mathcal{W}$ weights, we obtain an infinite sequence of formulas
  where each formula is an immediate subformula or an unfolding of the
  previous formula. The formulas appearing infinitely often admit a
  minimum with regard to the subformula ordering, such a formula is
  the \textit{minimal formula of the thread}.

  A non-stationary thread is \textit{valid} if its minimal formula is
  a $\nu$ formula.
\end{defi}

This notion of bouncing valid threads induces an associated
notion of valid proof. Note that in the definition below, we use a
simpler notion than that of~\cite{bouncing}, which we justify
immediately after the definition and example.

\begin{defi}[Validity of $\mu$MALL]
  We say that an infinite branch is valid if
  there exists a valid thread starting from one of its sequents, whose
  visible part is contained in this branch.

A $\mu$MALL pre-proof $\pi$ is \textit{valid} if every
infinite branch of $\pi$  is valid.
  \end{defi}

\begin{exa}

    The infinite derivation $\begin{array}{c}\infer[\mu]{\vdash \mu X.
    X}{\infer[\mu]{\vdash \mu X. X}{\vdots}}\end{array}$ is not valid
    as the only existing thread $(\mu X. X; \vdash \mu X. X;
    \up)^\infty$ has as sole minimal formula a $\mu$ formula.
\end{exa}

This validity condition ensures cut-elimination for {\mumall}, as proved in~\cite{bouncing}.
The reader will notice a difference in the notion of
$\mumall$ validity between~\cite{bouncing} and our work.
Indeed, bouncing validity as stated in~\cite{bouncing} relies on the notion of {\it slice}
(which comes from LL proof-nets). However, as we shall only be interested
in a fragment of {\mumall} pre-proofs in the following, that is the target of the translations from  {\it iso}s to proofs, we will not need to reason about
slices (this will be further detailed in Section~\ref{sec:ch}).

\begin{rem}
Original cut-elimination results on non-wellfounded
{\mumall} dealt with sequents as
sets of (locative) occurrences~\cite{baelde2016infinitary,bouncing} rather than sequents as lists as we
consider here. However, it was shown recently by the second
author~\cite{Saurin23} how to transfer those cut-elimination results
to sequent calculi manipulating lists, in a uniform way.
\end{rem}

\begin{exa}
  \label{mumall:ex:succ}
    Remember that $\natT = \mu X. \one\oplus X$ and hence $\natT^\bot
    = \nu X. \bot\with X$. We can then define the successor function
    as:

    \noindent
    \begin{minipage}{\textwidth}
    \[\pi_{\opn{succ}} = \begin{prooftree}

      \infer0[$\one$]{\vdash \one}
      \infer1[$\oplus^1$]{\vdash \one\oplus\natT}
      \infer1[$\mu$]{\vdash \natT}
      \infer1[$\bot$]{\one\vdash \natT}

      \hypo{$\tikzmark{red6}$\pi_{\opn{succ}}}
      \infer1{$\tikzmark{red5}$\natT\vdash$\tikzmark{blue5}$\natT}
      \infer1[$\oplus^2$]{$\tikzmark{red4}$\natT\vdash $\tikzmark{blue4}$\one\oplus\natT}
      \infer1[$\mu$]{$\tikzmark{red3}$\natT\vdash $\tikzmark{blue3}$\natT}
      \infer2[$\with$]{$\tikzmark{red2}$\one\oplus\natT \vdash $\tikzmark{blue2}$\natT}
      \infer1[$\nu$]{$\tikzmark{red1}$\natT \vdash $\tikzmark{blue1}$\natT}
    \end{prooftree}
  \]
  ~\hfill\begin{tikzpicture}[overlay,remember picture,-,line cap=round,line width=0.1cm]
      \draw[rounded corners, smooth=2,red, opacity=.25]
      ($(pic cs:red1)+(.3cm,.1cm)$)   to
      ($(pic cs:red2)+(.4,.2cm)$)     to
      ($(pic cs:red3)+(.3,.2cm)$)     to
      ($(pic cs:red4)+(0,.2cm)$)      to
      ($(pic cs:red5)+(.4,.5cm)$)     to
      ($(pic cs:red6)+(0.3,0.5cm)$);

      \draw[rounded corners, smooth=2,cyan, opacity=.25]
      ($(pic cs:blue1)+(+.2cm,.1cm)$)   to
      ($(pic cs:blue2)+(+.2,.2cm)$)     to
      ($(pic cs:blue3)+(+.2,.2cm)$)     to
      ($(pic cs:blue4)+(+.5,.2cm)$)     to
      ($(pic cs:blue5)+(+.1,1.1cm)$);
    \end{tikzpicture}\hfill~
  \end{minipage}

We can show that this derivation is valid, by following the two
possible thread (painted in red and blue) in the infinite branch.

The red-thread $t_l$ have for weight $(irWW)^\infty$ and the set of
its formulas encountered infinitely often is $\set{\nu X. \bot\with
X, \bot\with \nu X.\bot\with X}$. The smallest formula is a $\nu$
formula, validating the branch.

On the other hand, the blue-thread has weight $(WWir)^\infty$ and for
smallest formula a $\mu$ formula and hence is not valid. Since we only
require \emph{one} valid thread per infinite branch, the whole
derivation is therefore valid and a proof.
\end{exa}

\begin{exa}
  Recall the type of lists of type $\psi$ as $[\psi] = \mu X.
  \one\oplus (\psi\otimes X)$.
    \label{ex:mumall-thread}
    \newcommand{\aob}{\ensuremath{F}}
    \newcommand{\boa}{\ensuremath{G}} Consider the {\mumall} proof
    $\pi_S$ of $A\otimes B\vdash B\otimes A$, which swaps the input of
    type $A\otimes B$, defined as:

    \[\pi_S = \begin{array}{c}\infer[\parr]{A\otimes B \vdash
    {B\otimes A}}{\infer[\otimes]{{A}, {B}\vdash {B\otimes
    A}}{\infer[\id]{{A}\vdash {A}}{}\qquad \infer[\id]{{B}\vdash
    {B}}{}}}\end{array}\]

We can use this proof to define the infinite, circular proof
$\phi(\pi_S)$, which swap the elements of a given input list,
where $\aob = {A\otimes B}$ and $\boa =
{B\otimes A}$:

  \[\scalebox{0.77}{
    \begin{prooftree}
      \infer0[$\one$]{\vdash \one}
      \infer1[$\oplus^1$]{\vdash \one \oplus (\boa \otimes [\boa])}
      \infer1[$\mu$]{\vdash [\boa]}
      \infer1[$\bot$]{\one \vdash [\boa]}
  \infer0[\id]{\aob\vdash \aob}
      \hypo{\pi_S}
      \infer1{\aob\vdash \boa}
      \infer2[\cut]{\aob\vdash \boa}
    \hypo{\tikzmark{r77}\qquad\tikzmark{r88}}
    \infer1[\id]{\colorbox{cyan!30}{$[\tikzmark{r7}\aob]$}\vdash \colorbox{cyan!30}{$[\tikzmark{r8}\aob]$}}
      \hypo{\tikzmark{r10}\phi(\pi_S)}
      \infer1{\colorbox{cyan!30}{$[\tikzmark{r9}\aob]$}\vdash [\boa]}
      \infer2[\cut]{\colorbox{cyan!30}{$[\tikzmark{r6}\aob]$}\vdash [\boa]}

    \infer0[\id]{\boa\vdash \boa}
    \infer0[\id]{[\boa]\vdash [\boa]}
    \infer2[$\otimes$]{\boa, [\boa]\vdash (\boa) \otimes [\boa]}
    \infer1[$\oplus^2$]{\boa, [\boa]\vdash \one \oplus (\boa \otimes [\boa])}
    \infer1[$\mu$]{\boa, [\boa]\vdash [\boa]}

    \infer1[]{\boa, [\boa]\vdash [\boa]}
    \infer2[\cut]{\boa, \colorbox{cyan!30}{$[\tikzmark{r5}\aob]$}\vdash [\boa]}
    \infer2[\cut]{\aob, \colorbox{cyan!30}{$[\tikzmark{r4}\aob]$}\vdash [\boa]}
    \infer1[$\parr$]{\aob \otimes \colorbox{cyan!30}{$[\tikzmark{r3}\aob]$} \vdash [\boa]}

    \infer2[$\with$]{\one \oplus \colorbox{cyan!30}{$(\aob \tikzmark{r2}\otimes [\aob])$} \vdash [\boa]}
    \infer1[$\nu$]{ \colorbox{cyan!30}{[\tikzmark{r1}\aob]}\vdash [\boa]}

    \end{prooftree}

    \begin{tikzpicture}[overlay,remember picture,-,line cap=round,line width=0.1cm]
      \draw[rounded corners, smooth=2,cyan, opacity=.25] ($(pic
    cs:r1)+(-.1cm,.1cm)$) to ($(pic cs:r2)+(.1cm,-.2cm)$)to ($(pic
    cs:r2)+(-.1,.2cm)$) to($(pic cs:r3)+(0,.2cm)$) to($(pic
    cs:r4)+(0,.2cm)$) to($(pic cs:r5)+(0,.1cm)$) to($(pic
    cs:r6)+(.1cm,.1cm)$) to($(pic cs:r7)+(0,.2cm)$) to($(pic
    cs:r77)+(0,.2cm)$) to($(pic cs:r88)+(0,.2cm)$) to($(pic
    cs:r8)+(0,.2cm)$) to($(pic cs:r9)+(0,.2cm)$) to($(pic
    cs:r10)+(0,.2cm)$);
    \end{tikzpicture}

    }
  \]

  We painted in blue the pre-thread that will validate the proof. The
  pre-thread starts at the root on the formula $[F]$ and follows a
  sub-formula when some rules (here $\nu, \with, \parr$) are applied
  to it. The pre-thread is then inactive during the multiple
  \cut{} rules until it reaches the \id{} rule, where the
  pre-thread bounces and starts going down before bouncing back up
  again in the \cut{} rule, into the infinite branch, where the
  behavior of the pre-thread will repeat itself. One can then show
  that this pre-thread is indeed a thread, according
  to~\cite{bouncing} and that it is valid: among the formulas it
  visits infinitely often on its visible part, the minimal formula is
  $[F]$, which is a $\nu$ formula.

  Notice that the sub-derivation containing the
  \cut{} and the proof $\pi_S$ could have been simplified without the
  \cut{} and by putting $\pi_S$ directly.
  However, as we will see in Section~\ref{sec:ch},
  this derivation will correspond exactly to the translation
  of a term.
\end{exa}

\subsubsection{Circular Representation}
As mentioned earlier, {\mumall} allows us to build infinite
derivations and the fragment of \emph{circular} proofs is of a special
interest as it allows to be finitely represented and therefore
subject to algorithmic treatment. We now introduce circular
representations which provide a finite mean to deal with circular
proofs with the help of \emph{back-edges}. Back-edges are arrows in
the derivation that represent a repetition of the derivation as a
(strict) sub-derivation. Derivations with back-edges are
represented with the addition of sequents marked by a back-edge label,
denoted $\vdash^f$, and an additional rule,
$\begin{array}{c}\infer[be(f)]{\vdash \Gamma}{}\end{array}$, which
represents a back-edge pointing to the sequent $\vdash^f$. We take the
convention that from the root of the derivation  to
rule $be(f)$ there must be \emph{exactly one} sequent annotated by
$f$.

\begin{exa}\label{ex:circular-proof} An infinite derivation and three
  different circular representations with back-edges:

  \[\begin{prooftree}

    \hypo{}
    \ellipsis{\tikzmark{ex_top3}}{}
    \infer1[$\mu$]{\vdash \mu X. X}
    \infer1[$\mu$]{\vdash \tikzmark{ex_bot3}\mu X. X}
    \end{prooftree}\qquad
    \qquad
    \qquad
    \begin{prooftree}
    \infer0[$be(f)$]{\vdash \mu X. X}
    \infer1[$\mu$]{\vdash^f \mu X. X}
    \end{prooftree}
    \qquad
    \begin{prooftree}
    \infer0[$be(f)$]{\vdash \mu X. X}
    \infer1[$\mu$]{\vdash^g \mu X. X}
    \infer1[$\mu$]{\vdash^f \mu X. X}
    \end{prooftree}
    \qquad
    \begin{prooftree}
    \infer0[$be(g)$]{\vdash \mu X. X}
    \infer1[$\mu$]{\vdash^g \mu X. X}
    \infer1[$\mu$]{\vdash^f \mu X. X}
    \end{prooftree}
  \]
\end{exa}

While a(n infinite) circular proof has infinitely many circular
representations (depending on where the back-edge is placed), an
unfolding operation allows to map those representations to the
circular derivation:

\begin{defi}[Unfolding]

  We define the unfolding of a circular representation $P$ with a
  valuation $v$ from back-edge labels to circular-derivations along
  with a valuation by:
  \begin{itemize}
    \item $\unfold\left(P :\begin{array}{c}
    \infer[r]{\vdash\Gamma}{P_1 ~\dots~ P_n}\end{array}, v\right) =
    \begin{array}{c}\infer[r]{\vdash\Gamma}{\unfold(P_1, v) ~\dots~ \unfold(P_n, v)}\end{array}$
    \item \medskip$\unfold(be(f), v) = \unfold{(v(f))}$
    \item \medskip$\unfold\left(P : \begin{array}{c}\infer[r]{\vdash^f
    \Gamma}{P_1 ~\dots~ P_n}\end{array}, v\right) =
    \left(\begin{array}{c}\infer[r]{\vdash\Gamma}{\unfold(P_1,
      v') ~\dots~ \unfold(P_n, v')}\end{array}\right)$ \\ with
$v'(g) =\left\{
  \begin{array}{@{}ll@{}}
    (P, v) & \text{if $g = f$,}\\
    v(g) & \text{else}.
  \end{array}\right.$
\end{itemize}
\end{defi}

\section{First-order Isos}\label{section:language}

Sabry et al.~\cite{sabry2018symmetric} introduced a typed, functional,
reversible programming language with the type of lists as the only
infinite data-type. Although their paper eventually extends the
framework to quantum computing, in our paper we only consider the
reversible aspect. In particular, we focus on the first-order
fragment, extended with both a more general rewriting system and more
general inductive type. We consider arbitrary inductive types, making
it possible to not only build list but also e.g. trees. We start by
presenting the syntax of the language before moving on to its
rewriting system and then its typing system.

\subsection{Terms and Types}
The language is first-order: it consists of basic types, i.e. types
for values, patterns, expressions, and general terms, and isos types,
i.e. function types. The grammar for these notions is given as follows.
\begin{alignat*}{100}
  &\text{(Base types)} \quad& A, B &\hspace{3ex}&&::= &\hspace{3ex}& \one \alt A \oplus B
                                              \alt A \otimes B \alt \mu X. A \alt X \\
  &\text{(Isos types)} & T &&&::=&&A\iso B \\
  &\text{(Values)} & v &&&::=&& () \alt x \alt \inl{v} \alt \inr{v}
                                \alt
                                \pv{v_1}{v_2} \alt \fold{v} \\
  &\text{(Pattern)} & p &&&::=&& x \alt \pv{p_1}{p_2} \\
  &\text{(Expressions)} & e &&&::=&& v \alt \letv{p_1}{\omega~p_2}{e} \\
  &\text{(Isos)} & \isoterm &&&::=&& f \alt \isobasique \alt \\&&&&&&&\fix
                                     f.\isobasique \\
  &\text{(Terms)} & t &&&::=&& () \alt x \alt \inl{t} \alt \inr{t}
                               \alt
                               \pv{t_1}{t_2} \alt  \\
  & & &&& && \fold{t} \alt \isoterm~t \alt
             \letv{p}{t_1}{t_2}
\end{alignat*}

\begin{rem}
  \label{rem:Barendregt}
  The language makes use of free and bound variables.
  In order to avoid conflicts between variables we will always
  work up to $\alpha$-conversion and use Barendregt's
  convention~\cite[p.26]{henk1984lambda} which consists in keeping all
  bound and free variable names distinct, even when this remains
  implicit.
\end{rem}

\subsubsection{Values, patterns, expressions, terms and basic types}
They allow us to construct first-order terms. The constructors
$\inl{}$ and $\inr{}$ represent the choice between either the left or
right-hand side of a type of the form $A\oplus B$; the constructor
$\pv{}{}$ builds pairs of elements, with the corresponding type
constructor $\otimes$; $\fold{}$ represents inductive constructor of
the inductive type $\mu X. A$, where $\mu$ is a binder, binding the
type-variable $X$ in $A$. A value can serve both as a result and as a
pattern in the defining clause of an iso. We write $(x_1, \dots, x_n)$
for $\pv{x_1}{\pv{\dots}{x_n}}$ or $\overrightarrow{x}$ when $n$ is
clear from the context and $A_1\otimes \dots\otimes A_n$ for
$A_1\otimes (\dots \otimes A_n)$ and $A^n$ for $A\otimes \dots \otimes
A$, the $n$-th tensor of $A$. We assume that every top-level type is
closed.

\begin{rem}\label{rem:list}
  In the original paper~\cite{sabry2018symmetric},
  the type for lists of type $A$ is
  defined equationally as $[A] = \one \oplus (A\otimes
  [A])$. In our language, we instead use the inductive type
  constructor $\mu X. A$ to define the type of lists as
  $[A] \triangleq \mu X. \one \oplus (A\otimes X)$. The list constructors become
  \begin{itemize}
  \item $[\,]~\triangleq~\fold{(\inl{()})}$ for the empty list;
  \item $h :: t~\triangleq~\fold{(\inr{\pv{h}{t}})}$ for the addition of a head $h$ to a list $t$.
  \end{itemize}
\end{rem}

\subsubsection{Isos types.}
An iso of type $A\iso B$ acts on terms of
basic types. An iso is a function with inputs of type $A$ and outputs of
type $B$ and defined as a set
of clauses: \isobasique. In each clause, the token $v_i$
stands for an open value while  $e_i$ is an expression. For each $i$,
$v_i$ and $e_i$ share the same variables: the variables of $v_i$ can
be seen as a binder for the free variables of $e_i$.

The construction $\fix g . \isobasique$ represents the creation of a
recursive function. The intuition (enforced by the operational
semantics in \autoref{sec:opsem}) is that $\fix g . \isobasique$ is
identified with $\isobasique$, where $g$ is replaced with the fixed point
$\fix g . \isobasique$.

\subsection{Typing System}
\label{subsec:typing}

The typing system is two-fold. We define typing judgments for values,
patterns, terms and expressions, denoted with $\TypCtxta; \Psi\vdash_e
t : A$, and typing judgments for isos, denoted with $\Psi\entailiso
\omega : A\iso B$. In the judgments, the (linear) contexts $\TypCtxta$
are sets of pairs that consist of a term-variable and a base type,
where each variable can only occur once. The (non-linear) context
$\Psi$ is a set of size at most one. It can only contain a pair of an
iso-variable and an iso-type.

Based on~\cite{sabry2018symmetric}, the typing system is akin to the rules for {\mumall}.
It requires two criteria developed below: (i) that every recursive iso
terminates---we require structural recursion---and (ii) that every
iso is exhaustive and non-overlapping on the left and on the
right---to enforce the totality and bijectivity of isos.

\subsubsection{Typing of Terms and Expressions}
\label{subsubsection:typing}
We say that a term $t$ has type $A$ under contexts $\TypCtxta$ and
$\Psi$ if it can be inferred inductively from the following rules. We
start by first defining the typing rules for values:
\[
  \begin{array}{c}
    \infer{\emptyset;\Psi\entaile () : \one}{}
    \quad
    \infer{x:A;\Psi \entaile x :A}{}
    \quad
    \infer{\TypCtxta;\Psi\entaile \inl{t}:A \oplus B}{\TypCtxta;\Psi\entaile
    t:A}
    \quad
    \infer{\TypCtxta;\Psi\entaile\inr{t}:A \oplus B}{\TypCtxta;\Psi\entaile t:B}
    \mynl
        \infer{
    \TypCtxta_1,\TypCtxta_2;\Psi\entaile \pv{t_1}{t_2} : A \otimes B
    }{
    \TypCtxta_1;\Psi\entaile t_1 : A
    &
      \TypCtxta_2;\Psi\entaile t_2 : B
      }
      \qquad
      \infer{\TypCtxta; \Psi \entaile \fold{t} : \mu X. A }{\TypCtxta; \Psi \entaile t : A[X\leftarrow \mu X. A]}
      \mynl
      \qquad
  \end{array}
\]

\begin{lem}[Linearity of term-variables]
  \label{lemma:lin-term-var}
  Let $v$ be a value and $\TypCtxta; \Psi\vdash_e v : A$, be a valid
  typing judgment. Then all the term-variables of $v$ occur in
  $\TypCtxta$. Moreover, each such variable appears once and only
  once in $v$.
\end{lem}
\begin{proof}
  The proof is done by structural induction on the derivation of
  $\TypCtxta; \Psi\vdash_e v : A$. Notice how in all derivation
  rules, whenever the terms at the root are values, so are the terms
  appearing in the hypotheses. Moreover, notice how for the (only)
  branching rule for the product, the linear context $\TypCtxta$ is
  partitioned into $\TypCtxta_1$ and $\TypCtxta_2$.
\end{proof}

\begin{lem}[Inversion]
  \label{lemma:inversion}
  Given a well-typed value $v$ of type $A$, either $v = x$ or one of
  the following is true:
  \begin{itemize}
    \item $v = ()$ and $A = \one$;
    \item $v = \pv{v_1}{v_2}$ where $v_1$ is of type $A_1$ and $v_2$
    is of type $A_2$ and $A = A_1 \otimes A_2$;
    \item $v = \inl{v_1}$ and $v_1$ is of type $A_1$ and $A = A_1 \oplus A_2$;
    \item $v = \inr{v_2}$ and $v_2$ is of type $A_2$ and $A = A_1 \oplus A_2$;
    \item $v = \fold{v'}$ and $v'$ is of type $A'[X\leftarrow A']$ and
    $A = \mu X. A'$.
  \end{itemize}
\end{lem}
\begin{proof}
  This can be proved directly by case analysis on $v$.
\end{proof}

With this Lemma, we can define a notion of \textit{flattening} of a
 pattern $p$ of type $A$, denoted $|(p, A)|$, which gives
us a list of the variables in $p$ as well as their corresponding
types. The flattening is defined inductively as: $|(x, A)| = ([x],
[A]), |(\pv{p_1}{p_2}, A_1\otimes A_2)| = |(p_1, A_1)| ++ |(p_2,
A_2)|$ where $++$ is the pointwise list concatenation.

We can now define the typing derivation for general terms, where
in the typing rule for the \lett, we have that $|(p, A)| = ([x_1, \dots,
x_n], [A_1, \dots, A_n])$.
By abuse of notation we shall later
identify $|(p,A)|$ with the corresponding typing context $x_1:A_1,\ldots x_n:A_n$.
Note how this imposes that $p$ is of type $A$ in the context
$|(p,A)|$.

\[
  \begin{array}{c}
  \infer{
      \TypCtxta;\Psi\entaile \isoterm~t : B
      }{
      \Psi \entailiso \isoterm : A \iso B
    &
      \TypCtxta;\Psi\entaile t : A
      }
      \mynl
      \infer{\TypCtxta_1,\TypCtxta_2;\Psi \entaile \letv{p}{t_1}{t_2} :
      B}{\TypCtxta_1;\Psi\entaile t_1 : A \qquad\TypCtxta_2, x_1 : A_1, \dots, x_n : A_n; \Psi\entaile t_2 : B}
      \mynl
  \end{array}
\]

Note how the context $\TypCtxta$ is used linearly while the context
$\Psi$ is not. By abuse of notation, we will write $\TypCtxta \vdash_e t : A$
for $\TypCtxta ; \emptyset \vdash_e t : A$

\subsubsection{Exhaustivity, non-overlapping}
In order to apply
an iso to a term, the iso must be of type $A\iso B$ and the argument of
type $A$. Since we want our language to represent
\emph{isomorphisms} we impose isos of shape {\isobasique} and
type $A\iso B$ to be exhaustive and non-overlapping.
\emph{Exhaustivity}, means that the expressions on the left (resp. on
the right) of the clauses describe all possible values for the type
$A$ (resp. the type $B$). \emph{Non-overlapping} means that two
expressions cannot match the same value. For instance, the left and
right injections $\inl{v}$ and $\inr{v'}$ are non-overlapping while a
variable $x$ is always exhaustive. In order to formally define these
two notions, we first need to characterize \emph{pattern-matching}, as follows.

\begin{defi}[Pattern-Matching]
  \label{def:pattern-matching}\label{tab:pattern_matching} We say that
a value $v$ matches against a value $v'$ if there exists a
substitution $\sigma$, that is, a mapping from variables to (closed)
values such that $v$ under the substitution $\sigma$ is equal to $v'$.
This is noted $\sigma[v] = v'$ and is defined inductively over the
following rules.
\[
\infer{\sigma[\inl{e}] = \inl{e'}}{\sigma[e] = e'}
\quad
\infer{\sigma[\inr{e}] = \inr{e'}}{\sigma[e] = e'}
\quad
\infer{\sigma[x] = e}{\sigma = \{ x \mapsto e\}}
\quad
\infer{\sigma[\fold{e}] = \fold{e'}}{\sigma[e] = e'}
\]
\[
\infer{
\sigma[\pv{e_1}{e_2}] = \pv{e'_1}{e'_2}
}{
\sigma_1[e_1] = e'_1
&
\sigma_2[e_2] = e'_2
&
\operatorname{supp}(\sigma_1) \cap \operatorname{supp}(\sigma_2) = \emptyset
&
\sigma = \sigma_1\cup\sigma_2
} \quad \infer{\sigma[()] = ()}{}
\]
The \textit{support} of a substitution
$\sigma$ is defined as $\operatorname{supp}(\sigma) = \{x \mid (x
\mapsto v) \in \sigma \}$.
\end{defi}

\begin{defi}[Exhaustivity]
  \label{def:exhaustivity}
  We call a set $S$ of well-typed values of type $A$ (where each
  value can be typed under a different context) \emph{exhaustive} if,
  for every closed values $v$ of type $A$ there exists $v' \in S$ and
  $\sigma$ such that $\sigma[v'] = v$.
\end{defi}

\begin{defi}[Non-Overlapping]
  \label{def:nonoverlap}
  We call a set $S$ of values of type $A$ (where each value can be
  typed under a different context) \emph{non-overlapping} if, for
  every closed values $v$ of type $A$ there exists at most one $v' \in
  S$ and $\sigma$ such that $\sigma[v'] = v$.
\end{defi}

In order to define the typing of isos in \autoref{sec:typiso}, we
define the predicate $\PE_A$, used to ensure that isos are exhaustive
and non-overlapping, and that they indeed represent isomorphisms.

\pagebreak 
\begin{defi}[Orthogonal Decomposition]
  \label{def:OD}
  We say that a list of well-typed values $S$ of type
  $A$ satisfies $\PE_A$, denoted $\PE_A(S)$, if it can be inferred
  inductively from the following rules.
  \[
  \scalebox{.9}{\text{$\begin{array}{c}
  \infer[\PE\text{-}var]{\PE_A{(\set{x})}}{} \qquad \infer[\PE\text{-}$\one$]{\PE_\one{(\set{()})}}{} \qquad
  \infer[\PE\text{-}\oplus]{\PE_{A\oplus B}(\{\inl{v} \alt v \in S \} \cup \{\inr{v} \alt v\in T \})}{\PE_A(S)\qquad \PE_B(T)} \qquad
  \mynl
  \infer[\PE\text{-}\mu]{\PE_{\mu X. A}{(\set{\fold{v} \alt v \in S})}}{\PE_{A[X \leftarrow \mu X. A]}(S)} \qquad
  \infer[\PE\text{-}\otimes_1]{\PE_{A\otimes B}(S)}{
  \PE_{A}(\pi_1(S)), \forall v\in \pi_1(S), \PE_B(S^1_v)}
  \mynl
  \infer[\PE\text{-}\otimes_2]
  {\PE_{A\otimes B}(S)}
  {\PE_{B}(\pi_2(S)), \forall v\in \pi_2(S), \PE_A(S^2_v)}
  \end{array}$}}
  \]
  In the rules $\PE\text{-}\otimes_1$ and $\PE\text{-}\otimes_2$,
  we have $S = \{\pv{v_1}{v_1'}, \dots, \pv{v_n}{v'_n}\}$ and the
  sets $\pi_1(S)$ and $\pi_2(S)$ are respectively $\set{v \mid
  \pv{v}{w} \in S}$ and $\set{w\mid\pv{v}{w}\in S}$. The sets $S_v^1$
  and $S_v^2$ are respectively $\set{w \mid \pv{v}{w}\in S}$ and
  $\set{w \mid \pv{w}{v} \in S}$.
\end{defi}

\begin{rem}
  \label{remark:cnot-od}
  The definition of $\PE_A(S)$ is inspired
  from~\cite{sabry2018symmetric}. The main difference is the rule for
  the tensor, and the new rule $\PE\text{-}\mu$ for inductive types.
  The problem in the original definition comes from the impossibility
  to type important isos. Indeed, the system in
  \cite{sabry2018symmetric} has the following rule for the tensor:
  \begin{equation}
   \label{eq:buggyrule}
    \infer{\PE_{A\otimes
        B}\set{\pv{v_1}{v_2} \mid v_1 \in S_1, v_2\in S_2, \FV(v_1) \cap
        \FV(v_2) = \emptyset}}{\PE_A(S_1) \qquad \PE_B(S_2)}.
  \end{equation}
  Consider the following encoding of the so-called Toffoli gate:
  \[
    \left\{
      \begin{array}{ccc}
      \pv{\inl{()}}{x}    & {\iso} & \pv{\inl{()}}{x}\\
      \pv{\inr{()}}{\inl{()}}    & {\iso} & \pv{\inr{()}}{\inr{()}}\\
      \pv{\inr{()}}{\inr{()}}    & {\iso} & \pv{\inr{()}}{\inl{()}}
      \end{array}
    \right\} : (\one\oplus\one)^2 \iso (\one\oplus\one)^2
  \]
  This iso does not satisfy the criterion shown in
  \autoref{eq:buggyrule}. Indeed, the variable $x$ on the
  right-hand side of the pair overlap with the values $\inl{()}$ and
  $\inr{()}$. Meanwhile, this set of clauses satisfy our definition of
  $\PE$.
\end{rem}

We can already show that $\PE$ is sound:

\begin{lem}[Soundness of $\PE_A(S)$]~
  \label{lem:od-ok}
  Given a set $S$ of well-typed values of type $A$ such that
  $\PE_A(S)$ holds, then $S$ is exhaustive and non-overlapping, i.e.
  for all closed values $v$ of type $A$, there exists a unique $v'\in
  S$ and a unique $\sigma$ such that $\sigma[v'] = v$.
\end{lem}
\begin{proof}
  By induction on a derivation of $\PE_A(S)$:
  \begin{itemize}
  \item $\PE_A(\{x\})$ is direct, and $\PE_\one(\set{()})$
    follows from \autoref{lemma:inversion}.

  \item Assume that the root of the derivation is
      \[\infer[\PE\text{-}\oplus]{\PE_{A\oplus B}(\{\inl{v}
          \alt v \in S_A \} \cup \{\inr{v} \alt v\in S_B
          \})}{\PE_A(S_A)\qquad \PE_B(S_B)}.\]
      In this case, the closed value $v$ is of type $A\oplus B$.
      By~\autoref{lemma:inversion}
      we know that either $v = \inl{\tilde{v}}$ where $\tilde{v}$ is a
      closed value of type $A$ or $v = \inr{\tilde{v}}$ where
      $\tilde{v}$ is a closed value of type $B$. Both cases are proved similarly:
      let us focus on the
      first case. By induction hypothesis on $\PE_A(S_A)$ we know that
      there exists a unique $v' \in S_A$ and a unique $\sigma'$ such
      that $\sigma'[v'] = \tilde{v}$. Therefore, we know that
      $\sigma'[\inl{v'}] = \inl{\tilde{v}}$. By definition of the
      pattern-matching, we know that there is no term of the shape
      $(\inr{-})$ and no substitution $\delta$ such that
      $\delta[\inl{v'}] = (\inr{-})$, therefore $\inl{v'}$ is the only
      value in $(\{\inl{v} \mid v \in S \} \cup \{\inr{v} \mid v\in T 
      \})$ that matches $\inl{\tilde{v}}$ with the unique substitution
      $\sigma'$.
      \item Assume that the root of the derivation is
      \[\infer[\PE\text{-}\mu]{\PE_{\mu X. A}{(\set{\fold{v} \alt v
              \in S})}}{\PE_{A[X \leftarrow \mu X. A]}(S)}.
      \]
      By~\autoref{lemma:inversion} we know that $v = \fold{\tilde{v}}$
      where $\tilde{v}$ is a closed value of type $A[X \leftarrow \mu
      X. A]$. By induction hypothesis on $\PE_{A[X\leftarrow \mu X.
      A]}(S)$, we know that there exists a unique $v' \in S$ and a
      unique $\sigma'$ such that $\sigma'[v'] = \tilde{v}$. Therefore,
      we know that $\sigma'[\fold{v'}] = \fold{\tilde{v}}$. Since by
      induction hypothesis we know that $v'$ and $\sigma'$ are unique
      for $\tilde{v}$, then $\fold{v'}$ and $\sigma'$ are also unique
      for $\fold{\tilde{v}}$.
      \item Assume that the root of the derivation is
          \[
            \infer[\PE\text{-}\otimes_1]{\PE_{A\otimes
                B}(S)}{ \PE_{A}(\pi_1(S)), \forall v\in \pi_1(S),
              \PE_B(S^1_v)}.
          \]
      Recall that $S = \set{\pv{v_1}{v_1'}, \dots,
      \pv{v_n}{v_n'}}$. By~\autoref{lemma:inversion} we know that $v =
      \pv{\tilde v}{\tilde{v}'}$, where $\tilde v$ is a closed value
      of type $A$ and $\tilde{v}'$ is a closed value of type $B$. By
      induction hypothesis on $OD_A(\pi_1(S))$ we know that there
      exists a unique $v_i$ (for $i\in\set{1,\dots, n}$) and a
      unique $\sigma_i$ such that $\sigma_i[v_i] = \tilde v$.
      Similarly, there exists a unique $w$ in $S^1_{v_i}$ and a unique
      $\sigma_w$ such that $\sigma_w[w] = \tilde{v}'$. Due to the
      fact that the free variables of well-typed values are all
      distinct (by \autoref{lemma:lin-term-var}) the support of
      $\sigma_i$ and $\sigma_w$ are disjoint. We can then conclude
      that $(\sigma_i \cup \sigma_w)[\pv{v_i}{w}] = \pv{\tilde
      v}{\tilde{v}'}$.
      \item The case for $\PE\text{-}\otimes_2$ is similar. \qedhere
  \end{itemize}
\end{proof}

\begin{rem}
While we have shown that the definition of the $\PE$ predicate is
sound, we can notice that it is non-complete: there exists some set of
clauses that are exhaustive and non-overlapping, but that do not match
the $\PE$ predicate. For instance, the following set of clauses of
type $(A \oplus B)\otimes (C\oplus (D \oplus E))$:
\begin{center}
    $\left\{\begin{array}{lll}
      \pv{x}{\inl{(y)}} & \iso & \dots \\
      \pv{\inl{(x)}}{\inr{(\inl{(y)})}} & \iso & \dots \\
      \pv{\inl{(x)}}{\inr{(\inr{(y)})}} & \iso & \dots \\
      \pv{\inr{(x)}}{\inr{(y)}} & \iso & \dots \\
    \end{array}\right\}$
  .
\end{center}

While this set of clauses is indeed exhaustive and non-overlapping and
would be accepted in a typical functional programming language such as
OCaml or Haskell, it does not check the $\PE$ predicate: both
$\PE\text{-}\otimes_1$ and $\PE\text{-}\otimes_2$ are not satisfied.
In the first case, the clauses $x, \inl{(y)}$ and $\inr{(y)}$ are
overlapping, and similarly for $\PE\text{-}\otimes_2$. However, for
each closed value of type $(A \oplus B)\otimes (C\oplus (D \oplus
E))$, there is a single clause that matches.
Although this could be considered as a limitation, for the purpose of
the paper it is expressive enough, in the sense of \autoref{sec:rpp}.
\end{rem}

\subsubsection{Typing of Isos}
\label{tab:typisos}\label{sec:typiso}
We can now define the typing of isos. An iso will be required to
satisfy $\PE$ both on the left and the right-hand sides of the
set of clauses. Moreover, since we allow isos to be recursive we need
to make sure that on a given input an iso always terminates. Indeed,
otherwise it would not necessarily be a total function. For that we
introduce a notion of \emph{structural recursion}, stating that a
recursive call can only be used on a strict subterm of the input.

Formally, we say that an iso $\omega$ has type $T$ under some
context $\Psi$, denoted $\Psi\entailiso \omega$, if it can be inferred
inductively by the following rules:
\[\def\mynl{\\[2ex]}
      \begin{array}{c}
        \infer{f : A\iso B \entailiso f : A\iso B}{}
        \mynl
        \infer{
        \Psi \entailiso
        \isobasique : A \iso B.
        }{
        \begin{array}{l@{\quad}l@{\quad}l@{\qquad}l}
          \TypCtxta_1 \entaile v_1 : A
          &
            \ldots
          &
            \TypCtxta_n\entaile v_n : A
          &
            \PE_A(\{v_1, \dots, v_n\})   \\
          \TypCtxta_1;\Psi\entaile e_1 : B
          &
            \ldots
          &
            \TypCtxta_n;\Psi\entaile e_n : B
          & \PE_B(\{\Val{e_1},\ldots,\Val{e_n}\})
        \end{array}
            }
        \mynl
        \inferrule{
          f: A\iso B\entailiso \isobasique : A\iso B \quad \fix f. \omega\text{ is structurally recursive}
        }
        {
          \Psi\entailiso \fix f.\isobasique : A\iso B
        }
      \end{array}
\]
There, $\Val e$ is
defined as $\Val{\letv{p}{\omega~p'}{e}} = \Val e$, and $\Val v = v$
otherwise.
In the second rule, the term variables of the $\Theta_1, \dots,
\Theta_n$ are bound by the pattern-matching construction: they are
not visible outside the iso, thus not appearing anymore in the \linebreak typing     
context.

In the last rule, we furthermore ask $\fix f. \omega$ to be
\emph{structurally recursive} :

\begin{defi}[Structurally Recursive]
  An iso $\fix f. \isobasique : A \iso B$ is said to be structurally
  recursive when $A = A_1 \otimes \dots \otimes A_m$ and
  $B = B_1 \otimes \dots \otimes B_l$, and when $A_j = \mu X. C_1$ and
  $B_k = \mu X. C_2$ for some $1 \leq j \leq m$ and $1 \leq k \leq
  l$. Moreover, we require that for all $i\in\{1,\dots,n\}$ the value
  $v_i$ is of the form $(v^{1}_i, \dots, v^{m}_i)$ and that
  $v^{j}_i$ \linebreak is either:    
  \begin{itemize}
  \item A closed value, in which case $f$ does not occur in $e_i$ and
  $\Val{e_i} = ({v'}_i^1, \dots, {v'}_i^l)$ is such that ${v'}_i^k$ is
  also a closed value.
  \item An open value, in which case for all subterms of the form
    $(\letv{p'}{f~p}{\dots})$ in $e_i$ we have
    $p' = (y_1, \dots, y_l)$ and $p = (x_1, \dots, x_m)$ where
    $x_j : \mu X. B$ is a strict subterm of $v_i^j$. We also ask that
    $\Val(e_i) = ({v'}_i^1, \dots, {v'}_i^m)$, where ${v'}_i^k$
    contains the value $y_k$ as a strict subterm.
  \end{itemize}
  Finally, we ask that there is at least one clause where $v^{j}_i$ is
  a closed value. We call the value $v^{j}_i$ (resp. the variable
  $x_j$) the \emph{decreasing argument} (resp. the \emph{focus}) of
  the structurally recursive criterion.
\end{defi}

\begin{rem}
  \label{rem:one-iso-var}
  Note that, given a well typed closed term $\emptyset ; \emptyset
  \vdash t : A$, along its typing derivation, the iso-context will be
  either empty or a singleton.
\end{rem}

\begin{rem}
  We impose a simple notion of structural recursion: the typing rules
  of isos allow to have at most one iso-variable in the context.
  Indeed, the last typing rule of isos erases the context
  $\Psi$ and replaces it with the new context containing only a single
  iso-variable. Also, we cannot have intertwined recursive calls.
    This mean that for an iso of the form
    $
      \fix f. \isobasique,
    $
    if $e_i$ is
    \[
        \begin{array}{lcl}
          & & \letv{p_{i,1}}{\omega_{i,1}~p'_{i,1}}{} \\
            & & \dots \\
            & & \letv{p_{i,n}}{\omega_{i,n}~p'_{i,n}}{v'_i},
        \end{array}
    \]
    for
    each $j\in \set{1, \dots, n}$, either $\omega_{i,j} = f$ or
    $f$ does not appear at all inside $\omega_{i,j}$.
\end{rem}

\begin{exa}
  \label{ex:iso1}
  We can define the iso of type : $A\oplus (B\oplus C)\iso C\oplus
  (A\oplus B)$, where $a, b$ and $c$ are variables as:
  \[
    \scalebox{1}{$\left\{
      \begin{array}{l@{~}c@{~}l}
        \inl{(a)} & {\iso} & \inr{(\inl{(a)})} \\
        \inr{(\inl{(b)})} & {\iso} & \inr{(\inr{(b)})} \\
        \inr{(\inr{(c)})} & {\iso} & \inl{(c)} \\
      \end{array}
    \right\}$}
  \]
  The first clause is typed under context $\set{a : A}$, the second
  clause under context $\set{b : B}$ and the third clause under
  context $\set{c : C}$. Notice that the $\PE$ predicate holds for
  both the right and left-hand side clauses.
\end{exa}

\begin{exa}\label{ex:map}
  Consider a (closed) iso $\vdash \omega : A\iso B$, and recall the
  list construction of \autoref{rem:list}. Let us define the operation \textit{map}$(\omega) : [A]\iso[B]$ as follows.
  \[
    \textit{map}(\omega):[A]\iso[B]
    = \fix f.\left\{
      \begin{array}{l@{~}c@{~}l}
        [~] & {\iso} & [~] \\
        h :: t & {\iso} & \letv{h'}{\omega~h}{} \\
         & &  \letv{t'}{f~t}{}\\
        & & h' :: t' \\
      \end{array}
    \right\}
  \]
  Note how the iso is indeed structurally recursive.
  Note also how the left and right-hand side of the $\iso$ respect both the
  criteria of exhaustivity---every value of each type is being
  covered by at least one expression---and non-overlapping---no two
  expressions cover the same value.
\end{exa}

\begin{exa}
  There are of course fixed points that do not respect the structural
  recursive constraint, such as e.g.
  $\fix f. \set{x\iso \letv{y}{f~x}{y}}$.
\end{exa}

\subsection{Operational Semantics}
\label{sec:opsem}

From now on, we will only consider well-typed terms.
Our language is equipped with a rewriting system $\to$ on
terms, that follows a deterministic call-by-value strategy: each
argument of a function is fully evaluated before applying the
$\beta$-reduction. This is done through the use of an evaluation
context $C[]$, which consists of a term with a hole (where $C[t]$ is
$C$ where the hole has been filled with $t$). Due to the deterministic
nature of the strategy we directly obtain the unicity of the normal
forms. The evaluation of an iso applied to a value relies on
pattern-matching, as discussed in \autoref{tab:pattern_matching}.
Because we ensure exhaustivity and
non-overlapping (\autoref{lem:od-ok}), the pattern-matching always succeeds
on closed, well-typed values.

\begin{defi}[Substitution]
  \label{def:substitution}
  Applying a substitution $\sigma$ to an expression $t$, written
  $\sigma(t)$, is defined as: $\sigma(()) = (), \sigma(x) = v$ if
  $\{x\mapsto v\}\subseteq \sigma, \sigma(\inr{(t)}) =
  \inr{(\sigma(t))}, \sigma(\inl{(t)}) = \inl{(\sigma(t))},
  \sigma(\fold{(t)}) = \fold{(\sigma(t))}, \sigma(\pv{t}{t'}) =
  \pv{\sigma(t)}{\sigma(t')}, \sigma(\omega~t) = \omega~\sigma(t)$ and finally,
  $\sigma(\letv{p}{t_1}{t_2}) = (\letv{p}{\sigma(t_1)}{\sigma(t_2)})$
  where the variables from $p$ do not occur in the domain of $\sigma$
  (this assumption can always be fulfilled thanks to
  $\alpha$-renaming).
\end{defi}

\begin{rem}
  Notice the difference between $\sigma[v] = t$ and $\sigma(t)$. The
  first one defines the fact that $v$ and $t$ matches under the
  substitution $\sigma$, while the second one defines the actual
  substitution. In fact if $\sigma[v] = t$ then $\sigma(v) = t$.
  Notice that if $\sigma[v] = t$ then the support of $\sigma$ is the
  set of free variables of $v$.
\end{rem}

As we mentioned, the rewriting system is defined through the use of
evaluation contexts and with the help of the pattern-matching and
substitution from~\autoref{def:pattern-matching}
and~\autoref{def:substitution} respectively.

\begin{defi}[Evaluation Contexts]
  The evaluation contexts $C$ are defined as:
  $$C ::= [~] \alt \inl{C} \alt \inr{C} \alt \omega~C \alt
\letv{p}{C}{t} \alt \pv{C}{t} \alt \pv{v}{C} \alt \fold{C}$$
\end{defi}

\begin{defi}[Evaluation relation  $\to$]
  \label{def:rewriting-system}
  We define $\to$ the rewriting system of our language as follows:
    \[
    \begin{array}{c}
    \infer[\mathrm{Cong}]{C[t_1] \to C[t_2]}{t_1 \to t_2}
    \quad
    \infer[\mathrm{IsoRec}]{
      (\fix f.\isoterm)~v \to (\omega[f := (\fix f. \omega)])~v
    }{
    }\\[1.5ex]
    \infer[\mathrm{LetE}]{\letv{p}{v}{t} \to \sigma(t)}{\sigma[p] = v}
    \quad
    \infer[\mathrm{IsoApp}]{ \isobasique~v \to \sigma(e_i)}{
      \sigma[v_i] = v}
    \end{array}
    \]
    As usual, we denote $\to^*$ for the reflexive transitive closure of
    $\to$.
\end{defi}

\subsection{Property of the typed language}

The language features the standard properties of typed languages,
namely progress and subject reduction. As there are two classes of
variables: term and iso variables, we have two substitution lemmas.

\begin{lem}[Substitution Lemma for term variables]
  \label{lemma:sub_var}
  Assume that for all $i$, we have $\TypCtxta_i;\Psi\vdash_e v_i : A_i$.
  Furthermore, assume that $\TypCtxtb, x_1 :
  A_1, \dots, x_n : A_n;\Psi \vdash_e t : B$.
  Then, provided that
  $\sigma = \{x_1 \mapsto v_1, \dots, x_n \mapsto v_n\}$ we have
  $\TypCtxtb, \TypCtxta_1, \dots, \TypCtxta_n;\Psi \vdash_e \sigma(t) : B$.
\end{lem}

\begin{proof}
  By induction on $t$.
  \begin{itemize}
  \item
  Case $x$. We have $n=1$ and $x_1 = x$ and $\TypCtxtb =
  \emptyset$, and we have $\sigma = \{x_1 \mapsto v_1\}$ for some $v_1$
  of type $B$ under some context $\TypCtxta_1$, then we get
  $\TypCtxta_1 ; \Psi \vdash_e \sigma(x) : B$ which leads to
  $\TypCtxta_1 ; \Psi \vdash_e v : B$ which is typable by our
  hypothesis.

  \item Case $()$. The result follows immediately from the fact that
  $\Sigma$ is empty, $n=0$ and $\sigma(()) = ()$.

  \item Case $\inl{t'}$. By induction hypothesis on $t'$ we know that
  $\TypCtxtb,\TypCtxta_1,\dots, \TypCtxta_n;\Psi \vdash_e \sigma(t') :
  A$ is typable and therefore by applying the typing rule for the
  $\inl{-}$, and by definition of the substitution, we get that
  $\TypCtxtb,\TypCtxta_1,\dots, \TypCtxta_n;\Psi \vdash_e
  \sigma(\inl{t'}) : A\oplus B$ is typable.

  \item Case $\inr{t'}, \fold{t'}, \omega~t'$ are similar.

  \item Case $\pv{t_1}{t_2}$, by typing we get that we can split
  $\TypCtxtb$ into $\TypCtxtb_1, \TypCtxtb_2$ and the variables $x_1,
  \dots, x_n$ are split into two parts for typing both $t_1$ or $t_2$
  depending on whenever or not they occur freely in $t_1$ or
  $t_2$, w.l.o.g. say that $x_1, \dots, x_l$ are free in $t_1$ and
  $x_{l+1}, \dots, x_n$ are free in $t_2$ then we get:

  $\infer{\TypCtxtb_1,\TypCtxtb_2, x_1: A_1, \dots, x_l: A_l,
  x_{l+1}:A_{l+1}, \dots, x_n:A_n;\Psi\vdash_e \pv{t_1}{t_2} : B_1\otimes
  B_2}{\TypCtxtb_1, x_1:A_1,\dots, x_l:A_l;\Psi \vdash_e t_1 : B_1 \qquad
  \TypCtxtb_2, x_{l+1} : A_{l+1}, \dots, x_n : A_n;\Psi \vdash_e t_2 : B_2}$

  By substitution, we get that $\sigma(\pv{t_1}{t_2}) =
  \pv{\sigma(t_1)}{\sigma(t_2)}$, so we get the following typing
  derivation which is completed by induction hypothesis on the
  subterms:

  $\infer{\TypCtxtb_1,\TypCtxtb_2, \TypCtxta_1, \dots, \TypCtxta_l,
  \TypCtxta_{l+1}, \dots, \TypCtxta_n;\Psi\vdash_e
  \pv{\sigma_1(t_1)}{\sigma_2(t_2)} : B_1\otimes
  B_2}{\TypCtxtb_1, \TypCtxta_1, \dots, \TypCtxta_l;\Psi \vdash_e
  \sigma_1(t_1) : B_1 \qquad \TypCtxtb_2, \TypCtxta_{l+1}, \dots,
  \TypCtxta_n;\Psi \vdash_e \sigma_2(t_2) : B_2}$

  \item Case $\letv{p}{t_1}{t_2}$ Similar to the case of the
  tensor.\qedhere
  \end{itemize}
\end{proof}

\begin{lem}[Substitution Lemma for Isos]
\label{lemma:substitution_isos}
The following substitution properties hold:

\begin{itemize}
  \item If $\TypCtxta; f : T_1 \vdash_e t : A$ and $g : T_2
  \entailiso \omega : T_1$ then $\TypCtxta; g : T_2 \vdash_e t[f
  \leftarrow \omega] : A$.

  \item If $f : T_1 \entailiso \omega_1 : T_2$
      and $h : T_3 \entailiso \omega_2 : T_1$
  then $h : T_3 \entailiso \omega_1[f \leftarrow \omega_2] : T_2$.

\end{itemize}

\end{lem}
\begin{proof}

  We prove those two propositions by mutual induction on $t$ and
  $\omega_1$.

  \textbf{Terms}, by induction on $t$.

  \begin{itemize}
      \item If $t = x$ or $t = ()$ then in the first case
      $\TypCtxta = x : A$ and in the second case $\TypCtxta =
      \emptyset$ and $A = \one$ by~\autoref{lemma:inversion}, and in
      both cases we have that $t[f\leftarrow \omega] = t$. Therefore,
      we have to type $x : A ; g : T_2 \vdash x : A$ in one case and
      $\emptyset ; g : T_2 \vdash_e () : \one$ in the other. Both are
      possible by definition of the typing system.

      \item If $t = \inl{t'}$ or $\inr{t'}$ or $\fold{t'}$ or
      $\pv{t_1}{t_2}$ or $\letv{p}{t_1}{t_2}$, then similarly to the
      proof of~\autoref{lemma:sub_var} the substitution goes to
      the subterms and we can apply the induction hypothesis.

      \item If $t = \omega'~t'$. In that case, the substitution goes
      to both subterms: $t[f \leftarrow \omega] = (\omega'[f
      \leftarrow \omega])~(t'[f \leftarrow \omega])$. We
      can then conclude by induction hypothesis on $t'$ and by the
      induction hypothesis on isos.
  \end{itemize}

  \textbf{Isos}, by induction on $\omega_1$.

  \begin{itemize}
      \item If $\omega_1 = f$, then we get $h : T_3 \entailiso f[f
      \leftarrow \omega_2] : T_2$ which is typable by hypothesis.

      \item If $\omega_1 = g \not= f$ is impossible by our typing
      hypothesis.

      \item If $\omega_1 = \fix g. \omega$, then by typing $f$ does
      not occur in $\omega_1$ so nothing happens.

      \item If $\omega_1 = \isobasique$, then, by definition of the
      substitution we have that \begin{align*}&\isobasique[f \leftarrow
      \omega_2] \\ &= \{v_1[f \leftarrow \omega_2] \iso e_1[f \leftarrow \omega_2] \alt \dots \alt
      v_n[f \leftarrow \omega_2] \iso e_n[f \leftarrow \omega_2]\} \\
      &=  \{v_1 \iso e_1[f \leftarrow \omega_2] \alt \dots \alt v_n
      \iso e_n[f \leftarrow \omega_2] \}\end{align*} in which case we
      apply the induction hypothesis of isos on terms.
      \qedhere
  \end{itemize}
\end{proof}

We can then deduce subject reduction and progress, as follows.

\begin{lem}[Subject Reduction]\label{prop:subject_reduction} If
$\TypCtxta;\Psi\vdash_e t : A$ and $t\rightarrow t'$ then
$\TypCtxta;\Psi\vdash_e t' : A$.
\end{lem}
\begin{proof}
   By induction on $t\to t'$ and direct by~\autoref{lemma:sub_var}
   and~\autoref{lemma:substitution_isos}
\end{proof}

\begin{lem}[Progress]\label{prop:progress} If $~\vdash_e t : A$
then, either $t$ is a value, or $t\to t'$.
\end{lem}
\begin{proof}
    Direct by induction on $\vdash_e t : A$. The two possible
    reduction cases, $\omega~v$ and $\letv{p}{v}{t}$ always reduce by
    typing, pattern-matching and by~\autoref{lem:od-ok}.
\end{proof}

\subsection{Inversion}

One thing to note is that, in many models of reversible
computing, such as Reversible Turing
Machines~\cite{bennett1973logical, rtm1, rtm2}, assembly
code~\cite{vieri1995pendulum}, imperative or functional programming
languages~\cite{GLUCK2023113429, yokoyama2011reversible}, each step of
the evaluation is reversible. It is a local property. This is not the
case in our language. When a term $t$ reduces to another term $t'$,
one cannot directly inverse the rewriting system to go from $t'$ to
$t$. Instead, reversibility should be understood in the broader
context of the design of the language, in which any isos can be
inverted. Therefore, given an iso $\omega : A\iso B$ one can build its
inverse $\omega^\bot : B\iso A$. The inverse operation is defined
inductively on $\omega$ and is given in~\autoref{def:iso-inv}.
\begin{defi}[Inversion]
  \label{def:iso-inv}
  Given an iso $\omega$, we define its dual $\omega^\bot$ as: $f^\bot
  = f, (\fix f. \omega)^\bot = \fix f. \omega^\bot, \{(v_i\iso
  e_i)_{i\in I}\}^\bot = \{((v_i\iso e_i)^\bot)_{i\in I}\}$ and the
  inverse of a clause as:
\[\begin{array}{c} \left(
      \begin{array}{l@{~}c@{~}l} v_1&{\iso}&{\tt
        let}\,p_1=\isoterm_1\,p'_1\,{\tt in} \\
           && \cdots \\
           && {\tt let}\,p_n=\isoterm_n\,p'_n\,{\tt in}~v'_1
      \end{array}
      \right)^\bot := \left(
      \begin{array}{lcl@{}l@{}l} v'_1&{\iso}&{\tt
        let}\,p'_n&=\isoterm_n^{\bot}&\,p_n\,{\tt in} \\
           && \cdots \\
           && {\tt let}\,p'_1&=\isoterm_1^{\bot}&\,p_1\,{\tt in}~v_1
      \end{array}
      \right).
    \end{array}
  \]

\end{defi}

We can show that the inverse is well-typed and behaves as expected:
\begin{lem}[Inversion is well-typed]
  \label{lem:inv-type}
  If $f : C \iso D \entailiso \omega : A\iso B$, then
  $f : D \iso C \entailiso\omega^\bot : B\iso A$.
\end{lem}

\begin{proof}
  The proof is done by structural induction on the typing derivation
  of $f : C \iso D \entailiso \omega : A\iso B$.
  \begin{itemize}
  \item
  The case where $\omega =
  f$ is direct.
\item
  For the case where the root of the typing derivation  is $\omega = \{v_1~\iso~e_1 \alt \dots
  \alt v_m~\iso~e_m\}$,
  recall the typing rule of isos:
  \[
    \infer{
      f : C\iso D \entailiso
      {\{v_1~\iso~e_1 \alt \dots \alt v_m~\iso~e_m\}} : A \iso B.
    }{
      \begin{array}{l@{\quad}l@{\quad}l@{\quad}l}
        \TypCtxta_1 \entaile v_1 : A
        &
          \ldots
        &
          \TypCtxta_m\entaile v_m : A
        &
          \PE_A(\{v_1, \dots, v_m\})   \\
        \TypCtxta_1; f : C\iso D\entaile e_1 : B
        &
          \ldots
        &
          \TypCtxta_m; f : C\iso D\entaile e_m : B
        & \PE_B(\{\Val{e_1},\ldots,\Val{e_m}\})
      \end{array}
    }
  \]
  First, notice that the predicate $\PE$ still holds for
  $\omega^\bot$ as values do not change. We then need to check that if
  one clause is typable, then its dual is also typable. Without lost
  of generality we consider the case of the first clause, the other being
  similar:
\[\begin{array}{c} \left(
  \begin{array}{l@{~}c@{~}l} v_1&{\iso}&\letv{p_1}{\omega_1~p_1'}{} \\
          && \cdots \\
          && \letv{p_n}{\isoterm_n~p'_n}{v'_1}
  \end{array}
  \right)^\bot := \left(
  \begin{array}{lcl@{}l@{}l} v'_1&{\iso}&\letv{p'_n}{\isoterm_n^{\bot}~p_n}{} \\
          && \cdots \\
          && \letv{p'_1}{\isoterm_1^{\bot}~p_1}{v_1}
  \end{array}
  \right).
  \end{array}
\]
  Without loss of generality, using \autoref{rem:Barendregt} we can
  assume that all of the term variables introduced in the $p_i$ are
  fresh.
  By typing we know that $\TypCtxta_1;f:C\iso D \vdash_e v_1 : A$ and
  \[
    \TypCtxta_1;
    f : C\iso D \vdash_e \letv{p_1}{\omega_1~p_1'}{\dots v_1'} : B.
  \]
  Writing $\Psi$ for $f : C\iso D$, a typing derivation for the latter, written $\pi_1$,
  starts with
  \[
    \infer{
      \TypCtxta'_1,\TypCtxta''_1;\Psi
      \vdash_e \letv{p_1}{\omega_1~p_1'}{\dots v_1'} : B
    }{
      \infer{
        \TypCtxta'_1;\Psi
        \vdash_e \omega_1~p_1':B_1
      }{
        \infer*{\Psi \vdash_i \omega_1 : A_1\iso B_1}{\pi_{\omega_1}}
        &
        \infer*{\TypCtxta'_1;\Psi\vdash_e p_1' : A_1}{\pi_{p_1'}}
      }
      &
      \infer*{
        \TypCtxta''_1,\Sigma_1;\Psi
        \vdash_e \letv{p_2}{\omega_2~p_2'}{\dots v_1'} : B
      }{
        \pi_2
      }
    }
  \]
  where $\TypCtxta_1$ is decomposed as
  $\TypCtxta'_1,\TypCtxta''_1$. By linearity, we have
  $|(p'_1,A_1)| = \TypCtxta'_1$, and $|(p_1,B_1)| = \Sigma_1$.
  At each level $1<i<n$, the typing derivation $\pi_i$ is similar:
  \[
    \infer{
      \TypCtxta''_{i-1},\Sigma_{i-1};\Psi
      \vdash_e \letv{p_i}{\omega_i~p_i'}{\dots v_1'} : B
    }{
      \infer{
        \TypCtxta'_i;\Psi
        \vdash_e \omega_i~p_i':B_i
      }{
        \infer*{\Psi \vdash_i \omega_i : A_i\iso B_i}{\pi_{\omega_i}}
        &
        \infer*{\TypCtxta'_i;\Psi\vdash_e p_i' : A_i}{\pi_{p_i'}}
      }
      &
      \infer*{
        \TypCtxta''_i,\Sigma_i;\Psi
        \vdash_e \letv{p_{i+1}}{\omega_{i+1}~p_{i+1}'}{\dots v_1'} : B
      }{
        \pi_{i+1}
      }
    }
  \]
  where $\TypCtxta''_{i-1},\Sigma_{i-1}$ is decomposed as
  $\TypCtxta'_i,\TypCtxta''_i$. By linearity, we have
  $|(p'_i,A_i)| = \TypCtxta'_i$, and $|(p_i,B_i)| = \Sigma_i$.
  At the level $n$, there is only one let-term left, and the proof $\pi_n$ is
  \begin{equation}\label{eq:v1-typ}
    \infer{
      \TypCtxta''_{n-1},\Sigma_{n-1};\Psi
      \vdash_e \letv{p_n}{\omega_n~p_n'}{v_1'} : B
    }{
      \infer{
        \TypCtxta'_n;\Psi
        \vdash_e \omega_n~p_n':B_n
      }{
        \infer*{\Psi \vdash_i \omega_n : A_n\iso B_n}{\pi_{\omega_n}}
        &
        \infer*{\TypCtxta'_n;\Psi\vdash_e p_n' : A_n}{\pi_{p_n'}}
      }
      &
      \infer*{
        \TypCtxta''_n,\Sigma_n;\Psi
        \vdash_e v_1' : B
      }{
        \pi_{v_1'}
      }
    }
  \end{equation}
  where $\TypCtxta''_{n-1},\Sigma_{n-1}$ is decomposed as
  $\TypCtxta'_n,\TypCtxta''_n$, and where, by linearity, we have
  $|(p'_n,A_n)| = \TypCtxta'_n$, and $|(p_n,B_n)| = \Sigma_n$.

  Let us now build a typing derivation for the judgment
  \[
    \TypCtxta''_n,\Sigma_n; \Psi
    \vdash_e \letv{p'_n}{\omega_n^\bot~p_n}{\dots \letv{p'_1}{\omega_1^\bot~p_1}{ v_1}} : A
  \]
  starting from the top. First, we can build the typing derivation $\pi_1^\bot$
  \[
    \infer{
      \TypCtxta''_1,\Sigma_1;\Psi^\bot\vdash_e \letv{p'_1}{\omega_1^\bot~p_1}{ v_1} :A
    }{
      \infer{
        \Sigma_1;\Psi^\bot\vdash_e \omega_1^\bot~p_1  : A_1
      }{
        \infer*{\Psi^\bot\vdash_e \omega_1^\bot: B_1\iso A_1}{\pi_{\omega_1}^\bot}
        &
        \infer*{\Sigma_1;\Psi^\bot\vdash_e p_1 : B_1}{\pi_{p_1}}
      }
      &
      \infer*{\TypCtxta'_1,\TypCtxta''_1;\Psi^\bot\vdash_e v_1 : B}{\pi_{v_1}}
    }
  \]
  Remember how $\TypCtxta'_1$ corresponds to $p'_1$ and $\Sigma_1$
  corresponds to $p_1$ and $\Psi^\bot$ is $f : D\iso C$, but it is
  not used in $\pi_{p_1}$ nor $\pi_{v_1}$. The derivation
  $\pi_{\omega_1}^\bot$ is built by invoking the induction hypothesis,
  and $\pi_{p_1}$ by recalling that $\Sigma_1$ is $|(p_1,B_1)|$. The
  derivation $\pi_{v_1}$ comes from the hypothesis that
  $\TypCtxta_1;\Psi\vdash_e v_1 : A$, that
  $\TypCtxta_1=\TypCtxta'_1,\TypCtxta''_1$ and that $f$ is not used in
  $v_1$, so that we can safely change its type to $D\iso C$ in the
  typing context.

  We now iteratively derive the proof $\pi_i$ of the intermediate
  typing judgments
  \[
    \TypCtxta''_i,\Sigma_i;\Psi^\bot\vdash_e
    \letv{p'_i}{\omega_i^\bot~p_i}{\dots
      \letv{p'_1}{\omega_1^\bot~p_1}{ v_1}
    } :A
  \]
  for incremental $1<i\leq n$, assuming that we already have $\pi_{i-1}$.
  The typing derivation $\pi_i$ is as follows:
  \[\scalebox{0.85}{
    \infer{
      \TypCtxta''_i,\Sigma_i;\Psi^\bot\vdash_e
      \letv{p'_i}{\omega_i^\bot~p_i}{\dots v_1}
      :A.
    }{
      \infer{
        \Sigma_i;\Psi^\bot\vdash_e \omega_i^\bot~p_i  : A_i
      }{
        \infer*{\Psi^\bot\vdash_e \omega_i^\bot: B_i\iso A_i}{\pi_{\omega_i}^\bot}
        &
        \infer*{\Sigma_i;\Psi^\bot\vdash_e p_i : B_i}{\pi_{p_i}}
      }
      &
      \infer*{\TypCtxta''_{i-1},\Sigma_{i-1};\Psi^\bot\vdash_e
        \letv{p'_{i-1}}{\omega_{i-1}^\bot~p_{i-1}}{\dots v_1} : A}{\pi_{i-1}^\bot}
    }
  }\]
  Remember how $\TypCtxta''_{i-1},\Sigma_{i-1}$ and
  $\TypCtxta''_i,\TypCtxta'_i$ are two decompositions of the same
  typing context.  The typing derivation $\pi_{p_i}$ comes from the
  fact that $\Sigma_i$ is $|(p_i,B_i)|$, and $\pi_{\omega_i}^\bot$
  comes from the induction hypothesis.

  When $i$ reaches $n$, we get a derivation for the typing judgment
  \[
    \TypCtxta''_n,\Sigma_n;\Psi^\bot\vdash_e
      \letv{p'_n}{\omega_n^\bot~p_n}{\dots \letv{p'_1}{\omega_1^\bot~p_1}{v_1}}
      :A.
  \]
  On the other hand, in \autoref{eq:v1-typ} we have a typing derivation $\pi_{v'_1}$ for
  \[
    \TypCtxta''_n,\Sigma_n;\Psi^\bot\vdash_e v'_1 : B.
  \]
  As $f$ is not used in $v'_1$ (since it is a value), this also gives us a typing derivation
  \[
    \TypCtxta''_n,\Sigma_n\vdash_e v'_1 : B.
  \]
  The reasoning we did on the clause $\{v_1\iso
  \letv{p_1}{\omega_1~p_1'}{\dots \letv{p_n}{\omega_n~p_n'}{v_1'}}\}$
  is general and can be reproduced for all clauses $\{v_i\iso e_i\}$.
  We can then derive a proof of the desired judgment
  \[
    f : D \iso C \entailiso\{v_1\iso e_1|\dots|v_m\iso e_m\}^\bot : B\iso A.
  \]

\item For the case of a recursive iso $\fix f.\omega$ of type $A\iso
    B$, one can directly invoke the inductive hypothesis on $\omega$
    and rely on the symmetry of the criterion for structural
    recursion.  \qedhere
\end{itemize}
\end{proof}

\begin{exa}
  \label{ex:map2}
  Recall the iso \textit{map}$(\omega)$ defined in
  \autoref{ex:map}. Its inverse $\textit{map}(\omega)^\bot$ is
  \[
    \textit{map}(\omega)^\bot
    = \fix f.\left\{
      \begin{array}{l@{~}c@{~}l}
        [~] & {\iso} & [~] \\
        h' :: t' & {\iso} & \letv{t}{f~t'}{} \\
            &  &  \letv{h}{\omega^\bot h'}{}\\
            & & h :: t \\
      \end{array}
    \right\}
  \]
  It can indeed be typed with $[B]\iso[A]$.
\end{exa}

We are left to show that our isos indeed represent isomorphisms,
meaning that given some iso $\entailiso \omega : A\iso B$ and some
value $\vdash_e v : A$, then $\omega^\bot(\omega~v) \to^* v$. For that
we need the following lemma:

\begin{lem}[Commutativity of substitution]
\label{isos:lem:commutativity-substitution}
Let $\sigma_1, \sigma_2$ and $v$, such that $\sigma_1\cup\sigma_2$
closes $v$ and $\mathtt{supp}(\sigma_1) \cap \mathtt{supp}(\sigma_2) =
\emptyset$ then $\sigma_1(\sigma_2(v)) = \sigma_2(\sigma_1(v))$
\end{lem}

\begin{proof}
Direct by induction on $v$ as $\sigma_1$ and $\sigma_2$ have disjoint
support: In the case where $v = x$ then either $\set{x\mapsto v'} \in
\sigma_1$ or $\set{x\mapsto v'}\in\sigma_2$ and hence
$\sigma_1(\sigma_2(x)) = v' = \sigma_2(\sigma_1(x))$. All the other
cases are by direct induction hypothesis as the substitutions enter the
subterms.
\end{proof}

\begin{lem}
  \label{lem:isos-iso}
  For all well-typed isos $\entailiso \omega : A\iso B$, and for
  all well-typed values $\vdash_e v : A$, if~
  $\omega~v\rightarrow^* v'$ a value, then $\omega^\bot\,v'\to^* v$.
\end{lem}

\begin{proof}
      Write $X$ for the number of $\mathrm{IsoRec}$ rules being used
      in the rewrite sequence
      \begin{equation}\label{eq:rw-step-1}
        \omega~v\rightarrow^* v'.
      \end{equation}
      We prove the result by induction on the lexicographical order on
      ($X$, size of $\omega$). Let us proceed by case distinction on
      $\omega$.
      \begin{itemize}
      \item $\omega$ cannot be an iso-variable since the typing context
        of $\vdash_i \omega : A\iso B$ is empty.
      \item Suppose that $\omega$ is of the form $\fix f.\omega'$.
        Then $\omega^\bot$ is $\fix f.\omega'^\bot$. The rewriting in
        \autoref{eq:rw-step-1}
        starts with an $\mathrm{IsoRec}$ rules as follows:
        \[
          ((\fix f.\omega')~v)
          \to
          ((\omega'[f:=\fix f.\omega'])~v)
          \to^*
          v_0 \] As the reduction from $(\omega'[f:=\fix
        f.\omega'])~v$ to $v_0$ takes one less number of
        $\mathrm{IsoRec}$ rules we can apply the induction hypothesis
        and we get that
        \begin{equation}
          \label{eq:rw-step-ih}
            (\omega'[f := \fix f. \omega'])^\bot v_0 \to^* v
        \end{equation}
        Notice that we have
        \begin{equation}
          \label{eq:rw-step-id}
          (\omega' [f := \fix f. \omega'])^\bot = \omega'^\bot [f := \fix f. \omega'^\bot]
        \end{equation}

        We can now show that $\omega^\bot v_0 \to^* v$. By definition
        of the inverse and the rewriting system we have
        \[\omega^\bot v_0 = (\fix f. \omega'^\bot) v_0
        \to_{\mathrm{IsoRec}} \omega'^\bot [f := \fix f.\omega'^\bot]
        v_0\] By using \autoref{eq:rw-step-id} and
        \autoref{eq:rw-step-ih} we can conclude that
        \[\omega'^\bot [f := \fix f.\omega'^\bot] v_0 \to^* v \]

      \item Consider the case where $\omega$ is of the form
        $\isobasique$.  Without loss of generality, we consider
        that the clause $\set{v_1 \iso e_1}$ matches with $v$ and
        therefore there is a substitution $\sigma_0$ such that
        $\sigma_0[v_1] = v$, and
        \begin{equation}\label{eq:rw-inv-beg}
          \omega~v \to \sigma_0(e_1).
        \end{equation}
        Assume that $e_1$ is
        \[
          \letv{p_1}{\omega_1~p_1'}{}\dots\letv{p_n}{\isoterm_n~p'_n}{v'_1}.
        \]
      By linearity, we can decompose $\sigma_0$ into $\sigma_0^1, \dots,
      \sigma_0^n, \sigma_0^{n+1}$ such that, after substitution we obtain
      \[
        \sigma_0(e_1)\quad=\quad
        \begin{array}{l}
          \letv{p_1}{\omega_1\,\sigma_0^1(p'_1)}{} \\
          \dots \\
          \letv{p_n}{\omega_n\,\sigma_0^n(p'_n)}{}
          \sigma_0^{n+1}(v_1')
        \end{array}
      \]
      By~\autoref{prop:progress}, each {\lett} construction will
      reduce, and by the rewriting strategy we will first rewrite
      $\letv{p_1}{\omega_1~p_1'}{\dots}$ before rewriting the other
      {\lett}. We have that
      \[
        \omega_1~\sigma_0^1(p_1') \to^* \ov{v_1}
      \]
      for some $\ov{v_1}$. At this point, a pattern-matching occurs
      between $p_1$ and $\ov{v_1}$, generating a new substitution
      $\sigma_1$: we have $\sigma_1[p_1] = \ov{v_1}$.
      Again, due to linearity this new substitution can be
      split into $\sigma_1^2, \dots, \sigma_1^n, \sigma_1^{n+1}$,
      where the support of $\sigma_1^i$ are the free variables of
      $\sigma_0^i(p_i')$ appearing in $p_1$ for $i\in \set{2, \dots, n}$ and the support
      of $\sigma_1^{n+1}$ is the free variables of $v_1'$ appearing in $p_1$. We get
      \begin{equation}\label{eq:rw-inv-1}
        \begin{array}{l}
          \letv{p_1}{\omega_1\,\sigma_0^1(p'_1)}{}\\
          \letv{p_2}{\omega_2\,\sigma_0^2(p'_2)}{}\\
          \letv{p_3}{\omega_3\,\sigma_0^3(p'_3)}{}\\
          \dots \\
          \letv{p_n}{\omega_n\,\sigma_0^n(p'_n)}{}\\
          \sigma_0^{n+1}(v_1')
        \end{array}
        \to^*
        \begin{array}{l}
          \letv{p_1}{\overline{v_1}}{}\\
          \letv{p_2}{\omega_2\,\sigma_0^2(p'_2)}{}\\
          \letv{p_3}{\omega_3\,\sigma_0^3(p'_3)}{}\\
          \dots \\
          \letv{p_n}{\omega_n\,\sigma_0^n(p'_n)}{} \\
          \sigma_0^{n+1}(v_1')
        \end{array}
        \to
        \begin{array}{l}
          \letv{p_2}{\omega_2\,\sigma_1^2(\sigma_0^2(p'_2))}{}\\
          \letv{p_3}{\omega_3\,\sigma_1^3(\sigma_0^3(p'_3))}{}\\
          \dots \\
          \letv{p_n}{\omega_n\,\sigma_1^n(\sigma_0^n(p'_n))}{} \\
          \sigma_1^{n+1}(\sigma_0^{n+1}(v_1'))
        \end{array}
      \end{equation}
      The operation can be repeated, and to each $p_k$ is associated a
      value $\overline{v_k}$ and a substitution $\sigma_k$ with
      $\sigma_k[p_k] = \overline{v_k}$ which can be split in
      $\sigma_k^{k+1},\dots,\sigma_k^{n+1}$. They verify
      \begin{equation}\label{eq:rw-inv-2}
        \omega_k(\sigma_{k-1}^k\dots\sigma_1^k\sigma_0^k(p'_k))
        \to^*\overline{v_k}.
      \end{equation}
      Continuing the rewriting started in \autoref{eq:rw-inv-beg} and
      \autoref{eq:rw-inv-1}, we then have
      \[
        \omega~v
        \to^*
        \sigma_{n}^{n+1}\dots\sigma_1^{n+1}\sigma_0^{n+1}(v_1')
      \]
      We were looking for $\omega^\bot(\omega~v)$: at the stage
      \[
        \omega^\bot(\sigma_{n}^{n+1}\dots\sigma_1^{n+1}\sigma_0^{n+1}(v_1')),
      \]
      the clause to fire is
      \[
        \begin{array}{c} \left(
          \begin{array}{l@{~}c@{~}l} v_1&{\iso}&\letv{p_1}{\omega_1~p_1'}{} \\
                                        && \cdots \\
                                        && \letv{p_n}{\isoterm_n~p'_n}{v'_1}
          \end{array}
          \right)^\bot := \left(
          \begin{array}{lcl@{}l@{}l} v'_1&{\iso}&\letv{p'_n}{\isoterm_n^{\bot}~p_n}{} \\
                                         && \cdots \\
                                         && \letv{p'_1}{\isoterm_1^{\bot}~p_1}{v_1}
          \end{array}
          \right).
        \end{array}
      \]
      Let $\Sigma = \sigma_{n}^{n+1}\dots\sigma_1^{n+1}\sigma_0^{n+1}$,
      we are therefore left to evaluate
      \[
        \left(
          \begin{array}{lcl@{}l@{}l} v_1'&{\iso}&
                                                  \letv{p'_n}{\isoterm_n^\bot~p_n}{} \\
                                         && \cdots \\
                                         && \letv{p'_1}{\isoterm_1^\bot~p_1}{v_1}
          \end{array}
        \right)~\Sigma(v_1')
      \]
      We get $\Sigma[v_1'] = \Sigma(v_1')$.
      Hence, after applying the substitution we have :
      \begin{equation}\label{eq:rw-inv-3}
        \begin{array}{l}
          \letv{p_n'}{\omega_n^\bot\,\Sigma(p_n)}{} \\
          \dots \\
          \letv{p_1'}{\omega_1^\bot\,\Sigma(p_1)}{}
          \Sigma(v_1)
        \end{array}
      \end{equation}
      Note that in $\Sigma$ the only substitution acting on $p_n$ is
      $\sigma_n^{n+1}$ the last one in the composition (since all the
      others can only have the variables in the $p_j$ for $j<n$ in
      their support). But $\sigma_n^{n+1}(p_n)$ is exactly
      $\sigma_n(p_n)=\overline{v_n}$.

      Recall that the rewriting shown in \autoref{eq:rw-inv-2} is part
      of the sequence of reduction in \autoref{eq:rw-step-1}. Then:
      either the number of $\mathrm{IsoRec}$ unfolding is smaller, or
      the size of the term is smaller. In both cases we can apply the
      induction hypothesis and deduce that
      \begin{equation}\label{eq:rw-inv-5}
        \omega_n^\bot\Sigma(p_n) = \omega_n^\bot\overline{v_n} \to^* (\sigma_{n-1}^n\dots\sigma_1^n\sigma_0^n(p'_n))
      \end{equation}
      Therefore, rewriting \autoref{eq:rw-inv-3} eventually yields
      \begin{equation}\label{eq:rw-inv-7}
        \begin{array}{l}
          \letv{p_{n-1}'}{\omega_{n-1}^\bot\,(\sigma_{n-1}^n\dots\sigma_1^n\sigma_0^n\sigma_{n}^{n+1}\dots\sigma_1^{n+1}\sigma_0^{n+1})(p_{n-1})}{} \\
          \dots \\
          \letv{p_1'}{\omega_1^\bot\,\Sigma(p_1)}{}
          \Sigma(v_1)
        \end{array}
      \end{equation}
      Note how in the sequence of compositions
      \begin{equation}\label{eq:rw-inv-6}
        \sigma_{n-1}^n\dots\sigma_1^n\sigma_0^n\sigma_{n}^{n+1}\dots\sigma_1^{n+1}\sigma_0^{n+1}
      \end{equation}
      the only substitution with support matching the free variables
      of $p_{n-1}$ are $\sigma_{n-1}^n$ and $\sigma_{n-1}^{n+1}$:
      The composition in \autoref{eq:rw-inv-6} is $\sigma_{n-1}$. The term in \autoref{eq:rw-inv-7} is then
      \[
        \begin{array}{l}
          \letv{p_{n-1}'}{\omega_{n-1}^\bot\,(\sigma_{n-1}(p_{n-1}))}{} \\
          \dots \\
          \letv{p_1'}{\omega_1^\bot\,\Sigma(p_1)}{}
          \Sigma(v_1).
        \end{array}
      \]
      We can iterate the process by reducing
      $\omega_{n-1}^\bot\,(\sigma_{n-1}(p_{n-1})) =
      \omega_{n-1}^\bot\,\overline{v_{n-1}}$ using the induction
      hypothesis, and continue until we reach $v_1$. When we reach
      this level, the substitution is then $\sigma_0$, and we retrieve
      $\sigma_0(v_1) = v$, as expected.
      \qedhere
    \end{itemize}
\end{proof}

\begin{thm}[Isos are isomorphisms]
  \label{proof:isos-iso}
  \label{thm:isos-iso}
  For all well-typed isos $\entailiso \omega : A\iso B$, and for
  all well-typed values $\vdash_e v : A$, if~
  $(\omega^\bot~(\omega~v))\rightarrow^* v'$ then $v = v'$.
\end{thm}

\begin{proof}
  According to the evaluation
  strategy, we can decompose the rewriting of
  $\omega^\bot~(\omega~v)$ into two steps:
  \begin{equation}\label{eq:rw-step-0}
    \omega^\bot~(\omega~v) \to^* \omega^\bot v_0
  \end{equation}
  followed by
  \begin{equation}\label{eq:rw-step-2}
    \omega^\bot v_0 \to^* v'.
  \end{equation}
  In \autoref{eq:rw-step-0}, we have $\omega~v\to^* v_0$ a value. We
  can invoke \autoref{lem:isos-iso} and get \autoref{eq:rw-step-2}
  with $v'=v$, relying on the determinism of the rewriting procedure.
\end{proof}

\section{Computational Content}
\label{sec:rpp}

In this section, we study the computational content of our
language. Specifically, if we identify types with the set of their
closed values ---associating for
instance $(\one \oplus \one) \otimes \one$ with the set of values
$\{\pv{\inl()}{()}, \pv{\inr()}{()}\}$, closed isos $A\iso B$ can be
regarded as bijective maps between the corresponding sets of
values.
In the case of a closed iso $A\iso B$ between types not involving the
type constructor $\mu X.A$, any bijection between values of type $A$
and values of type $B$ can be represented with a pattern matching, one
pattern for each closed value.
However, with (infinite) inductive types, the expressivity becomes
less clear as a bijection does not have a finite representation
anymore. In this section, we show that we can encode the class of
Recursive Primitive Permutations~\cite{rpp} (RPP), implying that,
although restricted, our language can express all primitive recursive
functions~\cite{rogers1987theory}.

\subsection{From RPP to Isos}

As the language $\RPP{}$ considers integers, we need to define a
type to represent them. We start by defining the type of
strictly positive natural numbers, $\mathtt{npos}$, as $\mathtt{npos}
= \mu X. \one \oplus X$. We then define $\underline{n}$, the encoding of a
positive natural number $n$ into a value of type $\mathtt{npos}$ inductively, as
\[
  \underline{1}   = \fold{(\inl{()})},
  \quad
  \underline{n+1} = \fold{(\inr{(\underline{n})})}.
\]
Finally, we define the type of integers
as $Z = \one \oplus (\mathtt{npos} \oplus \mathtt{npos})$ along with
$\overline{z}$ the encoding of any $z\in\intT$ into a value of type
$Z$ defined as:
\begin{align*}
  \overline{0} &= \inl{()},\\
  \overline{z} &= \inr{(\inl{(\underline{z})})} &&  \text{for $z$ positive,}\\
  \overline{z} &= \inr{(\inr{(\underline{- z})})} && \text{for $z$ negative.}
\end{align*}
Given some function $f \in \RPP{k}$, we show in the remainder of the
section how to build an iso $\operatorname{isos}(f) : Z^k \iso Z^k$
realizing $f$. The construction is defined inductively on the
structure of $f$, as presented in \autoref{background:rpp}.

\subsubsection{Encoding of Primitives}

\begin{itemize}
\item The Successor is
  \[\left\{\begin{array}{ccc}
    \inl{()} &\iso& \inr{(\inl{(\fold{(\inl{()})})})} \\
    \inr{(\inl{x})}&\iso&\inr{(\inl{(\fold{(\inr{x})})})} \\
    \inr{(\inr{(\fold{(\inl{()})})})}&\iso& \inl{()}\\
    \inr{(\inr{(\fold{(\inr{x})})})}& \iso & \inr{(\inr{x})}
  \end{array}\right\} : Z \iso Z \]

\item The Sign-change is
  \[\left\{\begin{array}{ccc}
    \inr{(\inl{x})} &\iso& \inr{(\inr{x})} \\
    \inr{(\inr{(x)})} &\iso& \inr{(\inl{x})} \\
    \inl{()} &\iso&\inl{()}
  \end{array}\right\} : Z \iso Z\]

\item The identity is $\{x\iso x\} : Z \iso Z$.

\item The Swap is $\{(x, y) \iso (y, x)\} : Z^2 \iso
  Z^2$.

\item The Predecessor is the inverse of the Successor.
\end{itemize}

\subsubsection{Encoding of Horizontal and Vertical Composition}

Consider $f$ and $g$ two permutations in $\RPP{j}$, and write
$\omega_f = \operatorname{isos}(f)$ and
$\omega_g = \operatorname{isos}(g)$ for the isos encoding $f$ and
$g$. We encode the composition $\operatorname{isos}(f;g)$ of $f$ and
$g$ with type $Z^{j} \iso Z^{j}$ as:
\[
  \operatorname{isos}(f;g) =
  \left\{
    \begin{array}{ccc}
      &  & \letv{(y_1, \dots, y_j)}{\omega_f~(x_1, \dots, x_j)}{} \\
      (x_1, \dots, x_j) & \iso & \letv{(z_1,
                                 \dots, z_j)}{\omega_g~(y_1, \dots, y_j)}{} \\
      & & (z_1, \dots, z_j)
    \end{array}
  \right\}
\]

If now $f \in \RPP{j}$ and $g \in \RPP{k}$, and if $\omega_f =
\operatorname{isos}(f)$ and $\omega_g = \operatorname{isos}(g)$, we
encode the parallel composition $\operatorname{isos}(f\mid\mid g)$ of $f$ and $g$ with type $ Z^{j+k} \iso
Z^{j+k}$ as:
\[
  \operatorname{isos}(f\mid\mid g) =
  \left\{
    \begin{array}{ccc}
      &
      & \letv{(x_1', \dots, x_j')}{\omega_f~(x_1, \dots, x_j)}{} \\
      (x_1, \dots, x_j, y_1, \dots, y_k)
      & \iso
      &
        \letv{(y_1',
        \dots, y_k')}{\omega_g~(y_1, \dots, y_k)}{} \\
      & & (x_1', \dots, x_j', y_1', \dots, y_k')
    \end{array}
  \right\}
\]

\subsubsection{Encoding of Finite Iteration}

Consider $f\in \RPP{k}$ and $\omega_f = \operatorname{isos}(f)$ its encoding.
We encode the finite iteration $\mathbf{It}[f]\in\RPP{k+1}$ with the help of the
auxiliary iso $\omega_{\operatorname{aux}}$, of type $Z^k
\otimes~\mathtt{npos} \iso Z^k \otimes~\mathtt{npos}$
doing the finite iteration using $\mathtt{npos}$, defined as:
\[
  \omega_{\operatorname{aux}} = \fix g.
  \left\{\begin{array}{ccc}
    (\overrightarrow{x}, \fold{(\inl{()})}) & \iso & \letv{\overrightarrow{y}}{\omega_f~\overrightarrow{x}}{}
    \\
                                            & & (\overrightarrow{y},\fold{(\inl{()})}) \\[0.5em]
    (\overrightarrow{x}, \fold{(\inr{n})})& \iso &  \letv{(\overrightarrow{y})}{\omega_f~(\overrightarrow{x})}{} \\
                                            & & \letv{(\overrightarrow{z}, n')}{g~(\overrightarrow{y}, n)}{} \\
                                            & & (\overrightarrow{z}, \fold{(\inr{n'})}) \\[0.5em]
  \end{array}\right\}
\]
Given a pair $(x,n)$, this iso iterates $f$ $n+1$ times on $x$.
We can now properly define $\operatorname{isos}(\mathbf{It}[f])$ of
type $Z^{k+1} \iso Z^{k+1}$ as:
\[
  \operatorname{isos}(\mathbf{It}[f]) = \left\{\begin{array}{ccc}
    (\overrightarrow{x}, \inl{()})&\iso&(\overrightarrow{x}, \inl{()}) \\[0.5em]
    (\overrightarrow{x}, \inr{(\inl{z})})&\iso& \letv{(\overrightarrow{y}, z')}{\omega_{aux} (\overrightarrow{x}, z)}{} \\
    & & (\overrightarrow{y}, \inr{(\inl{z'})}) \\[0.5em]
    (\overrightarrow{x}, \inr{(\inr{z})})&\iso& \letv{(\overrightarrow{y}, z')}{\omega_{aux} (\overrightarrow{x}, z)}{} \\
    & & (\overrightarrow{y}, \inr{(\inr{z'})}) \\[0.5em]
  \end{array}\right\}
\]
We simply have to perform a case distinction on whether the number of
iterations is given as a positive, negative, or null value, and run
$\omega_{\operatorname{aux}}$ accordingly.

\subsubsection{Encoding of Selection}

Consider $f, g, h\in\RPP{k}$ and their corresponding encoding $\omega_f =
\operatorname{isos}(f), \omega_g = \operatorname{isos}(g)$ and  $\omega_h =
\operatorname{isos}(h)$. We define
$\operatorname{isos}(\mathbf{If}[f,g,h])$ of type $Z^{k+1} \iso
Z^{k+1}$ as:
\[
  \operatorname{isos}(\mathbf{If}[f,g,h]) = \left\{\begin{array}{ccl}
    (\overrightarrow{x}, \inr{(\inl{z})}) &\iso&
    \letv{\overrightarrow{x'}}{\omega_f
    (\overrightarrow{x})}{(\overrightarrow{x'}, \inr{(\inl{z})})} \\
    (\overrightarrow{x}, \inl{()}) &\iso&
    \letv{\overrightarrow{x'}}{\omega_g
    (\overrightarrow{x})}{(\overrightarrow{x'}, \inl{()})} \\
    (\overrightarrow{x}, \inr{(\inr{z})}) &\iso&
    \letv{\overrightarrow{x'}}{\omega_h
    (\overrightarrow{x})}{(\overrightarrow{x'}, \inr{(\inr{z})})}
  \end{array}\right\}
\]

\subsubsection{Soundness of the Encoding}

In order to make sure that our encoding is sound, we need to make sure
of two things: first that it is well-typed, and then that the
semantics is preserved.

\begin{thm}[The encoding is well-typed]
\label{lem:rpp-typed}
Let $f\in \RPP{k}$, then $\entailiso \operatorname{isos}(f) : Z^k \iso Z^k$.
\end{thm}
\begin{proof}
  By induction on $f$, for the two compositions, iteration, and
  selection the variables $\overrightarrow{x}$ are all of type $Z$,
  while for $\omega_{\opn{aux}}$ the last argument is of type
  $\mathtt{npos}$. The predicate $\PE_Z$ is always satisfied as the
  left columns of the $n$-fold tensor are always variables and the
  right-most argument on each isos always satisfies $\PE_Z$.
\end{proof}

\begin{thm}[Simulation]
\label{thm:rpp-sim}
Let $f\in \RPP{k}$ and $~n_1, \dots, n_k$ elements of~$~\intT$
such that $f(n_1, \dots, n_k) = (m_1, \dots, m_k)$ then $\operatorname{isos}(f)
(\overline{n_1}, \dots, \overline{n_k}) \to^* (\overline{m_1}, \dots,
\overline{m_k})$
\end{thm}
\begin{proof}
  By induction on $f$.
  \begin{itemize}
      \item Direct for the identity, swap and sign-change.

      \item For the Successor:
      \[\omega = \left\{\begin{array}{ccc}
          \inl{()} &\iso& \inr{(\inl{(\fold{(\inl{()})})})} \\
          \inr{(\inl{(x)})}&\iso&\inr{(\inl{(\fold{(\inr{(x)})})})} \\
          \inr{(\inr{(\fold{(\inl{()})})})}&\iso& \inl{()}\\
          \inr{(\inr{(\fold{(\inr{x})})})}& \iso & \inr{(\inr{(x)})}
      \end{array}\right\}\] we do it by case analysis on the sole
      input $n$.
      \begin{itemize}
          \item Case $n = 0$. We have $\ov{0} = \inl{()}$ and
          $\omega~\inl{()} \to \inr{(\inl{(\fold{(\inl{()})})})} = \ov{1}$.
          \item Case $n = -1$. We have $\ov{-1} =
          \inr{(\inr{(\fold{(\inl{()})})})}$, so the term reduces to
          $\inl{()} = \ov{0}$.
          \item Case $n < -1$. We have $\ov{n} =
          \inr{(\inr{(\fold{(\inr{\underline{n'}})})})}$ with
          $\underline{n'} = \underline{-(n+1)}$, by
          pattern-matching we get $\inr{(\inr{\underline{n'}})}$
          which is $\overline{-(n+1)}$.

          \item Case $n > 1$ is similar.
      \end{itemize}

      \item The Predecessor is the dual of the Successor.

      \item \textbf{Composition \& Parallel composition}: Direct by
      induction hypothesis on $\omega_f$ and $\omega_g$.: for the
      composition, $\omega_f$ is first applied on all the input and
      then $\omega_g$ on the result of $\omega_f$. For the parallel
      composition, $\omega_f$ is applied on the first $j$ arguments
      and $\omega_g$ on the argument $j+1$ to $k$ before
      concatenating the results from both isos.

      \item \textbf{Finite Iteration}: $\opn{It}[f]$.

      We need the following lemma: $\omega_{\text{aux}} (\bar{x_1},
      \dots, \bar{x_n}, \underline{z}) \to^* (\bar{z_1}, \dots,
      \bar{z_n}, \underline{z})$ where $z$ is a non-zero integer and
      $(z_1, \dots, z_n) = f^{\mid z \mid} (x_1, \dots, x_n)$ which
      can be shown by induction on $\mid z\mid$: the case $z = 1$
      and $\bar{z} = \fold{\inl{()}}$ is direct by induction
      hypothesis on $\omega_f$. Then if $z = n+1$ we get it directly
      by induction hypothesis on both $\omega_f$ and our lemma.

      Then, for $\opn{isos}{\opn{It}[f]}$ we do it by case analysis
      on the last argument: when it is $\bar{0}$ then we simply
      return the result, if it is $\bar{z}$ for $z$ (no matter if
      strictly positive or strictly negative) then we enter
      $\omega_{aux}$, and apply the previous lemma.

      \item \textbf{Conditional}: $\opn{If}[f, g, h]$. Direct by case
      analysis of the last value and by induction hypothesis on
      $\omega_f, \omega_g, \omega_h$.
      \qedhere
  \end{itemize}
\end{proof}

\autoref{thm:rpp:soudness-completeness} and \autoref{thm:rpp-sim} tell us that any
primitive recursive function can effectively be encoded as a
well-typed iso with the appropriate semantics. It is in particular
possible to encode the Cantor-Pairing~\cite[Theorem 2 and Theorem
4]{rpp}, although with auxiliary arguments. Another
representation of the Cantor Pairing, without auxiliary arguments, would be
possible in our language but would require a more lax notion of
structural recursion. This is discussed in \autoref{section:conclusion}.

Also note that $\operatorname{isos}(f)^\bot \not=
\operatorname{isos}(f^{-1})$. Indeed,
$\operatorname{isos}(f)^\bot$ inverses the order of the
\lett~constructions, whereas this will not happen for
$\operatorname{isos}(f^{-1})$. The two can nonetheless be considered
equivalent up to a permutation of \lett~constructions and
renaming of variables. \autoref{thm:rpp-sim} also confirms that they
do indeed have the same behavior.

\section{Proof-Theoretical Content}\label{sec:ch}

In the present section, we want to relate our language of isos to
proofs in a \mumall. As mentioned earlier, an iso $\entailiso \omega :
A\iso B$ corresponds to both a computation sending a value of type $A$
to a result of type $B$ and a computation sending a value of type $B$
to a result of type $A$, inverse of each other. Under the principles
of Curry-Howard correspondence, an iso should therefore correspond to
a \emph{proof} isomorphism : the data of two proofs, $\pi$ of $A\vdash
B$ and $\pi^\bot$ of $B\vdash A$, which are inverse of each other in
the sense that, when cut together behave like an (expansion of the)
axiom rule, in the sense that cut with any proof $\pi'$ of $\vdash A$
(resp. $\vdash B$), cut-elimination returns the proof $\pi'$
unchanged, corresponding to~\autoref{thm:isos-iso}.

We will rely on the infinite derivations available in {\mumall} to
obtain a proper representation of the recursive behaviours available
in our iso. The main difficulty will consist in ensuring that any
pre-proof of {\mumall} obtained from an iso will indeed satisfy
\mumall validity criterion (\autoref{background:mumall}).

\subsection{Translating isos into \texorpdfstring{\mumall}{muMALL}}

We start by giving the translation from isos to pre-proofs, and then
show that they are actually proofs, therefore obtaining a
\textit{static} correspondence between reversible programs and \mumall
proofs. We then show that our translation entails the expected
\textit{dynamic} correspondence between the evaluation procedure of
our language and \mumall cut-elimination procedure. The derivations we
obtain are \emph{circular}, and we therefore translate an iso $\omega$
directly into a finite derivation with back-edges, written as
$\circu(\omega)$. A translation to infinite derivations can
straightforwardly be obtained by composing $\circu(-)$ with the
unfolding operation.

In this translation, we need to manage a correspondence between typing
contexts which are sets of variables with associated types, and
sequent contexts, which are lists of formulas. To do so, we assume
once and for all a fixed linear order on variables (for instance given
by an arbitrary enumeration of the set of variables) and will consider
that to a typing context $\Theta$ one associates the list made of the
formulas appearing in the typing contexts following the order given by
the variable ordering, written as $\overline{\Theta}$.

\subsubsection{Definition of the Translation}
Given an iso $\omega : A\iso B$, its translation into a \mumall
derivation of $A \vdash B$ is described in three separate phases:

\begin{description}
\item[\emph{Iso Phase}] The first phase consists in traveling through
the syntactical definition of an iso, keeping track of the last
encountered iso-variable bounded by a $\fix$-construction, if any.
When an iso of the shape {\isobasique} of type $A\iso B$ is
encountered, the negative phase begins by creating a sequent
$A\vdash^f B$ labeled with the last encountered iso-variable. This
phase will resume this phase when encountering another iso in one of
the $e_i$; a back-edge is created in the case of an iso variable.
\item[\emph{Negative Phase}] The negative phase is guided by a list of
formulas to be decomposed using the reversible rules of the logic,
beginning with a singleton list $[A]$ that corresponds
to the decomposition of formula $A$ according to how values of type
$A$ on the left-hand side of $\omega$ are pattern-matched. The
negative phase relies on the $\Neg(\cdot)$ function which takes $5$
arguments:
\begin{itemize}
  \item A set of pairs of a list of values and typing derivation,
  written $(l, \xi)$ where each element of the set corresponds to one
  clause $v \iso e$ of the given iso and $\xi$ is the typing
  derivation of $e$. The list of values corresponds to what is left to
  be decomposed on the left-hand side of the clause (for instance, if
  $v$ is a pair $\pv{v_1}{v_2}$ the list will have two elements to
  decompose).

  \item A list of formulas $\Gamma$ containing the formulas being
  decomposed. Intuitively, the $i$th formula of the list corresponds to
  the type of the $i$th value in the list $l$.

  \item A context $\Delta$, containing the formulas that have
  already been decomposed and will no longer be decomposed during
  the negative phase. This context can only grow during the negative
  phase: once a formula is in it, it will no longer be modified until the end
  of the negative phase.

  \item A formula $A$ which corresponds to the right-hand side formula
  of the conclusion sequent (or rather, to the single positive
  formula) of the derivation that we are currently building, and which
  will not be modified during the negative phase.

  \item A (potentially) annotated sequent $\vdash^f$, which will be
  used as the target of a back-edge.
\end{itemize}

Therefore, $\Neg(\set{(l_i, \xi_i) \mid{i\in I}}, \Gamma, \Delta, A,
\vdash^a)$ will build a derivation $\Gamma, \Delta \vdash^a A$ by
decomposing the formulas in $\Gamma$ according to the values present
in $l$. The negative phase ends when the list is empty, and hence when
$\Gamma = []$. When this is the case, we can start translating $\xi$
and the \emph{positive phase} begins. The negative phase is defined
inductively on the first element of the list of every set, which are
known by typing to have the same prefix, and is given
in~\autoref{tab:negative_phase}.
\item[\emph{Positive phase}] The translation of an expression is
pretty straightforward: each \textit{let} and iso-application is
represented by two \cut~rules, as usual in Curry-Howard
correspondence. The second argument is a list of formulas for which we
maintain the following invariant at each call to $\Pos{(\xi,
\Gamma)}$: $\Gamma$ is a list of formulas corresponding to an ordering
of the types declared in the typing context of $\xi$. This is useful
not only to generate the \mumall derivation but also to manage the
sequent contexts correctly when needing to split it, such as for the
$\tensor$ rule for instance. The definition of the positive phase is
given in~\autoref{tab:positive_phase}, where in the last rule we
have $|(p, A)| = ([x_1, \dots, x_n], [A_1, \dots, A_n])$ as defined
in~\autoref{subsubsection:typing}.
\end{description}

\begin{defi}[Translation of Isos into {\mumall}]
  \label{tab:circ_translation}
  Given a well-typed iso $\omega : A \iso B$, its translation into
  {\mumall}, $\circu(\omega)$ is defined altogether with the functions
  $\Neg$ and $\Pos$ by mutual recursion and produces a circular
  derivation of {\mumall}. $\circu(\omega)$ is defined below while
  $\Neg$ and $\Pos$ are defined in \autoref{tab:negative_phase} and
  \autoref{tab:positive_phase} respectively:
  \begin{itemize}
    \item $\circu(f : A\iso B \vdash f : A\iso B) =
    \begin{array}{c}\infer[be(f)]{A \vdash
    B}{}\end{array}$

    \item $ \circu(\Psi \entailiso \{(v_i\iso e_i)_{i\in I}\} : A\iso B) =
    {\Foc{([v_i], \xi_i)_{i\in I}, [A],\emptyset,
    B, \vdash}}
    $
    \item $\circu(\Psi \entailiso \fix f. \{(v_i\iso e_i)_{i\in I}\} : A\iso B)
    =
     \Foc{([v_i], \xi_i)_{i\in I}, [A],\emptyset,
    B, \vdash^f}$
  \end{itemize}

  Where the $\xi_i$ are the typing derivations of the $e_i$.
\end{defi}

\begin{exa}
  The translation $\pi = \circu(\omega)$ of the
  iso $\omega$ from~\autoref{ex:iso1} is:

\vspace*{4pt}
\begin{prooftree}

  \infer0[\id]{A\vdash A}
  \infer1[$\oplus^1$]{A\vdash A\oplus B}
  \infer1[$\oplus^2$]{A\vdash C\oplus (A\oplus B)}

  \infer0[\id]{B\vdash B}
  \infer1[$\oplus^2$]{B\vdash A\oplus B}
  \infer1[$\oplus^2$]{B\vdash C\oplus (A\oplus B)}

  \infer0[\id]{C\vdash C}
  \infer1[$\oplus^1$]{C\vdash C\oplus
  (A\oplus B)}

  \infer2[$\with$]{B\oplus C\vdash C\oplus (A\oplus B)}
  \infer2[$\with$]{A\oplus (B\oplus C)\vdash C\oplus (A\oplus B)}
  \end{prooftree}
\end{exa}

\begin{exa}
  Considering the iso swap of type $A\otimes B \iso B\otimes A$
  defined as $\opn{swap} = \{\pv {x}{y} \iso \pv{y}{x} \}$, its
  corresponding proof is the proof $\pi_S$ given
  in~\autoref{ex:mumall-thread} whereas the iso
  $\opn{map}(\opn{swap})$ as defined by~\autoref{ex:map} has as its
  corresponding proof the proof $\phi(\pi_S)$
  from~\autoref{ex:mumall-thread}. Notice that the blue thread follows
  the focus of the structurally recursive criterion. The
  \emph{negative phase} consists of the $\nu, \with, \parr$ and $\bot$
  rules, where the pre-thread is going up, while the \emph{positive
  phase} consists of the multiple cut-rules where the pre-thread is not
  active.
\end{exa}

\begin{figure}[t]
  \scalebox{.8}{\begin{minipage}{1.18\textwidth}
  \[
  \begin{array}{c}
  \Neg( \set{(\inl{v_j} :: l_j, \xi_j) \mid {j\in J}}
  \cup \set{(\inr{v_k} :: l_k, \xi_k)\mid {k\in K}},
  A_1\oplus A_2 :: \Gamma, \Delta, B, \vdash^\Psi)
  = \qquad~ \\[1em]
~\qquad   \infer[\with]{A_1 \oplus A_2, \Gamma, \Delta\vdash^\Psi
  B}{\Neg(\set{(v_j ::
  l_j, \xi_j)\mid {j\in J}},A_1 :: \Gamma, \Delta, B, \vdash) \qquad
  \Neg(\set{(v_k :: l_k, \xi_k)\mid {k\in K}}, A_2 :: \Gamma, \Delta, B, \vdash)}
  \end{array}
  \qquad\]

\[  \scalebox{0.9}{$
  {\Neg(\set{(() :: l_i, \xi_i) \mid i\in I},\one :: \Gamma, \Delta, B, \vdash^\Psi)}
  =
  \infer[\bot]{\one, \Gamma, \Delta \vdash^\Psi B}{
  {\Neg(\set{(l_i, \xi_i) \mid i\in I}, \Gamma , \Delta,B, \vdash)}}
  $}\]

  \[ \scalebox{0.9}{$
  {\Neg(\set{(\pv{v^1_i}{v^2_i} :: l_i, \xi_i) \mid{i\in I}}, A_1\otimes A_2 :: \Gamma , \Delta, B, \vdash^\Psi)}
  = \hspace{0cm}
 \infer[\parr]{A_1\otimes A_2, \Gamma, \Delta
  \vdash^\Psi B}{
  {\Neg(\set{(v_i^1 :: v_i^2 :: l_i, \xi_i) \mid{i\in
  I}},A_1:: A_2 :: \Gamma , \Delta, B, \vdash)}}
  $} \]

  \[ \scalebox{0.9}{$
  {\Neg(\set{(\fold{v_i} :: l_i, \xi_i)\mid {i\in
  I}}, \mu X. A :: \Gamma, \Delta,  B, \vdash^\Psi)}
  =
  \infer[\nu]{\mu X. A, \Gamma, \Delta\vdash^\Psi
  B}{
  {\Neg(\set{(v_i :: l_i, \xi_i)\mid {i \in I}},A[X \leftarrow \mu X. A] :: \Gamma, \Delta,
  B, \vdash)}}
  $}\]

  \[\scalebox{0.9}{$
  {\Neg(\set{(x :: l_i , \xi_i) \mid {i\in I}},A :: \Gamma , \Delta , B, \vdash^\Psi)}
  =
  \infer[ex
  ]{A, \Gamma, \Delta \vdash^\Psi  B}
  {
  {\Neg(\set{(l_i,\xi_i) \mid {i\in I}}, \Gamma, A::\Delta,B, \vdash)}}
  $} \]

  \[\scalebox{0.9}{$
 {\Neg(\set{([], \xi)},[], \Delta,B, \vdash^\Psi)}
 =
 {\Pos(\xi, \Delta
 )}
 $}
 \]
\end{minipage}}
  \caption{Negative Phase.}
  \label{tab:negative_phase}
\end{figure}

\begin{figure}[t]
  \centering

  \scalebox{.8}{\begin{minipage}{1.18\textwidth}
  \begin{align*}
    \myPos{\infer{\vdash_e () : \one}{}, \emptyset} &=
      \begin{array}{c}
      \infer[\one]{\vdash \one}{}\end{array} \\[0.5em]
      \myPos{\infer{x : A\vdash_e x : A}{}, [A]} &= \begin{array}{c}\infer[\id]{A \vdash A}{}\end{array} \\
      \myPos{\infer{{\Theta} \vdash_e \inl{t} :
      A_1\oplus A_2}{\xi:
      {({\Theta} \vdash_e t : A_1)}
      }, \Gamma} &=
      \begin{array}{c}\infer[\oplus^1]{{\Gamma} \vdash A_1 \oplus A_2}{
      {\Pos(\xi, \Gamma)}}\end{array} \\
      \myPos{\infer{{\Theta} \vdash_e \inr{t} : A_1\oplus
      A_2}{\xi:
      ({{\Theta} \vdash_e t : A_2})
      }, \Gamma} &=
      \begin{array}{c}\infer[\oplus^2]{{\Gamma} \vdash A_1 \oplus A_2}{
      {\Pos(\xi, \Gamma)}}\end{array} \\
       \myPos{\infer{{\Theta} \vdash_e \fold{t} :
      \mu X. A}{\xi:
      ({{\Theta} \vdash_e t : A[X\leftarrow \mu X.
      A]})
      }, \Gamma } &= \begin{array}{c}\infer[\mu]{ {\Gamma} \vdash \mu X. A}{
      {\Pos(\xi, \Gamma)}}\end{array} \\
      \myPos{\infer{{\Theta}_1,
      {\Theta}_2 \vdash_e \pv{t_1}{t_2} : A_1\otimes
      A_2}{
      \xi_1:
      ({{\Theta}_1 \vdash_e t_1 : A_1})
      \qquad
      \xi_2:
      ({{\Theta}_2\vdash_e t_2 : A_2})
}, \Gamma} &=
      \begin{array}{c}
      \infer[\rex^\star]{\Gamma\vdash A_1\otimes A_2}{\infer[\otimes]{ \overline{\Theta_1},\overline{\Theta_2}\vdash A_1\otimes
      A_2}{
      {\Pos(\xi_1,\overline{\Theta_1})}\qquad
      {\Pos(\xi_2,\overline{\Theta_2})}}}\end{array} \\
      \myPos{\infer{{\Theta};\Psi\vdash_e
      \omega~t : B}{\Psi\entailiso \omega : A \iso
      B\qquad\xi: (
      {{\Theta}\vdash_e t : A})
      }, \Gamma} &=
      \begin{array}{c}\infer[\cut]{ {\Gamma}\vdash B}{
      {\Pos(\xi, \Gamma)}\qquad
      {\circu(\omega)}}\end{array}
    \end{align*}

    \[
    \begin{array}{rl}
      \myPos{\infer{{\Theta}_1,{\Theta}_2\entaile
    \letv{p}{t_1}{t_2} :
    B}{
    \xi_1:
    ({{\Theta}_1\entaile t_1 : A})
    \qquad
    \xi_2;
    ({{\Theta}_2, x_1 : A_1, \dots, x_n : A_n\entaile t_2 : B})
    }, \Gamma} &= \\[2em]
    \infer[\rex^\star]{\Gamma\vdash B}{\infer[\cut]{ \overline{\Theta_1},\overline{\Theta_2}\vdash B}{
    {\Pos(\xi_1, \overline{\Theta_1})}\qquad
    {\Foc{\set{([x_1, \dots, x_n], \xi_2)}, [A_1\otimes\dots\otimes A_n], \overline{\Theta_2},B,\vdash}}}}
    \end{array}
    \]

\end{minipage}}
\caption{Positive Phase.}
  \label{tab:positive_phase}
\end{figure}

\begin{lem}
  \label{lem:inf-branch-form}
  Given an iso $\entailiso \omega : A\iso B$, and given $\pi =
  \circu(\omega)$, for each infinite branch of $\pi$, only one sequent
  tagged by an iso-variable is visited infinitely often.
\end{lem}
\begin{proof}
  Since we have at most one iso-variable, we never end up in a
  situation where, between an annotated sequent $\vdash^f$ and a
  back-edge pointing to $f$ we encounter another annotated sequent.
\end{proof}

Among the terms that we translate, the translation of a value yields
what we call a \emph{purely positive proof}, which is trivially a valid
pre-proof.

\begin{defi}[Purely Positive Proof]
  A \textit{purely positive proof} is a finite, cut-free proof whose
  rules are only $\oplus^i, \otimes, \mu, \rex, \one, \id$ for
  $i\in\{1,2\}$.
\end{defi}

\begin{lem}[Values are Purely Positive Proofs]
  \label{lem:valppp}
Given $\Theta = \set{x_1 : A_1, \dots, x_n : A_n}$, $\xi : (\Theta
  \vdash v : A)$ and $\Delta$ a list of formulas corresponding to
  $\Theta$ then
$
\Pos{(\xi, {\Delta})}
$ is
  a purely positive proof.
\end{lem}

\begin{proof}
  By induction on the structure of $\xi$:
  \begin{itemize}
    \item $x : A \vdash_e x : A$ then the derivation is
    $\infer[\id]{A\vdash A}{}$,
    which is a purely positive proof;

    \item $\vdash () : \one$ then the derivation is
$\infer[\one]{\vdash \one}{}$,
    which is a purely positive proof;

    \item ${\Theta_1}, {\Theta_2} \vdash
    \pv{v_1}{v_2} : A\otimes B$ then we get
    $\begin{array}{c}
    \infer[\rex^\star]{\Delta\vdash A\otimes B}{\infer[\otimes]{\overline{\Theta_1},\overline{\Theta_2}\vdash A\otimes
    B}{\infer{\overline{\Theta_1}\vdash A}{\pi_1}\qquad \infer{\overline{\Theta_2}\vdash
    B}{\pi_2}}}\end{array}$
    and then by induction hypothesis on the typing derivations $\xi_1$ and
    $\xi_2$ of $v_1$ and $v_2$ we get the expected result;

    \item ${\Theta}\vdash \inl{v} : A\oplus B$ then the
    derivation is
    $\infer[\oplus^1]{\Delta\vdash A\oplus
    B}{\infer{\Delta\vdash A}{\pi}}$
    then by induction hypothesis on
    the typing derivation $\xi$ of $v$ we get the desired result;

    \item The proof is similar for $\inr{v}$ and $\fold{v}$. \qedhere
  \end{itemize}
\end{proof}

The converse is also true: any purely positive proof describes a
unique value of the language (up to $\alpha$-equivalence).

The well-definedness of $\circu(\omega)$ relies on guaranteeing the
following invariant upon each call to the function $\Pos(\xi,\Gamma)$:
{\it $\Gamma$ is a list corresponding to some ordering of the formulas
occurring in the typing context of $\xi$.}

This property is trivially preserved upon each recursive call to
$\Pos(\cdot)$ in \autoref{tab:positive_phase}. Therefore, the only
thing to check is that, at the initial call to $\Pos(\cdot)$, in the
last case of the definition of $\Neg(\cdot)$ in
\autoref{tab:negative_phase}, the invariant is satisfied.

\begin{prop}
\label{form-val-neg}
Let $\omega = \isobasique : A\iso B$ (resp. $\omega = \fix f.
\isobasique : A\iso B$) be a well-typed iso of typing derivation
$\xi_\omega$ and let $(\TypCtxta_i)_{1\leq i \leq n}$ be the typing
contexts given by $\xi_\omega$ such that $\TypCtxta_i \vdash_e v_i
: A$ and $\TypCtxta_i \vdash_e e_i : B$ for $1\leq i \leq
n$. Then for any $(\set{([v_i^1,\dots ,v_i^{k}],\xi_i), i\in I}, [F_1,
\dots F_k], \Gamma, B, \vdash^{\Psi'})$ that is the argument of a
recursive call to $\Neg(\cdot)$ from $\Neg(\set{([v_i],\xi_i), 1\leq
i\leq n}, [A], \emptyset, B, \vdash^\Psi)$ is such that there exists
typing contexts $\TypCtxta_i^j, 1\leq j \leq k$ such that
$\TypCtxta_i^j\vdash v_i^j : F_j$ and $\Gamma,
\overline{\TypCtxta_i^1},\dots ,\overline{\TypCtxta_i^n}$ is a
reordering of the list $\overline{\TypCtxta_i}$.
\end{prop}
\begin{proof}
In the following proof, when this does not introduce ambiguity, we
will indulge ourselves in a slight abuse of notation, writing
$\TypCtxta$ instead of $\overline{\TypCtxta}$ to keep the notations
light. The property is proved by induction on the number $n$ of
recursive calls needed to reach
\[
  (\set{([v_i^1,\dots ,v_i^{k}],\xi_i), i\in I}, [F_1, \dots F_k], \Gamma, B)
\]
in the (finite) tree of recursive calls from
\[
  \Neg(\set{([v_i],\xi_i), 1\leq i\leq n}, [A], \emptyset, B).
\]
\begin{itemize}
\item If $n=0$, then the property is trivial since $k=1$, $\Gamma =
\emptyset$: it is sufficient to set $\TypCtxta_i^1 = \TypCtxta_i$.

\item Assuming that the property holds for any $i\leq n$ and assuming
$\Neg(\set{([v_i^1,\dots ,v_i^{k}],\xi_i), i\in I}, [F_1, \dots
F_k], \Gamma, B, \vdash^{\Psi'})$ is a subtree of height $n+1$. We
reason by case on the last rule the derivation $\pi$ of which
$\Neg(\set{([v_i^1,\dots ,v_i^{k}],\xi_i), i\in I}, [F_1, \dots
F_k], \Gamma, B, \vdash^{\Psi'})$ is a premise.
\begin{itemize}

\item Case of $\parr$. In this case, $\pi$ is some derivation to which
the induction hypothesis applies and, by definition of $\Neg(\cdot)$
it is of the form $\Neg(\set{([\langle v_i^1, v_i^2\rangle,
v_i^3,\dots v_i^{k}],\xi_i), i\in I}, [F_1\otimes F_2,F_3, \dots F_k],
\Gamma, B, \vdash^{\Psi'})$. By induction hypothesis, there exist
$\TypCtxta^0$ and $\TypCtxta_i^j, 3\leq j \leq k$ such that
(i) $\TypCtxta^0\vdash \langle v_i^1,v_i^2\rangle : F_1\otimes F_2$,
(ii) $\TypCtxta_i^j\vdash v_i^j : F_j$ for $3\leq j\leq k$ and
$\Gamma, \TypCtxta^0, \TypCtxta_i^3, \dots, \TypCtxta_i^n$ is
a reordering of $\TypCtxta_i$.

Since $\langle v_i^1,v_i^2\rangle$ is a well-typed value, there exists
$\TypCtxta_i^1, \TypCtxta_i^2$ such that $\TypCtxta^0 = \TypCtxta_i^1,
\TypCtxta_i^2$ and $\TypCtxta_i^1\vdash v_i^1:F_1$ and
$\TypCtxta_i^2\vdash v_i^2: F_2$. As a consequence, $\Gamma,
\TypCtxta_i^1,\dots, \TypCtxta_i^n$ is a reordering or
$\TypCtxta_i$ and $\TypCtxta_i^j\vdash v_i^j : F_j$ for $1\leq j\leq
k$ as expected.

\item Case of $\with$. W.l.o.g we assume that the premise we are
considering is the left premise of the $\with$. Then, $\pi$ is of the
form $\Neg(\set{([\inl {v_i^1}, v_i^2,\dots, v_i^{k}],\xi_i),
i\in K}\cup \set{(\inr{v_j} :: l_j,\xi_j), j\in J}, [F_1\oplus
F'_1,F_2, \dots F_k], \Gamma, B, \vdash^{\Psi'})$. By induction
hypothesis, there exists $\TypCtxta_i^j, 1\leq j \leq k$ such that (i)
$\TypCtxta_i^1\vdash \inl{ v_i^1} : F_1\oplus F'_1$, (ii)
$\TypCtxta_i^j\vdash v_i^j : F_j$ for $2\leq j\leq k$ and $\Gamma,
\TypCtxta_i^1, \TypCtxta_i^2, \dots \TypCtxta_i^n$ is a
reordering of $\TypCtxta_i$.

Since $\inl{ v_i^1}$ is a well-typed value, $\TypCtxta_i^1\vdash {
v_i^1} : F_1$ and $\Gamma,  \TypCtxta_i^1, \dots, \TypCtxta_i^n$ is a
reordering of $\TypCtxta_i$ and $\TypCtxta_i^j\vdash v_i^j : F_j$ for
$1\leq j\leq k$ as expected. The case of the right premise of a
$\with$ is similar.

\item The $\nu$ rule is similar to the above case, just adapted to the
fold constructor. Then, $\pi$ is of the form
$\Neg(\set{([\fold{v_i^1}, v_i^2,\dots ,v_i^{k}],\xi_i), i\in I},
[\mu X. F,F_2, \dots F_k], \Gamma, B, \vdash^{\Psi'})$, with $F_1 =
F[X\leftarrow \mu X. F]$. By induction hypothesis, there exists
$\TypCtxta_i^j, 1\leq j \leq k$ such that (i) $\TypCtxta_i^1\vdash
\fold{ v_i^1} : \mu X. F$, (ii) $\TypCtxta_i^j\vdash v_i^j : F_j$ for
$2\leq j\leq k$ and $\Gamma, \TypCtxta_i^1,\dots , \TypCtxta_i^n$
is a reordering of $\TypCtxta_i$.

Since $\fold{ v_i^1}$ is a well-typed value, $\TypCtxta_i^1\vdash {
v_i^1} : F[X \leftarrow \mu X. F]$ and finally $\Gamma,
\TypCtxta_i^1,\dots , \TypCtxta_i^n$ is a reordering of
$\TypCtxta_i$ and $\TypCtxta_i^j\vdash v_i^j : F_j$ for $1\leq j\leq
k$ as expected.

\item Case of $\top$. Then, $\pi$ is of the form \[\Neg(\set{([(),
 v_i^1, \dots v_i^{k}],\xi_i),i\in I}, [\one, F_1,F_2, \dots F_k],
 \Gamma, B, \vdash^{\Psi'})\] and induction hypothesis ensures that
 there exists $\TypCtxta_i^j, 0\leq j \leq k$ such that (i)
 $\TypCtxta_i^0\vdash () : \one$, (ii) $\TypCtxta_i^j\vdash v_i^j :
 F_j$ for $1\leq j\leq k$ and $\Gamma, \TypCtxta_i^0,\dots
 ,\TypCtxta_i^n$ is a reordering of $\TypCtxta_i$.

By the typing rule for value $()$ we have that $\TypCtxta_i^0$ is
empty from which we can conclude, as expected, that $\Gamma,
\TypCtxta_i^1,\dots ,\TypCtxta_i^n$ is a reordering of
$\TypCtxta_i$.

 \item Case of $ex$. In this case, $\pi$ is of the form
   \[
     \Neg(\set{([x, v_i^1, \dots, v_i^{k}],\xi_i), i\in I}, [F, F_1,F_2, \dots, F_k], \Gamma', B, \vdash^{\Psi'})
   \]
 and induction hypothesis ensuring that there exists $\TypCtxta_i^j,
 0\leq j \leq k$ such that (i) $\TypCtxta_i^0\vdash x : F$, (ii)
 $\TypCtxta_i^j\vdash v_i^j : F_j$ for $1\leq j\leq k$ and $\Gamma',
 \TypCtxta_i^0, \dots, \TypCtxta_i^n$ is a reordering of
 $\TypCtxta_i$.

By the typing rule for value $x$ we have that $\TypCtxta_i^0 = (x:F)$
from which we can conclude, as expected, that $\Gamma,\TypCtxta_i^1,
\dots, \TypCtxta_i^n$ is a reordering of $\TypCtxta_i$ since $F,
\Gamma' = \Gamma$. \qedhere
\end{itemize}
\end{itemize}
\end{proof}

We can now show that our translation is well-defined:

\begin{prop}
    Given a closed iso $\entailiso \omega$ then $\circu(\omega)$ is
    well-defined.
\end{prop}
\begin{proof}
    By induction on $\entailiso$. By typing, the iso is either
    $\fix f. \isobasique$ or $\isobasique$. In both cases, we end up
    in the negative phase, where in the first case, the root of the
    derivation is annotated with $f$, while in the second case, the
    root of the derivation is not annotated with any label.

    The negative phase $\Neg(([v_i], \xi_i')_{i\in I}, [A], \emptyset,
    B, \vdash^\Psi)$ is well-defined due to the predicate $\PE_A$.
    Indeed, whenever $\PE_A(S)$ holds on a well-typed set of
    values $S$ of the same type, all the values in $S$ necessarily
    have the same constructors at the root position.

    The positive phase is well-defined as it consists in recreating
    the typing judgment of each expression $e_i$.
    By~\autoref{form-val-neg} we know that the typing context of $e_i$
    and the one obtained from the negative phase contain the same
    variable associated with the same formula.
\end{proof}

\subsubsection{Validity of the Translation}
We now turn to the validity of the \mumall proof corresponding to a
well-typed iso. Given an iso $\omega = \fix f. \isobasique$, we want
to show that any infinite branch is inhabited by a valid thread.

Before going further, it is useful to notice that the bouncing thread
we encounter in translations of isos has a specific shape, the reason
for which we introduce the notion of \emph{bouncing-cut} and their
origin:

\begin{defi}[Bouncing-Cut]
A Bouncing-cut is a cut of the form:
\[\begin{array}{c}\infer[\cut]{\Gamma \vdash F}{\infer{\Gamma \vdash
G}{\pi}\qquad\infer[be(f)]{G \vdash F}{}}\end{array}\]
\end{defi}

Due to the syntactical restrictions of the language we get straightforwardly the
following fact:

\begin{prop}[Origin of Bouncing-Cut]
  Given a well-typed iso, every occurrence of a rule $be(f)$ in
  $\circu(\omega)$ is a premise of a bouncing-cut.
\end{prop}

In particular, when following a thread going up into a
\emph{bouncing-cut}, it will always start from some negative formula
(represented graphically as a positive formula in the left-hand side
of the sequent, according to our writing convention), before going
back down following a positive formula (represented as {\it the only}
right-hand side formula of the sequent) and bouncing back up on the
bouncing-cut to reach the back-edge immediately.

\medskip
As given by \autoref{lem:inf-branch-form}, an infinite branch is
uniquely defined by a single iso-variable.
Given the value $v_i^j$ of type $\mu X. B$ which is the decreasing
argument for the structurally recursive criterion, we want to build a
pre-thread that follows the variable $x_j : \mu X. B$ in $v_i^j$ that
is the focus of the criterion.

Since we want to show that $\circu(\omega)$ is a valid derivation, we
need to show that for every infinite branch, there exists a valid
thread. Among the multiple threads that exist inside an infinite
branch, we look at the one that follows the formula corresponding to
the decreasing argument of the structurally recursive criterion. We
describe in the following definitions how this thread behaves inside
$\circu(\omega)$. But first, notice that (i) the base case of
$\Neg(-)$ is necessarily a call to $\Pos(-)$ and (ii) when we have a
set of clauses $\isobasique$ there will be exactly $n$ calls to
$\Pos(-)$ for each typing derivation of each $e_i$. This is due to the
fact that branching occurs for each occurrence of a left or right
injection in the $v_i$.

\begin{defi}[Pre-thread of the negative phase]
  \label{def:preT-foc}
  Given a well-typed iso \\ $\omega = \fix f. \isobasique : A\iso
  B$ and a clause $v_i \iso e_i$ such that $f~p$ is a subterm of $e_i$
  and the variable $x$ of type $\mu X. C$, which is the focus of the
  primitive recursive criterion. Given that $x$ occurs free only in
  the clause $v_i \iso e_i$, we define $PT_n(x, (v_i \iso e_i),
  \circu(\omega))$ as the pre-thread that follows the subformula $\mu
  X. C$ until $\Pos(\xi_i)$ is reached, where $\xi_i$ is the typing
  derivation of $e_i$
\end{defi}

Now that we have defined the pre-thread that follows the decreasing
argument during the negative phase, we can look at its weight.

\begin{lem}[Weight of the pre-thread for the negative phase]
\label{lem:form-foc}
Given a well-typed iso $\omega = \fix f. \omega$ with a clause $v \iso
e$ where $x$ is the focus, then $w(PT_n(x, (v\iso e),
\circu(\omega)))$ is a word over $\{l,r,i,\mathcal{W}\}$.
\end{lem}
\begin{proof}
By case analysis on the rule of $\Neg(-)$, reasoning on whether the focus
$x$ is a subterm of the values being decomposed or not:
\begin{itemize}
    \item If the variable $x$ is not a subterm of the first value from
    the list $l$ then the thread has the form:$(A;C,\Delta\vdash B,
    \uparrow)\cdot (A; C',\Delta\vdash B, \uparrow)$ where
    $A\in\Delta$ and the weight is $\mathcal{W}$.
    \item If the variable $x$ is a subterm of the first value of the
    list $l$, then by direct case analysis on the first value: if the
    value is of the form $\pv{v_1}{v_2}$, depending on whether $x$ is
    in $v_1$ or $v_2$ the weight will be $l$ or $r$. The cases
    $\inl{}$ and $\inr{}$ will have for weights $l$ or $r$,
    respectively, while the $\fold{}$ will create weight $i$.
    \item Finally, in the case where the first value of the list is a
    variable $y$ (whether $x = y$ or not), then as we apply the
    {$\mathsf{ex}$} rule, the weight is $\mathcal{W}$. \qedhere
\end{itemize}
\end{proof}

The same reasoning can now be done for the positive phase.

\begin{defi}[Pre-thread of the positive phase]
\label{def:preT-e}
Given a well-typed iso \\ $\omega = \fix f. \isobasique : A\iso B$
and a clause $(v_i\iso e_i)$ where $\xi : \Theta \vdash_e e_i : B$,
such that $f~p$ is a subterm of $e_i$ and $x$ is the focus of the
primitive recursive criterion of type $\mu X. C$. Considering $\pi =
\circu(\omega)$, we define $PT_p(x, \Pos(\xi, \Theta))$ as the
pre-thread that starts from the conclusion of $\Pos(\xi, \Theta)$ and
follows the type $\mu X. C$ until the sequent $A\vdash^f B$ is
reached.
\end{defi}

We are guaranteed that such a pre-thread exists in an iso $\fix f.
\omega$ with clause $v\iso e$ where $f~p$ occurs in $e$ since (i) the
focus of the structurally recursive criterion, $x$, is used linearly
in $e$, (ii) $x$ is contained in the variables of $p$ and therefore,
the pre-thread will start by going up on the left-hand side of the
sequent until it bounces back on an axiom, and then go down and
bounces back up on the {\cut} generated by the subterm $f~p$. This is
formalized by the following lemma, where we look at the weight of the
positive phase:

\begin{lem}[Weight of the pre-thread for the positive phase]
\label{lem:weight-positive-phase}
Given a well-typed iso $\fix f. \omega$ such that the variable $x$
is the focus of the primitive recursive criterion in a clause $v\iso
e$ where $\xi : \Theta \vdash_e e : A$ and where $f~p$ is a subterm of
$e$, the weight of the pre-thread on the positive phase $PT_p(x,
\Pos(\xi, \Theta))$ is of the shape $\mathcal{W}^*
\mathcal{A}\{\overline{l}, \overline{r}, \mathcal{W}\}^* \mathcal{C}$
\end{lem}
\begin{proof}
By case analysis on $\Pos(e)$. As the thread starts on the left-hand
side (since the formula is in $\Theta$), it only goes up by
encountering cut-rules or right-rules, therefore we get
$\mathcal{W}^*$, and the thread goes up all the way to an axiom rule.
This axiom necessarily exists and is unique since the type system is
linear corresponding to the variable $x$. This axiom rule therefore
adds the $\mathcal{A}$ and bounce back. Finally, the thread goes down
on the purely positive proof, which, by syntactical definition of an
iso is a pattern $p$, and hence generating $\{\overline{l},
\overline{r}, \mathcal{W}\}^*$. The thread keeps going down as such
until bouncing again on the {\cut} rule generated by the subterm
$f~p$.
\end{proof}

We can then consider the infinite pre-thread as the concatenation of
both the pre-thread from the negative phase, and the one from the
positive phase.

\begin{lem}[Form of the pre-thread]
\label{lem:thread-form}
Given the pre-thread $t$ following the focus is the primitive
recursive criterion $x$ (meaning the composition of $PT_n(-)$ and
$PT_p(-)$) we have that $w(t) = p_0 (\Sigma_{i\leq n} p_i
\mathcal{W}_i^* \mathcal{A} q_i \mathcal{C})^\omega$ with $p_0$ any
prefix, $p_i \in \{l,r,i, \mathcal{W}\}^*$ and $q_i \in
\{\overline{l}, \overline{r}, \mathcal{W}\}^*$ and with, $\forall
i\leq n, \overline{q_i} \sqsubset \overline{\overline{p_i}}$ and $\mid
p_i \mid > \mid q_i \mid$ without counting the $\mathcal{W}$, where $p
\sqsubset q$ is $q$ is a prefix of $p$ and with
$\overline{a\cdot p} = \overline{a} \cdot \overline{p}$
if $a\in\{l,r,i\}$, $\overline{\overline{a}p} =
a\overline{p}$ if $a\in \{l, r, i\}$ and
$\overline{\mathcal{W}p} = \overline{p}$.
\end{lem}
\begin{proof}
$p_i$ is generated from~\autoref{def:preT-foc} while the rest up to
the $C$ (included) is generated from~\autoref{def:preT-e}.

First, we show that $\mid p_i \mid > \mid q_i \mid$ modulo the
$\mathcal{W}$.

Since $p_i$ is generated by the negative phase, we have that, modulo
$\mathcal{W}$, $p_i = r^*l?\set{l,r,i}^*$, this is due
to the definition of being primitive recursive and because we
are following the focus of the primitive recursion criterion. By
definition of being primitive recursive the input type of the iso is
$A_1\otimes (\dots \otimes A_n)$, hence $r^*l?$ correspond to
the search for the correct $A_i$ of type $\mu X. B$ in the
input type, while $\set{l,r,i}^*$ is the
decomposition of the primitive recursive value, as described
in~\autoref{form-val-neg}.

As $q_i$ corresponds to the purely positive proof, and by syntactical
definition of the iso, we know that the purely positive proof is the
encoding of a pattern $p = \pv{x_1}{\pv{\dots}{x_n}}$. Hence, $q_i$
can be decomposed as $\ov{l}?~\ov{r}^*$. By the fact that the
iso is primitive recursive we know that the variable in $p$ is a
strict subterm of the primitive recursive value, hence $\mid p_i \mid
> \mid q_i \mid$.

The fact that $\overline{q_i} \sqsubset \overline{\overline{p_i}}$ is
direct, as the purely positive proof reconstructs the type $A_1
\otimes \dots \otimes A_n$ without modifying the $A_i$ while
$\overline{\overline{p_i}}$ starts by searching for the corresponding
type $A_i$, so it is only composed of $\{l, r\}^*$, which will be the
same as $\overline{q_i}$.
\end{proof}

\begin{thm}[The pre-thread generated is a thread]
\label{thm:pre-t-ok}
Given a well-typed iso $\fix f. \omega : A\iso B$, and its
corresponding pre-proof $\pi$, the pre-thread following the focus of
the recursive criteria of formula $\mu X. C$ is a thread, i.e. it can be uniquely
decomposed into $\bigodot (H_i \odot V_i)$ with:
\begin{itemize}
    \item $w(V_i) \in \{l, r, i, \mathcal{W}\}^\infty$ and non-empty
    if $i \not= \lambda$
    \item $w(H_i) \in \goth{H}$, and
    it is non-empty if $i\not= 0$
\end{itemize}
\end{thm}
\begin{proof}
  We set $H_0$ as $(\mu X. C; A\vdash B; \up)$ (so that $w(H_0) =
  \epsilon \in \goth{H}$). $V_0$ consists of the entire negative phase
  followed by the positive phase until the axiom rule is reached.
  Therefore, $V_0$ starts by ($\mu X. C; A \vdash B; \up)$ hence the
  first elements of $H_0$ and $V_0$ coincide.
  From~\autoref{lem:form-foc}, we know that $w(V_0)$ is a word over
  $\set{l, r, i, \mathcal{W}}$ followed by the $\mathcal{W}^*$ that
  comes from the positive phase, as described
  in~\autoref{lem:weight-positive-phase}. Then, for all $i\geq 1$, we
  set:
  \begin{itemize}
    \item $H_i$ starts at $(\mu X. C; \mu X. C\vdash \mu X. C; \up)$,
    just before the axiom rule so that the first element of $w(H_i)$
    is $\mathcal{A}$. Then $H_i$ is composed of:
    \begin{itemize}
      \item All of the pre-thread going down on the Purely Positive
      Proof after the axiom, accumulating a word over
      $\{\overline{l},\overline{r}, \overline{i}, \mathcal{W}\}^*$;
      \item Going back up into the cut-rule of the bouncing cut,
      making a $\mathcal{C}$;
      \item Going up to compensate for every $\overline{x}$ seen in
      the Purely Positive Proof while going down. This is possible,  as
      shown in~\autoref{lem:thread-form}.
    \end{itemize}
    \item $V_i$ is the maximal possible sequence such that $w(V_i) \in
    \{l,r,i, \mathcal{W}\}^*$, i.e., the sequence that ends with $(A;
    A\vdash A; \up)$. \qedhere
  \end{itemize}
\end{proof}

\begin{thm}[\emph{Validity of proofs}]\label{thm:validity} If
  $~\entailiso \omega : A\iso B$, then ${\circu(\omega)}$ satisfies
  {\mumall} bouncing validity criterion.
\end{thm}
\begin{proof}

  By~\autoref{thm:pre-t-ok}, we know that we have a thread of which we
  also know that the visible part is not stationary.

  Finally, by~\autoref{lem:thread-form} and~\autoref{thm:pre-t-ok} we
  know that the visible part will infinitely often encounter the
  subformulas of the formula $\mu X. B$ that is the focus of the
  primitive recursive criterion. This is due to the difference in size
  in the part of the thread from the negative and from the positive
  phase and the fact that the positive phase does not encounter a
  $\mu$ formula when going down on a purely positive proof.

  Therefore, the smallest formula we will encounter is a $\nu$
  formula, validating the thread.
\end{proof}

\begin{rem}
  In the last proof, we used a notion of validity different
 from~\cite{bouncing}. While the one from~\cite{bouncing} is more
 general, both criteria coincide on derivations that are images of
 isos.
\end{rem}

\subsection{Proof Simulation}
We have shown that any iso from type $A\iso B$ gives rise to a valid
proof of conclusion $A\vdash B$. In order to show the relationship
between the cut-elimination procedure of {\mumall} and the rewriting
system of our language we first introduce a slightly modified version
of the iso rewriting system, which we call $\expB$, using explicit
substitution in order to make the iso-evaluation closer to the
cut-elimination procedure.

\begin{rem}
  \label{convention:letsigma}
  In the remaining of this section, we fix the following convention:
	given a substitution $\sigma = \set{x_1\mapsto v_1, \dots,
	x_n\mapsto v_n}$ (and assuming a linear order on variables) we use
	the shorthand $\letsig{\sigma}{t}$ for $\letv{x_1}{v_1}{\dots
	\letv{x_n}{v_n}{t}}$.
\end{rem}

\begin{defi}[Explicit Substitution Rewriting System]
    $\expB$ is defined by the following rules :
    \begin{align*}
        \letv{x}{v}{x} &\expLet v \\
        \letv{\pv{p_1}{p_2}}{\pv{t_1}{t_2}}{t} &\expLet \letv{p_1}{t_1}{\letv{p}{t_2}{t}}\\
        \letv{x}{v}{\pv{t_1}{t_2}} &\expLet \pv{\letv{x}{v}{t_1}}{t_2} &&\text{when } x\in FV(t_1) \\
        \letv{x}{v}{\pv{t_1}{t_2}} &\expLet \pv{t_1}{\letv{x}{v}{t_2}} &&\text{when } x\in FV(t_2) \\
        \letv{x}{v}{\inl{t}} &\expLet \inl{(\letv{x}{v}{t})} \\
        \letv{x}{v}{\inr{t}} &\expLet \inr{(\letv{x}{v}{t})} \\
        \letv{x}{v}{\fold{t}} &\expLet \fold{(\letv{x}{v}{t})} \\
        \letv{x}{v}{\omega~t} & \expLet \omega~(\letv{x}{v}{t})
    \end{align*}
    and the following rules (with $C$ an arbitrary context):
    \[
        \begin{array}{c}
        \infer[\beta\mathrm{-Cong}]{C[t]\expB C[t']}{t \expB\cup\expLet t'}
        \\[1.5ex]
        \infer[\beta\mathrm{-IsoRec}]{
        (\fix f.\isoterm)v \expB (\omega[f \leftarrow (\fix f. \omega)])v
        }{
        }\\[1.5ex]
        \infer[\beta\mathrm{-IsoApp}]{\isobasique~v' \expB \letsig{\sigma}{e_i}}{\sigma[v_i] = v'}
        \end{array}
    \]

    As always, we use $\expB^*$ for the reflexive-transitive closure
    of $\expB$.
\end{defi}

\begin{prop}
  The following rule
  \[
  \infer[\beta-\mathrm{LetE}]{\letv{p}{v}{t}\expB\letsig{\sigma}{t}}{\sigma[p]
  = v} \] is admissible.
\end{prop}
\begin{proof}
  We have two cases depending on the shape of $p$. If it is a simple
  variable then the case is direct; otherwise, $p = \pv{p_1}{p_2}$,
  then, by typing, $t$ reduces to $\pv{t_1}{t_2}$, in which case one
  can apply the second rule of $\expLet$ and reason inductively on the
  two subterms until we obtain $\letsig{\sigma}{t}$.
\end{proof}

Each step of this rewriting system corresponds to exactly one step of
cut-elimination, whereas in the previous system, the rewriting using a
substitution $\sigma$ corresponds to multiple steps of
cut-elimination. $\expB$ makes this explicit, clarifying the
relationship between logic and reversible computation. Before
establishing this result, we first need to ensure that both rewriting
systems have the same operational semantics:
\begin{lem}[Specialization of the substitution on pairs]
\label{lemma:split_subs}
Let $\sigma$ be a substitution that closes $\Delta \vdash_e
\pv{t_1}{t_2}$, then there exists $\sigma_1, \sigma_2$, such that
$\sigma(\pv{t_1}{t_2}) = \pv{\sigma_1(t_1)}{\sigma_2(t_2)}$ Where
$\sigma = \sigma_1 \cup \sigma_2$.
\end{lem}

\begin{proof}
By the linearity of the typing system we know that $FV(t_1) \cup
FV(t_2) = \emptyset$, so there always exists a decomposition of
$\sigma$ into $\sigma_1, \sigma_2$ defined as $\sigma_i = \{(x_i
\mapsto v_i) \mid x_i \in FV(t_i)\}$ for $i\in\{1,2\}$.
\end{proof}

\begin{lem}[Explicit substitution and substitution coincide]
\label{lem:sub-eB}
Let $\sigma = \{x_i \mapsto v_i\}$ be a substitution that closes $t$,
then $\letsig{\sigma}{t} \expLet^* \sigma(t)$.
\end{lem}

\begin{proof}
By induction on $t$.
\begin{itemize}
    \item If $t=x$, then $\sigma(x) = v$ and $\letv{x}{v}{x} \expLet v
    = \sigma(x)$.
    \item If $t=()$ then $\sigma$ is empty and no substitution applies.
    \item If $t=\pv{t_1}{t_2}$, then by~\autoref{lemma:split_subs}
    $\sigma(\pv{t_1}{t_2}) = \pv{\sigma_1(t_1)}{\sigma_2(t_2)}$. By
    $\expLet$, each $\lett$ construction will enter either $t_1$ or
    $t_2$, we can then conclude by using the induction hypothesis.
    \item If $t=\letv{p}{t_1}{t_2}$ is similar to the product case.
    \item If $t$ is $\inl{t'}, \inr{t'}, \fold{t'}$ or $\omega~t'$.
    All cases are treated in the same way: by definition of $\expLet$,
    each $\lett$ will enter into the subterm $t'$, as with the
    substitution $\sigma$, we then conclude by applying the induction
    hypothesis. \qedhere
\end{itemize}
\end{proof}

\begin{cor}
    If $t\to t'$ then $t\expB^* t'$.
\end{cor}
\begin{proof}
    By induction on $\to$, the case of $\mathrm{IsoRec}$ is the same,
    while for the other rules, just by applying either
    $\beta-\mathrm{IsoApp}$ or $\beta-\mathrm{LetE}$ and then
    by~\autoref{lem:sub-eB}, $\letsig{\sigma}{t} \expLet^* \sigma(t)$
    for any substitution $\sigma$ that closes $t$. But $\sigma(t) =
    t'$, so $t\expB^* t'$.
\end{proof}

It is then possible to show the first step of the simulation
procedure, namely that $\expB$ corresponds to one step of
cut-elimination:
\begin{lem}[Simulation of the let-rules of $\expB$]~
    \label{thm:sim}
    Let $\xi : {\TypCtxta} \vdash_e t : A$ be a well-typed closed
    term. If $t\expLet t'$, by \autoref{prop:subject_reduction} we
    have that there exists a typing derivation $\xi' : \TypCtxta
    \vdash_e t' : A$ and then we get: $\Pos(\xi, \Gamma)
    \cute^{*} \Pos(\xi', \Gamma)$ for any $\Gamma$ which is
    a list ordering the formulas occurring in $\TypCtxta$.
\end{lem}
\begin{proof}
    By case analysis on \expLet, for simplicity of reading, we
    will write directly $\Pos(t, \Gamma)$ for $\Pos(\xi, \Gamma)$ when
    $\xi$ is a typing derivation of $t$ and is clear from the
    context.

\begin{itemize}[label={$\bullet$}]
  \setlength\itemsep{1em}
    \item $\letv{x}{v}{x} \expLet v$.
    \[\begin{array}{c}\infer[\cut]{\Gamma \vdash A}{\infer{\Gamma \vdash
    A}{{\Pos(v, \Gamma)}} \qquad \infer[\id]{A\vdash A}{}}\end{array} \cute
    \begin{array}{c}\infer{\Gamma \vdash A}{{\Pos(v, \Gamma)}}\end{array}\]

    \item $\letv{\pv{p_1}{p_2}}{\pv{t_1}{t_2}}{t} \expLet
    \letv{p_1}{t_1}{\letv{p_2}{t_2}{t}}$. Denote by $\xi$ the typing
    derivation of $\letv{p_2}{t_2}{t}$ and by $\xi'$ the typing
    derivation of $t$, then by definition of the negative phase we
    know that $\Neg(\set{[p_1], \xi}, [A\otimes B], \Gamma_2\cup\Gamma_3, B,
    \vdash)$ consists only in a succession of $\parr$ and $ex$ rules.
    We note $\Delta_1$ the context obtained by decomposing the type
    $A$ through the negative phase, then:
    \[
      \begin{prooftree}
        \hypo{\Pos(t_1, \Gamma_1)}
        \infer1{\Gamma_1 \vdash A}
        \hypo{\Pos(t_2, \Gamma_2)}
        \infer1{\Gamma_2 \vdash B}
        \infer2[$\otimes$]{\Gamma_1, \Gamma_2 \vdash A\otimes B}
        \hypo{\Neg(\set{([p_1 :: p_2], \xi')},
        [A, B], \Gamma_3, C, \vdash)}
        \infer1{A, B, \Gamma_3 \vdash C}
        \infer1[$\parr$]{A\otimes B, \Gamma_3 \vdash C}
        \infer2[\cut]{\Gamma_1, \Gamma_2, \Gamma_3 \vdash C}
      \end{prooftree}
   \cute^{*}\]
\[\begin{prooftree}
        \hypo{{\Pos(t_1, \Gamma_1)}}
        \infer1{\Gamma_1\vdash A}

        \hypo{{\Pos(t_2, \Gamma_2)}}
        \infer1{\Gamma_2\vdash B}

        \hypo{{\Neg(\set{([p_2], \xi')}, B, \Gamma_3\cup\Delta_1, C,\vdash)}}
        \infer1{B, \Gamma_3, \Delta_1\vdash C}
        \infer2[\cut]{\Gamma_2,\Gamma_3, \Delta_1 \vdash C}
        \infer1[$\parr$ / \rex]{\Neg(\set{([p_1], \xi)}, A, \Gamma_2\cup\Gamma_3, C, \vdash)}
        \infer2[\cut]{\Gamma_1,\Gamma_2,\Gamma_3\vdash C}
      \end{prooftree}
    \]

    In the first proof, we know that $\Neg(\set{([p_1 :: p_2], \xi')},
    [A, B], \Gamma_3, C, \vdash)$ will start by applying only $\parr$
    rules and $\rex$ rules until $A$ is fully decomposed into
    $\Delta_1$, before doing the same with $B$. Because
    those are negative rules, we know that we can
    commute them, hence obtaining the same proofs.

    \item $\letv{x}{v}{\inl{(t)}} \expLet \inl{(\letv{x}{v}{t})}$.
      \[
    \begin{prooftree}

      \hypo{\Pos(v, \Gamma_1)}
      \infer1{\Gamma_1 \vdash C}

      \hypo{\Pos(t, \Gamma_2 :: C)}
      \infer1{\Gamma_2, C \vdash A}
      \infer1[$\oplus^1$]{\Gamma_2, C \vdash A \oplus B}
      \infer1[\rex]{C, \Gamma_2 \vdash A\oplus B}
      \infer2[\cut]{\Gamma_1, \Gamma_2 \vdash A \oplus B}
    \end{prooftree}~\cute^{*}
    \begin{prooftree}

      \hypo{\Pos(t, \Gamma_1)}
      \infer1{\Gamma_1 \vdash C}

      \hypo{\Pos(t, \Gamma_2:: C)}
      \infer1{\Gamma_2, C \vdash A}
      \infer1[\rex]{C, \Gamma_2 \vdash A}
      \infer2[cut]{\Gamma_1, \Gamma_2 \vdash A}
      \infer1[$\oplus^1$]{\Gamma_1, \Gamma_2\vdash A\oplus B}
    \end{prooftree}\]

    \item The same goes for $\letv{x}{v}{\inr{(t)}} \expLet
    \inr{(\letv{x}{v}{t})}$ and for
    $\letv{x}{v}{\fold{(t)}}\expLet
    \fold{(\letv{x}{v}{t})}$

    \item $\letv{x}{v}{\pv{t_1}{t_2}}\expLet\pv{\letv{x}{v}{t_1}}{t_2}$
    when $x\in FV(t_1)$.
    Then:

    \[
    \begin{array}{c}
      \begin{prooftree}
        \hypo{\Pos(v, \Gamma_1)}
        \infer1{\Gamma_1 \vdash A}
        \hypo{\Pos(t_1, \Gamma_2 :: A)}
        \infer1{\Gamma_2, A\vdash C}
        \hypo{\Pos(t_2, \Gamma_3)}
        \infer1{\Gamma_3\vdash B}
        \infer2[$\otimes$]{\Gamma_2, \Gamma_3, A\vdash C\otimes B}
        \infer1[\rex]{A, \Gamma_2, \Gamma_3 \vdash C \otimes B}
        \infer2[\cut]{\Gamma_1, \Gamma_2, \Gamma_3 \vdash C\otimes B}
      \end{prooftree}
    \end{array} \cute\]

\[
      \begin{array}{c}\infer[\otimes]{\Gamma_1,\Gamma_2,\Gamma_3\vdash
      C\otimes B}{\infer[\cut]{\Gamma_1,\Gamma_2\vdash
      C}{\infer{\Gamma_1\vdash A}{{\Pos(v, \Gamma_1)}} \qquad
      \infer[\rex]{A, \Gamma_2\vdash C}{\infer{\Gamma_2, A\vdash
      C}{{\Pos(t_1, \Gamma_2 :: A)}}}} \qquad
        \infer{\Gamma_3\vdash B}{{\Pos(t_2, \Gamma_3)}}}\end{array}
    \]

    \item Similar for the second rule on the pair.

    \item $\letv{x}{v}{\omega~t} \expLet \omega~(\letv{x}{v}{t})$.
    Then:

    \begin{prooftree}
      \hypo{{\Pos(v, \Gamma_1)}}
      \infer1{\Gamma_1\vdash A}

      \hypo{{\Pos(t, \Gamma_2 :: A)}}
      \infer1{\Gamma_2, A\vdash B}
      \infer1[\rex]{A, \Gamma_2\vdash B}

      \hypo{{\circu(\omega)}}
      \infer1{B\vdash C}
      \infer2[\cut]{A, \Gamma_2\vdash C}
      \infer2[\cut]{\Gamma_1, \Gamma_2\vdash C}
    \end{prooftree} $\cute$

    $\begin{array}{c}\infer[\cut]{\Gamma_1,\Gamma_2\vdash
    C}{\infer[\cut]{\Gamma_1,\Gamma_2\vdash B}{\infer{\Gamma_1\vdash
    A}{{\Pos(v, \Gamma_1)}}\qquad\infer[\rex]{A, \Gamma_2\vdash
    B}{\infer{\Gamma_2, A\vdash B}{{\Pos(t,
    \Gamma_2 :: A)}}}}\qquad\infer{B\vdash
    C}{{\circu(\omega)}}}\end{array}$ \qedhere
    \end{itemize}
\end{proof}

We can now show how the pattern-matching is captured by the
cut-elimination:

\begin{lem}
    \label{lem:pattern-matching-proof}
  Given a set of well-typed values and expressions $\TypCtxta_i
  \vdash_e v_i : A$ and $\xi_i : \TypCtxta_i \vdash_e e_i : B$ for
  $i\in\set{1, \dots, n}$ such that $\PE_A(v_1, \dots, v_n)$ and
  $\PE_B(\Val{e_1}, \dots, \Val{e_n})$, then, given a well-typed value
  $\TypCtxtb \vdash_e v : A$ such that there exists $j \in \set{1,
  \dots, n}$ and $\sigma$ such that $\sigma[v_j] = v$ then, given
  $\Gamma = \overline{\TypCtxtb}$, we have:
 \[
\pi =     \scalebox{1}{\begin{prooftree}
        \hypo{{\Pos(v, \Gamma)}}
        \infer1{\Gamma\vdash A}
        \hypo{{\Neg(\set{([v_i], \xi_i)_{i\in \set{1, \dots, n}}}, [A], \emptyset,B,\vdash)}}
        \infer1{A\vdash B}
        \infer2[\cut]{\Gamma\vdash B}
    \end{prooftree}} \cute^*
    \scalebox{1}{$
    \infer{\Gamma\vdash B}{\Pos(\letsig{\sigma}{e_j}, \Gamma)}
    $} = \pi'\]
\end{lem}
\begin{proof}
    By induction on the derivation of $\PE_A(\set{v_1, \dots,
    v_n})$:

    \begin{itemize}
        \setlength\itemsep{1em}
        \item Case $\PE_A(\set{x})$. We get $\sigma[x] = v$ then $\pi
        = \pi'$.

        \item Case $\PE_\one(\set{()})$. In this case, $\sigma[()] =
        ()$ and there is a single $e$ with typing derivation $\xi$.

          \medskip
        $\pi = \begin{array}{c}       \begin{prooftree}

          \infer0[$\one$]{\vdash \one}

          \hypo{{\Pos(\xi, \emptyset)}}
          \infer1{\vdash B}
          \infer1[$\bot$]{\one\vdash B}
          \infer2[\cut]{\vdash G}
      \end{prooftree}\end{array}$ which reduces to $\begin{array}{c}\infer{\vdash
        B}{{\Pos(e, \emptyset)}}\end{array} = \pi'$ as $\sigma$ is empty.

        \item Case $\PE_{\mu X.
        A}(\set{(\fold{v_i})_{i\in I}})$ such that
        $\sigma[\fold{v_j}] = \fold{v'}$. Then

        \[\pi =
        \begin{prooftree}
          \hypo{\Pos(v', \Gamma)}
          \infer1{\Gamma\vdash A[X\leftarrow\mu X. A]}
          \infer1[$\mu$]{\Gamma\vdash\mu X. A}
          \hypo{\Neg(\set{([v_i], \xi_i)_{i\in I}}, [A[X \leftarrow
          \mu X. A]], \emptyset, B, \vdash)}
           \infer1{H[X\leftarrow\mu
          X. H]\vdash B}
           \infer1[$\nu$]{\mu X. H \vdash B}
          \infer2[\cut]{\Gamma\vdash B}
        \end{prooftree}
        \]
will reduce to
        \[
        \begin{prooftree}
          \hypo{\Pos(v', \Gamma)}
          \infer1{\Gamma\vdash A[X\leftarrow\mu X. A]}
          \hypo{\Neg(\set{([v_i], \xi_i)_{i\in I}}, [A[X\leftarrow \mu X. A]],\emptyset,B,\vdash)}
          \infer1{A[X\leftarrow\mu X. A]\vdash B}
          \infer2[\cut]{\Gamma \vdash B}
        \end{prooftree}
        \]
        then we can apply our induction hypothesis.

        \item Case $\PE_{A\oplus B}(\set{(\inl{v_i})_{i\in
        I}}
        \cup
        \set{(\inr{v_k})_{k\in K}})$ such that
        $\sigma[\inl{v_j}] = \inl{v'}$. Then the proof
        \[
          \scalebox{.9}{$\pi = \begin{prooftree}
            \hypo{{\Pos(v', \Gamma)}}
            \infer1{\Gamma\vdash A}
            \infer1[$\oplus^1$]{\Gamma\vdash A\oplus B}

            \hypo{{\Neg(\set{([v_i], \xi_i)_{i\in I}},[A],\emptyset,C,\vdash)}}
            \infer1{H\vdash C}

            \hypo{{\Neg(\set{([v_k], \xi_k)_{k\in K}}, [B], \emptyset,C,\vdash)}}
            \infer1{B\vdash C}
            \infer2[$\with$]{A\oplus B\vdash C}
            \infer2[\cut]{\Gamma\vdash C}
        \end{prooftree}$}
        \]
        reduces to
        \[
        \scalebox{.9}{\begin{prooftree}
            \hypo{{\Pos(v', \Gamma)}}
            \infer1{\Gamma\vdash A}

            \hypo{{\Neg(\set{([v_i], \xi_i)_{i\in I}},[A],\emptyset,C,\vdash)}}
           \infer1{A\vdash C}
            \infer2[\cut]{\Gamma\vdash C}
          \end{prooftree}}
        \]
        then we can apply our induction hypothesis.

        \item The case $\PE_{A\oplus
        B}(\set{\inl{v_i}}\cup\set{\inr{v_k}})$ such that
        $\sigma[\inl{v_j}] = \inl{v'}$ is similar to the previous
        case.

        \item Case $\PE_{A\otimes
        B}(\set{(\pv{v_i^1}{v_i^2})_{i\in I}})$ with
        $\sigma[\pv{v_j^1}{v_j^2}] = \pv{v_1'}{v_2'}$ and
        therefore $\sigma = \sigma_1 \cup \sigma_2$. Then
        \[
          \pi = \begin{prooftree}

            \hypo{{\Pos(v_1', \Gamma_1)}}
            \infer1{\Gamma_1\vdash A}

            \hypo{{\Pos(v_2', \Gamma_2)}}
            \infer1{\Gamma_2\vdash B}
            \infer2[$\otimes$]{\Gamma_1,\Gamma_2\vdash A\otimes B}

            \hypo{{\Neg(\set{([\pv{v_i^1}{v_i^2}], \xi_i)_{i\in I}},[A, B],\emptyset, C,\vdash)}}
            \infer1{A, B \vdash C}
            \infer1[$\parr$]{A\otimes B\vdash C}
            \infer2[\cut]{\Gamma_1,\Gamma_2\vdash C}
        \end{prooftree}\]
        reduces to
        \[
        \begin{prooftree}
            \hypo{{\Pos(v_1', \Gamma_1)}}
            \infer1{\Gamma_1\vdash
            A}
            \hypo{{\Pos(v_2', \Gamma_2)}}
            \infer1{\Gamma_2\vdash B}
            \hypo{{\Neg(\set{((v_i^1 :: v_i^2)_i,
            \xi_i)_{i\in I}},[A, B],\emptyset, C,\vdash)}}
            \infer2[\cut]{\Gamma_2, A\vdash C}
            \infer2[\cut]{\Gamma_1,\Gamma_2\vdash C}
        \end{prooftree}\]

        Because the negative phase on $[v_1, v_2]$ only produces
        $\with, \parr, \top, \nu$ rules, we get that $\Neg(\set{(v_i^1 ::
        v_i^2,
            \xi_i)_{i\in I}},[A, B],\emptyset, C,\vdash) = \Neg(\set{(v_i^2 :: v_i^1,
            \xi_i)_{i\in I}},[B, A],\emptyset, C,\vdash)$
        up to the
        commutation of rules of Linear Logic. Therefore, we can get
        \[
        \begin{prooftree}
            \hypo{{\Pos(v_1', \Gamma_1)}}
            \infer1{\Gamma_1\vdash A}
            \hypo{{\Pos(v_2', \Gamma_2)}}
            \infer1{\Gamma_2\vdash B}
            \hypo{{\Neg(\set{(v_i^2 :: v_i^1,
            \xi_i)_{i\in I}},[B,A],\emptyset, C,\vdash)}}
            \infer1[\rex]{A, B\vdash C}
            \infer2[\cut]{\Gamma_2, A\vdash C}
            \infer2[\cut]{\Gamma_1,\Gamma_2\vdash C}
        \end{prooftree}
        \]
        which, by induction hypothesis on $v_2$, reduces to
        \[
        \begin{prooftree}
            \hypo{{\Pos(v_1', \Gamma_1)}}
            \infer1{\Gamma_1\vdash A}

            \hypo{{\Neg(\set{(v_i^1, \letsig{\sigma_2}{\xi_i})_{i\in I}},[A],\emptyset, C,\vdash)}}
           \infer1{\Gamma_2, A \vdash C}
            \infer2[\cut]{\Gamma_1,\Gamma_2\vdash C}
        \end{prooftree}
        \]
        and then we can apply our second induction hypothesis on
        $v_1$. \qedhere
    \end{itemize}
\end{proof}

We can then conclude the proof with the global simulation theorem as a
direct implication of the two previous lemmas:

\begin{thm}[Iso-substitution cut-elimination]
    \label{isos-mumall:thm:iso-sub-cut-elim}
    Let $\vdash \isobasique~v \to \sigma(e_i)$ when $\sigma[v_i] =
    v$ then ${\Pos(\isobasique~v, [])} \cute^*
    {\Pos(\letsig{\sigma}{e_i}, [])} \\ \cute^* {\Pos(\sigma(e_i),
    [])}$.
\end{thm}
\begin{proof}
    By~\autoref{lem:sub-eB}, we know that the explicit substitution and
    the substitution coincide. \autoref{thm:sim} tells us that one step
    of $\expLet$ is simulated by exactly one step of cut-elimination and
    finally~\autoref{lem:pattern-matching-proof} tells us that we
    simulate properly the pattern-matching.
\end{proof}

\begin{cor}[Simulation]
    \label{col:sim-iso}
Given an iso $\entailiso \omega : A \iso B$ and
values $\vdash_e v : A$ and $\vdash_e v' : B$ such that $\omega~v
\to^* v'$,  we have that $\Pos(\omega~v, []) \cute^* \Pos(v', [])$
\end{cor}
\begin{proof}
    Direct application of~\autoref{isos-mumall:thm:iso-sub-cut-elim}.
\end{proof}

We can wonder about the converse: given a term $\omega~v$ and its
corresponding proof $\pi = \Pos(\omega~v, [])$, if $\pi \cute^* \pi'$
and $\omega~v \to^* v'$, does $\Pos(v', []) = \pi'$ ? The answer to
that question is yes. Since $\Pos(v', [])$ is valid and is the
translation of an iso that terminates, we know that it reaches a
purely positive proof of $\vdash B$. This cut-free proof is
necessarily unique. By the simulation, we know that $\Pos(v', [])$
reduces to another purely positive proof of $\vdash B$, but by
uniqueness those two proofs must be the same.

This leads to the following corollary:

\begin{cor}[Isomorphism of proofs.]
Given a well-typed iso $\entailiso \omega : F\iso G$ and two
well-typed closed values $v_1$ of type $F$ and $v_2$ of type $G$ and
the proofs $\pi : F\vdash G$, $\pi^\bot : G\vdash F$, $\phi :~\vdash
F$, $\psi :~\vdash G$ corresponding respectively to the translation of
$\omega, \omega^\bot, v_1, v_2$ then:

\[ \scalebox{0.8}{$\infer{\vdash F}{\phi}~
            \rotatebox[origin=c]{180}{$\rightsquigarrow$}
            \infer[\cut]{\vdash F}{\infer[\cut]{\vdash G}{\infer{\vdash
            F}{\phi}\qquad\infer{F\vdash
            G}{\pi}}\qquad\infer{G\vdash F}{\pi^\bot}} \qquad
    \infer[\cut]{\vdash G}{\infer[\cut]{\vdash F}{\infer{\vdash
            G}{\psi}\qquad\infer{G\vdash
            F}{\pi^\bot}}\qquad\infer{F\vdash G}{\pi} }
            \rightsquigarrow \infer{\vdash G}{\psi}$} \]
\end{cor}
\begin{proof}

    As a direct implication of~\autoref{thm:isos-iso}
    and~\autoref{col:sim-iso}.
\end{proof}

We should note that this is a specific notion of proof-isomorphisms.
Usually, one would expect the proofs corresponding to $\omega$ and
$\omega^\bot$ reduce to the identity proof. However, it is not clear
how to prove this formally, and it is left for future work.

\section{Conclusion}\label{section:conclusion}

\paragraph*{Summary of the contribution.}
We presented a linear, reversible language with inductive types. We
showed how ensuring non-overlapping and exhaustivity is enough to
ensure the reversibility of the isos. The language comes with both an
expressivity result that shows that any Primitive Recursive Functions
can be encoded in this language as well as an interpretation of
programs into {\mumall} proofs. The latter result rests on the fact
that our isos are \emph{structurally recursive}.

\begin{figure}
  Define the proof $\pi_0$ and $\pi_S$ of the natural number $0$ and the
  Successor as:\\[0.5em] $\pi_0 = \begin{array}{c} \infer[\mu]{\vdash
  \natT}{\infer[\oplus^1]{\vdash \one\oplus\natT}{\infer[\one]{\vdash\one}{}}}
  \end{array}$ and $\pi_S = \begin{array}{c}
  \infer[\mu]{\natT\vdash\natT}{\infer[\oplus^2]{\natT\vdash\one\oplus\natT}{\infer[\id]{\natT\vdash\natT}{}}}
  \end{array}$.

  Then the proof of $\omega_1, \omega_2$ and $\opn{CantorPairing}$ are: \\[0.5em]

  \newcommand{\conclu}{\ensuremath{(\natT\otimes\natT) \oplus \one}}
  \newcommand{\non}{\ensuremath{\natT \otimes \natT}}

  \[
  \begin{array}{lll}
  \pi_{\omega_1} &~=~& \left\{\scalebox{.6}{\begin{prooftree}
    \infer0[$\one$]{\vdash \one}
    \infer1[$\oplus^2$]{\vdash \conclu}
    \infer1[$\bot$]{\one\vdash \conclu}
    \hypo{\pi_0}
    \hypo{\pi_0}
    \infer2[$\otimes$]{\vdash \non}
    \infer1[$\oplus^1$]{\vdash \conclu}
    \infer1[$\bot$]{\one\vdash \conclu}
    \hypo{\pi_S}
    \hypo{\pi_0}
    \infer2[$\otimes$]{\natT\vdash \non}
    \infer1[$\oplus^1$]{\natT\vdash\conclu}
    \infer2[$\with$]{\one\oplus\natT\vdash \conclu}
    \infer1[$\nu$]{\natT \vdash \conclu}
    \infer2[$\with$]{\one\oplus\natT\vdash \conclu}
    \infer1[$\nu$]{\natT \vdash \conclu}
    \infer1[$\bot$]{\one, \natT \vdash \conclu}
    \infer0[\id]{\natT \vdash \natT}
    \hypo{\pi_S}
    \infer2[$\otimes$]{\natT,\natT \vdash \non}
    \infer1[$\oplus^1$]{\natT,\natT \vdash \conclu}
    \infer2[$\with$]{\one\oplus\natT, \natT\vdash \conclu}
    \infer1[$\nu$]{\natT, \natT\vdash \conclu}
    \infer1[$\parr$]{\natT \otimes \natT \vdash \conclu}
  \end{prooftree}}\right.
    \\
    ~
    \\
    \pi_{\omega_2} &~=~& \left\{\scalebox{.6}{\begin{prooftree}
          \infer0[\id]{\non\vdash\non}
          \infer0[$be({\color{red} g} )$]{\non\vdash\natT}
          \infer2[\cut]{\non \vdash \natT}
          \hypo{\pi_S}
          \infer2[\cut]{\natT \otimes \natT \vdash \natT}
          \hypo{\pi_0}
          \infer1[$\bot$]{\one\vdash\natT}
          \infer2[$\with$]{\conclu \vdash \natT}
        \end{prooftree}}\right.
    \\
    ~
    \\
    \pi &~=~& \left\{\scalebox{0.6}{
    \begin{prooftree}
    \infer0[\id]{\non\vdash\non}
    \hypo{\pi_{\omega_1}}
    \infer1{\non\vdash\conclu}
    \infer2[\cut]{\non\vdash \conclu}
    \infer0[\id]{\conclu\vdash\conclu}
    \hypo{\pi_{\omega_2}}
    \infer1{\conclu\vdash\natT}
    \infer2[\cut]{\conclu\vdash\natT}
    \infer0[\id]{\natT\vdash\natT}
    \infer2[\cut]{\conclu\vdash\natT}
    \infer2[\cut]{\non\vdash^{\color{red}g} \natT}
  \end{prooftree}}\right.
  \end{array}
  \]
  \caption{Cantor Pairing in {\mumall}}
  \label{fig:cantor-mumall}
\end{figure}

\paragraph*{Future works.}
We showed a one-way encoding from isos to proofs of {\mumall}, it is
clear that there exist proof-isomorphisms of {\mumall} that do not
correspond to an iso of our language. Therefore, a first extension to
our work would be to relax the different constraints to allow for more
functions to be encoded. For instance, the Cantor Pairing
\cite{chardonnet24fscd} $\natT \iso \natT\otimes \natT$ can be
encoded, by representing $\natT$ with $\mathtt{npos}$, and then $0 =
\fold{(\inl{()})}$ and $S~x = \fold{(\inr{(x)})}$, as follows.

\begin{exa}[Cantor Pairing]
    \[\scalebox{1}{$\omega_1 = \left\{
        \begin{array}{l@{~}c@{~}l}
            \pv{S~i}{j} & {\iso} & \inl{(\pv{i}{S~j})}\\
            \pv{0}{S~S~j} & {\iso} & \inl{(\pv{S~j}{0})} \\
            \pv{0}{S~0} & {\iso} &  \inl{(\pv{0}{0})}\\
            \pv{0}{0} & {\iso} & \inr{()} \\
        \end{array}
        \right\} : \natT\otimes\natT \iso (\natT\otimes\natT) \oplus\one$}\]

    \[\scalebox{1}{$\omega_2 = \left\{
    \begin{array}{l@{~}c@{~}l}
        \inl{(x)} & {\iso} & \letv{y}{g~x}{S~y}\\
        \inr{(())} & {\iso} & 0 \\
    \end{array}
    \right\} : (\natT\otimes\natT) \oplus\one \iso \natT$}\]

    \[\scalebox{1}{$\opn{Cantor Pairing} = \fix g. \left\{
        \begin{array}{l@{~}c@{~}l}
            x & {\iso} & \letv{y}{\omega_1~x}{}\\
            &  & \letv{z}{\omega_2~y}{z} \\
        \end{array}
      \right\} : (\natT\otimes\natT) \iso \natT$}
  \]

  The iso $\omega_1$ is computing the predecessor of a point in the
  plane, marching in diagonal, as the Cantor isomorphism
  prescribes. The iso $\omega_2$ increases a counter if the origin has
  not been met. The iso $\opn{Cantor Pairing}$ is then simply repeatedly
  iterating the predecessor function of the plane $\natT\otimes\natT$,
  increasing the counter as it goes. It stops when the origin is met:
  the value of the counter is the desired number.
\end{exa}

Then, consider its associated proofs $\pi, \pi_{\omega_1},
\pi_{\omega_2}$, given in \autoref{fig:cantor-mumall}, of
(respectively) the isos $\opn{Cantor Pairing}$, $\omega_1$ and
$\omega_2$. While the iso has the expected operational semantics, it
is not well-typed as it is not structurally recursive. Its
associated derivation $\pi$ is also not valid: the visible part is not
contained in the infinite branches. However, whenever $\pi$ is cut with
any proof of $\vdash \natT \otimes \natT$ it will reduce to a finite
proof of $\vdash \natT$, computing the desired result. Therefore, we
need to extend the validity criterion of {\mumall}. In order to
properly type this iso, one would require accepting lexicographical
order (or more generally, any well-founded order) on recursive isos
that act as a termination proof. And then, see how such a criterion
would be captured in terms of pre-proof validity. Along with this,
allowing for coinductive statements and non-polarity-preserving types
and terms would allow to express more functions and capture a bigger
subset of {\mumall} isomorphisms. This would require relaxing the
condition on recursive isos, as termination would be no longer
possible to ensure. This is a focus of our forthcoming research.

A second direction for future work is to consider quantum computation,
by extending our language with linear combinations of terms. We plan
to study purely quantum recursive types and generalized quantum loops:
in~\cite{sabry2018symmetric}, lists are the only recursive type that
is captured and recursion is terminating. The logic {\mumall} would
help in providing a finer understanding of termination and
non-termination.

A quantum-oriented version of the language presented in this paper has
already been studied from a denotational semantics point of
view~\cite{these-kostia}, albeit without inductive types and the
proof-theoretical aspects have not been studied yet.

\bibliographystyle{alphaurl}
\bibliography{bibli.bib}

\end{document}